\documentclass[aps,prx,superscriptaddress,notitlepage, amsfonts,longbibliography, twocolumn]{revtex4-2}

\usepackage[dvipdfmx]{graphicx}
\usepackage{amsmath,amssymb,amsthm,mathtools,mathrsfs}
\usepackage{srcltx}
\usepackage{bm,braket,comment}
\usepackage{amsmath, amssymb,amsthm,mathtools}
\usepackage{xcolor}
\usepackage{longtable}
\usepackage{enumerate}
\usepackage{ulem}
\usepackage{multirow}
\usepackage{here}
\usepackage{comment}
\usepackage{url}

\newtheorem{theorem}{Theorem}
\newtheorem{definition}{Definition}
\newtheorem{lemma}{Lemma}

\newtheorem{proposition}{Proposition}

\newcommand{\bout}[1]

\usepackage{hyperref}
\hypersetup{colorlinks=true,linkcolor=blue,citecolor=blue,urlcolor=blue}

\allowdisplaybreaks[4]
\begin{document}

\title{Clifford Group and Unitary Designs under Symmetry}

\author{Yosuke Mitsuhashi}
\email{mitsuhashi@noneq.t.u-tokyo.ac.jp}
\affiliation{Department of Applied Physics, University of Tokyo, 7-3-1 Hongo, Bunkyo-ku, Tokyo 113-8656, Japan}

\author{Nobuyuki Yoshioka}
\email{nyoshioka@ap.t.u-tokyo.ac.jp}
\affiliation{Department of Applied Physics, University of Tokyo, 7-3-1 Hongo, Bunkyo-ku, Tokyo 113-8656, Japan}
\affiliation{Theoretical Quantum Physics Laboratory, RIKEN Cluster for Pioneering Research (CPR), Wako-shi, Saitama 351-0198, Japan}
\affiliation{JST, PRESTO, 4-1-8 Honcho, Kawaguchi, Saitama, 332-0012, Japan}

\begin{abstract}
We have generalized the well-known statement that the Clifford group is a unitary 3-design into symmetric cases by extending the notion of unitary design. Concretely, we have proven that a symmetric Clifford group is a symmetric unitary 3-design if and only if the symmetry constraint is described by some Pauli subgroup. 
We have also found a complete and unique construction method of symmetric Clifford groups with simple quantum gates for Pauli symmetries. 
For the overall understanding, we have also considered physically relevant
U(1) and SU(2) symmetry constraints, which cannot be described by a Pauli subgroup, and have proven that the symmetric Clifford group is a symmetric unitary 1-design but not a 2-design under those symmetries. 
Our findings are numerically verified by computing the frame potentials, which measure the difference in randomness between the uniform ensemble on the symmetric group of interest and the symmetric unitary group.
This work will open a new perspective into quantum information processing such as randomized benchmarking, and give a deep understanding to many-body systems such as monitored random circuits.
\end{abstract}
\maketitle

\let\oldaddcontentsline\addcontentsline% Store \addcontentsline
\renewcommand{\addcontentsline}[3]{}% Make \addcontentsline a no-op

\section{Introduction}

Randomness in quantum systems is a ubiquitous concept that underpins the core of quantum information processing and quantum many-body systems~\cite{dupuis2014oneshot,roberts2017chaos}. 
The uniform randomness plays central role not only in understanding fundamental phenomena such as thermalization~\cite{popescu2006entanglement, rio2011thermodynamic} and information scrambling~\cite{landsman2019verified, mi2021information}, but also realizing efficient quantum communication~\cite{devetak2004relating,horodecki2005partial} and encryption~\cite{ambainis2009nonmalleable, chau2005unconditionally}. 
To utilize the beautiful and powerful property of the randomness, there have been significant advancements in engineering of approximations of the Haar unitary ensemble, namely the unitary design.  A unitary $t$-design is an ensemble of unitaries that mimics the Haar random unitaries up to the $t$-th moment, and it has proven useful in tasks such as data hiding ~\cite{divincenzo2002quantum}, quantum state discrimination~\cite{ambainis2007quantum,matthews2009distinguishability}, quantum advantage~\cite{bouland2018quantum, boixo2018characterizing, arute2019quantum}, quantum gravity~\cite{roberts2017chaos}, to name a few.

One of the most prominent examples of unitary designs is the Clifford group. Initial interest in the Clifford group was primarily in the context of quantum computing, e.g., the classical simulability~\cite{gottesman1998heisenberg, aaronson2004improved}. However, currently it is known to be applicable to even wider fields such as the quantum state tomography
\cite{huang2020predicting} and hardware verification via randomized benchmarking~\cite{emerson2005scalable, knill2008randomized, magesan2011scalable, nakata2021quantum}.
Although for general qudits, the Clifford group is only a unitary 1-design~\cite{graydon2021clifford}, it elevates to a unitary 2-design if the local Hilbert space dimension is prime~\cite{divincenzo2002quantum,dur2005standard,gross2007evenly}. 
Intriguingly, the multiqubit Clifford group singularly qualifies as a unitary 3-design~\cite{webb2016clifford,zhu2017multiqubit}.

While the concurrent presence of classical simulability and the pseudorandomness of the multiqubit Clifford group has invoked numerous applications to quantum science~\cite{Li2019measurement, Farshi2022mixing, Farshi2023absence, Liu2023model-independent},
we point out that existing studies have focused predominantly on the full ensemble of unitary designs;
our comprehension on realistic scenarios with operational constraints/restrictions remains underdeveloped.
One of the most outstanding questions pertains to the relationship with symmetry, an essential concept responsible for a wealth of phenomena in the natural sciences.
To further explore the physics and quantum information processing under realistic constraints, it is an urgent task to establish how the symmetry impacts the Clifford group.

In this work, we introduce the concept of a symmetric unitary $t$-design and prove that the multiqubit Clifford group under symmetry forms a symmetric unitary 3-design, if and only if the symmetry constraints are essentially characterized by some Pauli subgroup. 
We also propose a complete and unique method for constructing symmetric Clifford operators with elementary quantum gates that operate on a maximum of two qubits.
We subsequently show that, other classes of symmetry stand in stark contrast, as the Clifford groups under these symmetries are merely symmetric unitary 1-designs.
For comprehensive understanding, we have highlighted such a remarkable disparity through practical examples of U(1) and SU(2) symmetries.
Finally, we provide numerical evidence for our findings by computing the frame potentials for the Clifford groups under two representative types of symmetries.

\section{Setup}

We first overview the conventional Clifford group and unitary designs. 
The Clifford group on $N$ qubits is defined as the normalizer of the Pauli group in the unitary group $\mathcal{U}_N$, i.e., $\mathcal{C}_N:=\{U\in\mathcal{U}_N|U\mathcal{P}_N U^\dag=\mathcal{P}_N\}$, where $\mathcal{P}_N:=\{\pm 1, \pm i\}\cdot\{\mathrm{I}, \mathrm{X}, \mathrm{Y}, \mathrm{Z}\}^{\otimes N}$ is the group generated by the Pauli operators $\mathrm{I}$, $\mathrm{X}$, $\mathrm{Y}$ and $\mathrm{Z}$ on each qubit.
It is convenient to introduce the $t$-fold twirling channel to characterize the randomness of a subgroup $\mathcal{X}$ of $\mathcal{U}_N$ as
\begin{align}
    \Phi_{t, \mathcal{X}}(L):= \int_{U\in\mathcal{X}} U^{\otimes t} L U^{\dag \otimes t} d\mu_{\mathcal{X}}(U), \label{eq:twirling_def}
\end{align}
where $L$ is a linear operator acting on $tN$ qubits and $\mu_{\mathcal{X}}$ denotes the normalized Haar measure on $\mathcal{X}$.
We say that the subgroup $\mathcal{X}$ is a unitary $t$-design if 
\begin{align}
    \Phi_{t, \mathcal{X}}=\Phi_{t, \mathcal{U}_N}. 
\end{align}
We note that the definition of unitary designs can be extended for general subsets of the unitary group by considering a distribution on the sets~\cite{scott2008optimizing}, and that our main statement is invariant under the extended definition, as we show in Appendix~\ref{SMsec:weighted_unitary_design}. 
From Eq.~\eqref{eq:twirling_def}, we see that unitary $t$-designs with larger $t$ better approximate the Haar random unitaries, which can be regarded as a unitary $\infty$-design. 
In this regard, it is known that the Clifford group $\mathcal{C}_N$ is a unitary $3$-design but not a $4$-design~\cite{webb2016clifford, zhu2017multiqubit}. 
Note that unitary $t$-designs are always $t'$-designs if $t>t'$, but the contrary does not hold in general.

The symmetric Clifford group and symmetric unitary designs are defined as the symmetric generalizations of the conventional ones. 
In the following, we consider symmetry that can be represented by a subgroup $\mathcal{G}$ of $\mathcal{U}_N$. 
We define the $\mathcal{G}$-symmetric Clifford group on $N$ qubits as the group consisting of the Clifford gates commuting with all the elements in $\mathcal{G}$. 
We can give a rigorous definition as follows:

\begin{definition} \label{def:symmetric_Clifford}
    (Symmetric Clifford group.)
    Let $\mathcal{G}$ be a subgroup of $\mathcal{U}_N$. 
    The $\mathcal{G}$-symmetric Clifford group $\mathcal{C}_{N, \mathcal{G}}$ is defined by 
    \begin{align}
        \mathcal{C}_{N, \mathcal{G}}:=\mathcal{C}_N\cap\mathcal{U}_{N, \mathcal{G}} 
    \end{align}
    with the $\mathcal{G}$-symmetric unitary group 
    \begin{align}
        \mathcal{U}_{N, \mathcal{G}}:=\{U\in\mathcal{U}_N\ |\ \forall G\in\mathcal{G},\ [U, G]=0\}.
    \end{align}
\end{definition}

\noindent
Now it is natural to define for a subgroup $\mathcal{X}$ of $\mathcal{U}_N$ to be a $\mathcal{G}$-symmetric unitary design if the subgroup approximates the $\mathcal{G}$-symmetric unitary group $\mathcal{U}_{N, \mathcal{G}}$. 
The rigorous definition is as follows:

\begin{definition} \label{def:symmetric_unweighted_unitary_design}
    (Symmetric unitary designs.) 
    Let $\mathcal{G}$ and $\mathcal{X}$ be subgroups of $\mathcal{U}_N$. 
    $\mathcal{X}$ is a $\mathcal{G}$-symmetric unitary $t$-design if the $t$-fold twirling channel $\Phi_{t, \mathcal{X}}$ satisfies
    \begin{align}
        \Phi_{t, \mathcal{X}}=\Phi_{t, \mathcal{U}_{N, \mathcal{G}}}.
    \end{align}
\end{definition}

\noindent
Note that in these definitions, the symmetry constraint is described by $\mathcal{U}_{N, \mathcal{G}}$ rather than by $\mathcal{G}$ itself, and the conventional definitions are included as the special case when the symmetry is trivial, i.e., $\mathcal{G} = \{I\}$.

\section{Main results}

Now we are ready to present our two main results. 
The first one is the description of the randomness of symmetric Clifford groups in terms of symmetric unitary designs, which we rigorously present in Theorem~\ref{thm:main}. 
The second one is the complete and unique construction of symmetric Clifford circuits with elementary gates, which we concisely state in Theorem~\ref{thm:expression}.

\subsection{Characterization of pseudorandomness in symmetric Clifford groups}

We prove that the $\mathcal{G}$-symmetric Clifford group $\mathcal{C}_{N, \mathcal{G}}$ is a $\mathcal{G}$-symmetric unitary $3$-design if and only if the symmetry constraint by $\mathcal{G}$ is essentially described by some Pauli subgroup. 
This can be rigorously stated as follows: 
\begin{theorem} \label{thm:main}
    (Randomness of the Clifford group under symmetry.)
	Let $\mathcal{G}$ be a subgroup of $\mathcal{U}_N$. 
	Then, $\mathcal{C}_{N, \mathcal{G}}$ is a $\mathcal{G}$-symmetric unitary $3$-design if and only if $\mathcal{U}_{N, \mathcal{G}}=\mathcal{U}_{N, \mathcal{Q}}$ with some subgroup $\mathcal{Q}$ of $\mathcal{P}_N$. 
\end{theorem}

This theorem provides a guarantee that, under a Pauli symmetry, the symmetric Clifford group maintains its pseudorandomness, which is applicable to various quantum information processing tasks~\cite{divincenzo2002quantum, ambainis2007quantum, matthews2009distinguishability, bouland2018quantum, boixo2018characterizing, arute2019quantum}. 
Moreover, it is remarkable that this theorem also states that the Clifford group maintains the pseudorandomness under symmetry \textit{only if} the symmetry can be characterized by some Pauli subgroup. 
We note that symmetry constraints can be captured by $\mathcal{U}_{N, \mathcal{G}}$ without using $\mathcal{G}$ itself, because when two subgroups $\mathcal{G}$ and $\mathcal{G}'$ of $\mathcal{U}_N$ satisfy $\mathcal{U}_{N, \mathcal{G}}=\mathcal{U}_{N, \mathcal{G}'}$, the $\mathcal{G}$- and $\mathcal{G}'$-symmetric Clifford groups are identical to each other, and moreover the notions of $\mathcal{G}$- and $\mathcal{G}'$-symmetric unitary designs are the same. 
We do not directly present the condition for $\mathcal{G}$ itself, because there are cases when $\mathcal{G}\neq\mathcal{G}'$ but $\mathcal{U}_{N, \mathcal{G}}=\mathcal{U}_{N, \mathcal{G}'}$, for example when $\mathcal{G}=\{\mathrm{I}, \mathrm{Z}\}$ and $\mathcal{G}'=\{e^{i\theta\mathrm{Z}} | \theta\in\mathbb{R}\}$.

We illustrate how we can use this theorem to know whether symmetric Clifford groups are symmetric unitary $3$-designs by taking the following three physically important examples: 
\begin{align}
	&\mathcal{G}=\left\{\mathrm{I}^{\otimes N}, \mathrm{Z}^{\otimes N}\right\}, \label{eq:example1}\\
	&\mathcal{G}=\left\{\left(e^{i\theta\mathrm{Z}}\right)^{\otimes N}\ \middle|\ \theta\in\mathbb{R}\right\}, \label{eq:example2}\\
	&\mathcal{G}=\left\{\left(e^{i(\theta_\mathrm{X}\mathrm{X}+\theta_\mathrm{Y}\mathrm{Y}+\theta_\mathrm{Z}\mathrm{Z})}\right)^{\otimes N}\ \middle|\ \theta_\mathrm{X}, \theta_\mathrm{Y}, \theta_\mathrm{Z}\in\mathbb{R}\right\}. \label{eq:example3}
\end{align}
These groups are isomorphic to $\mathbb{Z}_2$, $\mathrm{U}(1)$ and $\mathrm{SU}(2)$, respectively, which appear ubiquitously in quantum systems; first-principles description of electronic structures, atomic, molecular, and optical physics, quantum spin systems, and lattice gauge theory, to name a few.
When $\mathcal{G}$ is given by Eq.~\eqref{eq:example1}, $\mathcal{G}$-symmetric Clifford group is a $\mathcal{G}$-symmetric unitary $3$-design, because $\mathcal{G}$ itself is a Pauli subgroup. 
In contrast, when $\mathcal{G}$ is given by Eq.~\eqref{eq:example2} or \eqref{eq:example3} with $N\geq 2$, $\mathcal{G}$-symmetric Clifford group is \textit{not} a $\mathcal{G}$-symmetric unitary $3$-design, because in these cases $\mathcal{U}_{N, \mathcal{G}}$ cannot be expressed as $\mathcal{U}_{N, \mathcal{Q}}$ with any Pauli subgroups $\mathcal{Q}$. 
We note that when $\mathcal{G}$ is given by Eq.~\eqref{eq:example2} or \eqref{eq:example3} with $N=1$, $\mathcal{G}$-symmetric Clifford group is again a $\mathcal{G}$-symmetric unitary $3$-design. 
In fact, $\mathcal{G}$ itself is not a Pauli subgroup, but $\mathcal{U}_{N, \mathcal{G}}$ can be expressed as $\mathcal{U}_{N, \mathcal{Q}}$ with a Pauli subgroup $\mathcal{Q}=\{\mathrm{I}, \mathrm{Z}\}$ or $\mathcal{P}_1$.

We can prove that there is no Pauli subgroup $\mathcal{Q}$ such that $\mathcal{U}_{N, \mathcal{G}}=\mathcal{U}_{N, \mathcal{Q}}$ when $\mathcal{G}$ is given by Eq.~\eqref{eq:example2} or \eqref{eq:example3} with $N\geq 2$ as follows: 
First, we suppose that $\mathcal{U}_{N, \mathcal{G}}=\mathcal{U}_{N, \mathcal{Q}}$ with some Pauli subgroup $\mathcal{Q}$. 
Second, we note that we always have $\mathcal{Q}\subset\mathcal{U}_{N, \mathcal{U}_{N, Q}}$. 
Third, the qubit permutation group $\mathcal{S}$ satisfies $\mathcal{S}\subset\mathcal{U}_{N, \mathcal{G}}$, which implies that $\mathcal{U}_{N, \mathcal{U}_{N, G}}\subset\mathcal{U}_{N, \mathcal{S}}$. 
By these three relations, we get $\mathcal{Q}\subset\mathcal{U}_{N, \mathcal{U}_{N, \mathcal{Q}}}=\mathcal{U}_{N, \mathcal{U}_{N, \mathcal{G}}}\subset\mathcal{U}_{N, \mathcal{S}}$. 
Combined with $\mathcal{Q}\subset\mathcal{P}_N$, this implies that $\mathcal{Q}\subset\{\pm 1, \pm i\}\cdot\{\mathrm{I}^{\otimes N}, \mathrm{X}^{\otimes N}, \mathrm{Y}^{\otimes N}, \mathrm{Z}^{\otimes N}\}$. 
Then, we have $\mathrm{X}_{(1)}\mathrm{X}_{(2)}\in\mathcal{U}_{N, \mathcal{Q}}=\mathcal{U}_{N, \mathcal{G}}$, where $\mathrm{X}_{(j)}$ is the Pauli X operator on the $j$th qubit. 
However, this contradicts with Eq.~\eqref{eq:example2} as well as with Eq.~\eqref{eq:example3}.

We emphasize that we can completely characterize the randomness of the examples in terms of unitary designs, i.e., we can clarify the maximal $t$ such that $\mathcal{C}_{N, \mathcal{G}}$ is a $\mathcal{G}$-symmetric unitary $t$-design.
In fact, as expected from the non-symmetric case, we can prove the no-go theorem for $\mathcal{G}$-symmetric unitary 4-designs except for the most constrained case of $\mathcal{U}_{N, \mathcal{G}} = \{e^{i\theta}I|\theta\in\mathbb{R}\}$, which we will describe in Theorem~\ref{thm:4design}; generally we have $t_\mathrm{max} = 3$ for Pauli symmetry. 
On the other hand, when $\mathcal{G}$ is given by Eq.~\eqref{eq:example2} or \eqref{eq:example3} with $N\geq 2$, we get $t_{\rm max}=1$, which we will describe in Theorem~\ref{thm:1design}. 
Note that the single-qubit case is special since we have $t_{\rm max}=3, \infty$ for Eq.~\eqref{eq:example2} and \eqref{eq:example3}, respectively. 
This is because the symmetry constraint can be written by Pauli subgroup $\{\mathrm{I}, \mathrm{Z}\}$ in the case of Eq.~\eqref{eq:example2}, and the $\mathcal{G}$-symmetric Clifford operators are restricted to the identity operator up to phase in the case of Eq.~\eqref{eq:example3}. 
We finally remark that we cannot increase $t_\mathrm{max}$ by considering a nonuniform mixture in the definition of unitary designs, which we show in Appendix~\ref{SMsec:weighted_unitary_design}.

While we guide readers to Appendix~\ref{SMsec:3design} for details on the derivation, it is informative to provide a brief sketch on the proof.
In the proof of the ``if'' part, it is sufficient to show that $\mathcal{C}_{N, \mathcal{Q}}$ is a $\mathcal{Q}$-symmetric unitary $3$-design for all Pauli subgroups $\mathcal{Q}$.
We explicitly construct a map $\mathcal{D}$ with a certain class of symmetric Clifford operators and show that the twirling channels satisfy $\Phi_{3, \mathcal{C}_{N, \mathcal{G}}}=\Phi_{3, \mathcal{U}_{N, \mathcal{G}}}=\mathcal{D}$ by considering the fixed-points of $\Phi_{3, \mathcal{C}_{N, \mathcal{G}}}$ and $\Phi_{3, \mathcal{U}_{N, \mathcal{G}}}$. 
We emphasize that the nontrivial and technical contribution of Theorem~\ref{thm:main} resides in the ``only if" part.
Namely, if $\mathcal{C}_{N, \mathcal{G}}$ is a $\mathcal{G}$-symmetric unitary $3$-design, then there exists a Pauli subgroup $\mathcal{Q}$ such that $\mathcal{U}_{N, \mathcal{G}}=\mathcal{U}_{N, \mathcal{Q}}$. 
Concretely, we construct $\mathcal{Q}$ as the group generated by the set $\mathcal{Q}':=\{Q\in\{\mathrm{I}, \mathrm{X}, \mathrm{Y}, \mathrm{Z}\}^{\otimes N}|\exists G\in\mathcal{G}\ \mathrm{s.t.}\ \mathrm{tr}(GQ)\neq 0\}$, where $\mathcal{A}^{\otimes n}:=\mathcal{A}\otimes\mathcal{A}^{\otimes n-1}$ and $\mathcal{A}\otimes \mathcal{B}:=\{A\otimes B | A\in\mathcal{A}, B\in\mathcal{B}\}$ for general operator sets $\mathcal{A}$ and $\mathcal{B}$. 
The inclusion $\mathcal{U}_{N,\mathcal{G}} \supset \mathcal{U}_{N, \mathcal{Q}}$ directly follows from $\mathrm{span}(\mathcal{G})\subset\mathrm{span}(\mathcal{Q})$, because for any $G\in\mathcal{G}$, every Pauli basis in $G$ with a nonzero coefficient is included in $\mathcal{Q}$. 
However, the proof of the inverse inclusion $\mathcal{U}_{N,\mathcal{G}} \subset \mathcal{U}_{N, \mathcal{Q}}$ requires some technical lemmas (see Appendix~\ref{SMsubsec:3design_only_if}). 
We consider the function $U\mapsto UQU^\dagger$ from  $\mathcal{U}_{N, \mathcal{G}}$ to $\mathcal{U}_N$ for arbitrary taken $Q\in\mathcal{Q}'$ and show that it is a constant function.

\subsection{Construction of symmetric Clifford groups} \label{subsec:construction}

From the viewpoint of algorithms and experiments, it is crucial to give an explicit construction for the symmetric Clifford operators. 
In fact, for a Pauli symmetry, we show that the set of symmetric Clifford operators considered in the proof of Theorem~\ref{thm:main} actually form a complete and unique expression of the symmetric Clifford operators (see Fig.~\ref{fig:circuit} (a)).
We will later discuss the case for non-Pauli symmetry.
This is a symmetric extension of the result in Refs.~\cite{selinger2015generators, bravyi2021hadamard}, where they showed that the standard Clifford operators can be uniquely decomposed by elementary gate sets.

As a preparation for stating the theorem, it is crucial to mention that every Pauli subgroup naturally gives a decomposition into three parts. 
Concretely, we note that any Pauli subgroup $\mathcal{Q}$ can be transformed into the form 
\begin{align}
	\mathcal{R}:=\mathcal{P}_0\{\mathrm{I}, \mathrm{X}, \mathrm{Y}, \mathrm{Z}\}^{\otimes N_1}\otimes\{\mathrm{I}, \mathrm{Z}\}^{\otimes N_2}\otimes\{\mathrm{I}\}^{\otimes N_3} \label{eq:standard_Pauli}
\end{align}
by some Clifford conjugation action up to phase, i.e., $\mathcal{P}_0 W\mathcal{Q}W^\dag=\mathcal{R}$ with some $W\in\mathcal{C}_N$, where $\mathcal{P}_0:=\{\pm 1, \pm i\}$. 
We denote the subsystem of $N_k$ qubits by $\mathrm{A}_k~(k=1, 2, 3)$ and the set of indices representing the qubits in $\mathrm{A}_k$ by $\Gamma_k$. 
We can get $N_1$, $N_2$, $N_3$, and $W$ by considering the following two types of induction processes. 
Let $\mathcal{Q}$ be a Pauli subgroup on $n$ qubits. 
The process is to take a Pauli subgroup $\mathcal{Q}'$ on $n-1$ qubits such that (i) $W'\mathcal{Q}W'^\dag=\{\mathrm{I}, \mathrm{X}, \mathrm{Y}, \mathrm{Z}\}\otimes\mathcal{Q}'$ or (ii) $W'\mathcal{Q}W'^\dag=\{\mathrm{I}, \mathrm{Z}\}\otimes\mathcal{Q}'$ up to phase with some Clifford operator $W'$. 
We can conduct the first type of induction process while $\mathcal{Q}$ has noncommutative pairs of elements, and the second type of process while $\mathcal{Q}\neq\{I\}$ up to phase, as we show in Lemma~\ref{SMlem:Pauli_subgroup_equivalence} in Appendix~\ref{SMsec:technical}. 
$N_1$ and $N_2$ are given as the numbers of the first and the second induction processes, respectively, and $N_3=N-N_1-N_2$. 
We can get the Clifford operator $W$ by taking the product of the Clifford operators $W'$ in all the induction processes. 
By using these notations, we can present the following theorem:

\begin{theorem} \label{thm:expression}
	(Complete and unique construction of the Clifford group under Pauli symmetry.)
    Let $\mathcal{Q}$ be a subgroup of $\mathcal{P}_N$. 
	Then, there exists some $W\in\mathcal{C}_N$ and $\mathcal{R}$ in the form of Eq.~\eqref{eq:standard_Pauli} such that $\mathcal{P}_0 W\mathcal{Q}W^\dag=\mathcal{R}$, and every $\mathcal{Q}$-symmetric Clifford operator $U$ can be uniquely expressed as Fig.~\ref{fig:circuit}(a) as
	\begin{align}
        U=
        &W^\dag\left(\mathrm{T}\prod_{j\in\Gamma_2} \mathrm{C}(P_j)_{(j, \Gamma_3)}\right)V \nonumber\\
		&\ \ \ \ \ \times \left(\prod_{j, k\in\Gamma_2, j<k} \mathrm{CZ}_{(j, k)}^{\nu_{j, k}}\right)\left(\prod_{j\in\Gamma_2} \mathrm{S}_{(j)}^{\mu_j}\right) W \label{eq:expression}
	\end{align}
	with $\mu_j\in\{0 ,1, 2, 3\}$, $\nu_{j, k}\in\{0, 1\}$, $V\in\mathcal{C}_{N_3}$ and $P_j\in\{\mathrm{I}, \mathrm{X}, \mathrm{Y}, \mathrm{Z}\}^{\otimes N_3}$, where $\mathrm{S}_{(j)}$ is the S gate on the $j$th qubit, $\mathrm{CZ}_{(j, k)}$ is the controlled-Z gate on the $j$th and $k$th qubit, $V$ acts on the subsystem $\mathrm{A}_3$, and $\mathrm{C}(P_j)_{(j, \Gamma_3)}$ is the controlled-$P_j$ gates with the $j$th qubit as the control qubit and the qubits in the subsystem $\mathrm{A}_3$ as the target qubits, and $\mathrm{T}\prod$ means the ordered product, i.e., $\mathrm{T}\prod_{j=1}^n O_j:=O_n\cdots O_2 O_1$. 
\end{theorem}

\begin{figure}[t]
\begin{center}
\includegraphics[width=\columnwidth]{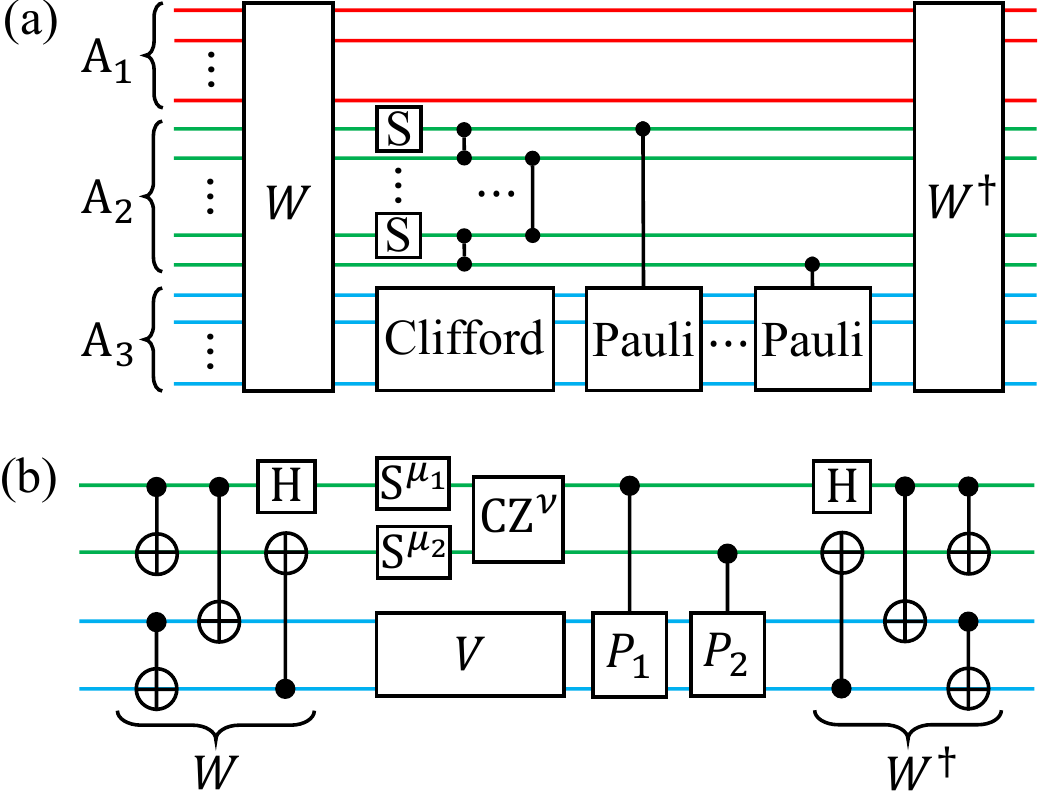}
\caption{
Circuit representation of $\mathcal{Q}$-symmetric Clifford operators presented in Theorem~\ref{thm:expression}. 
(a) Complete and unique expression for $\mathcal{Q}$-symmetric Clifford operators in a general case. 
(b) Example for the symmetry $\mathcal{Q}=\{\mathrm{I}^{\otimes 4}, \mathrm{X}^{\otimes 4}, \mathrm{Y}^{\otimes 4}, \mathrm{Z}^{\otimes 4}\}$. Every $\mathcal{Q}$-symmetric Clifford operator can be uniquely expressed with $\mu_j\in\{0, 1, 2, 3\}$, $\nu\in\{0, 1\}$, $V\in\mathcal{C}_2$ and $P_j\in\{\mathrm{I}, \mathrm{X}, \mathrm{Y}, \mathrm{Z}\}^{\otimes 2}$. 
}
\label{fig:circuit}
\end{center}
\end{figure}

The complete and unique expression by Eq.~\eqref{eq:expression} gives an efficient way to generate all the elements of a $\mathcal{Q}$-symmetric Clifford group. 
In fact, we can understand that it is much more efficient than choosing symmetric elements from the entire Clifford group. 
Namely, the size of the quotient group of the symmetric Clifford group $\mathcal{C}_{N, \mathcal{G}}$ divided by the freedom of phase $\mathcal{U}_0:=\{e^{i\theta}|\theta\in\mathbb{R}\}$ is given by 
\begin{align}
	|\mathcal{C}_{N, \mathcal{G}}/\mathcal{U}_0|
	=&4^{N_2}\cdot 2^{N_2(N_2-1)/2}\cdot |\mathcal{C}_{N_3}/\mathcal{U}_0|\cdot \left(4^{N_3}\right)^{N_2} \nonumber\\
	\sim&2^{2(N_2/2+N_3)^2+3(N_2/2+N_3)}, \label{eq:symm_clifford_group_size}
\end{align}
where we used the fact that there are $4$, $2$, $|\mathcal{C}_{N_3}/\mathcal{U}_0|$ and $4^{N_3}$ choices for each $\mu_j$, $\nu_{j, k}$, $V$, and $P_j$, respectively, and $|\mathcal{C}_{N_3}/\mathcal{U}_0|\sim 2^{2N_3^2+3N_3}$~\cite{Calderbank1998}.
This is much smaller than the size $|\mathcal{C}_N/\mathcal{U}_0|\sim 2^{2N^2+3N}$ of the entire Clifford group. 
The reduction rate is exponential with $N$ in a standard setup where $N_1$ and $N_2$ are $O(1)$~\cite{bravyi2017tapering, setia2020reducing}, which highlights the significance of the explicit construction of symmetric Clifford operators. 
We can see that each qubit in $\mathrm{A}_1$, $\mathrm{A}_2$ and $\mathrm{A}_3$ contributes as $0$, $1/2$ and $1$ qubit in the estimation of the size $|\mathcal{C}_{N, \mathcal{G}}/\mathcal{U}_0|$. 
When we ignore the phase degree of freedom, the size $|\mathcal{C}_{N, \mathcal{G}}/\mathcal{U}_0|$ of the $\mathcal{G}$-symmetric Clifford group on $N$ qubits is almost the same as the size $|\mathcal{C}_{N_2/2+N_3}/\mathcal{U}_0|$ of the entire Clifford group on $N_2/2+N_3$ qubits.

We illustrate the construction of Pauli-symmetric Clifford operators by taking the symmetry $\mathcal{Q}=\{\mathrm{I}^{\otimes 4}, \mathrm{X}^{\otimes 4}, \mathrm{Y}^{\otimes 4}, \mathrm{Z}^{\otimes 4}\}$ on four qubits as an example, which appears as the symmetry of the XYZ Hamiltonian with arbitrary connectivity. 
We know from Theorem~\ref{thm:expression} that every $\mathcal{Q}$-symmetric Clifford operator $U$ can be uniquely expressed as Fig.~\ref{fig:circuit} (b) by noting that $W\mathcal{Q}W^\dag=\{\mathrm{I}_{(1)}, \mathrm{Z}_{(1)}\}\otimes\{\mathrm{I}_{(2)}, \mathrm{Z}_{(2)}\}$ with $W=\mathrm{H}_{(1)}\mathrm{CNOT}_{(4, 2)}\mathrm{CNOT}_{(1, 3)}\mathrm{CNOT}_{(3, 4)}\mathrm{CNOT}_{(1, 2)}$. 
We can confirm that the symmetry constraint greatly reduces the size of the Clifford group by seeing that $|\mathcal{C}_{4, \mathcal{Q}}/\mathcal{U}_0|\sim 10^8$ and $|\mathcal{C}_4/\mathcal{U}_0|\sim 10^{13}$. 
Such a striking difference is displayed in more depth in Fig.~\ref{fig:symm_clifford_group_size}.
Here, we indeed find that the existence of Pauli symmetry leads to exponential reduction of $|\mathcal{C}_{N, \mathcal{G}}/\mathcal{U}_0|$. 
As can be seen from  Eq.~\eqref{eq:symm_clifford_group_size}, we can understand that the entire curve is shifted by $N_1 + N_2/2$ in the asymptotic limit, which gives the advantage of using the construction method presented in Theorem~\ref{thm:expression}.

While we leave the detailed proof of this theorem to Appendix~\ref{SMsec:construction}, we here provide the proof sketch. 
It is sufficient to consider the construction of $\mathcal{C}_{N, \mathcal{R}}$ with a specific class of Pauli subgroups $\mathcal{R}$ given by Eq.~\eqref{eq:standard_Pauli}, because there exists some Clifford conjugation action $W\cdot W^\dag$ that gives one-to-one correspondence from $\mathcal{C}_{N, \mathcal{Q}}$ to $\mathcal{C}_{N, \mathcal{R}}$ for general Pauli subgroups $\mathcal{Q}$. 
In the proof of the completeness, the key is to take the Heisenberg picture, i.e., to see how the conjugation action of a unitary operator transforms Pauli operators. 
For arbitrary $U\in\mathcal{C}_{N, \mathcal{R}}$, $U$ satisfies $U\mathrm{Z}_{(j)} U^\dag=\mathrm{Z}_{(j)}$ and $U\mathrm{X}_{(j)} U^\dag=\mathrm{X}_{(j)}$ for all $j\in\Gamma_1$, and $U\mathrm{Z}_{(j)} U^\dag=\mathrm{Z}_{(j)}$ for all $j\in\Gamma_2$. 
We can inductively construct $U'$ such that $U'U$ is in the form of Eq.~\eqref{eq:expression}, and $U'\mathrm{Z}_{(j)} U'^\dag=\mathrm{Z}_{(j)}$ and $U'\mathrm{X}_{(j)} U'^\dag=\mathrm{X}_{(j)}$ for all $j\in\Gamma_1$ and $j\in\Gamma_2$. 
This implies that $U'$ is a Clifford operator acting nontrivially only on the subsystem $\mathrm{A}_3$, and thus $U'$ is in the form of Eq.~\eqref{eq:expression}. 
Since $U$ can be written as $U=U'^\dag(U'U)$, and both $U'$ and $U'U$ is in the form of Eq.~\eqref{eq:expression}, we know that $U$ is in the form of Eq.~\eqref{eq:expression}. 
We prove the uniqueness by the proof by contradiction. 
Namely, we take arbitrary $U \in \mathcal{C}_{N, \mathcal{R}}$ and suppose that there are two different sets of $(\mu_j, \nu_{j, k}, V, P_{j})$ that realize $U$, and show that they must coincide with each other.

\begin{figure}[t]
\begin{center}
\includegraphics[width=1.0\columnwidth]{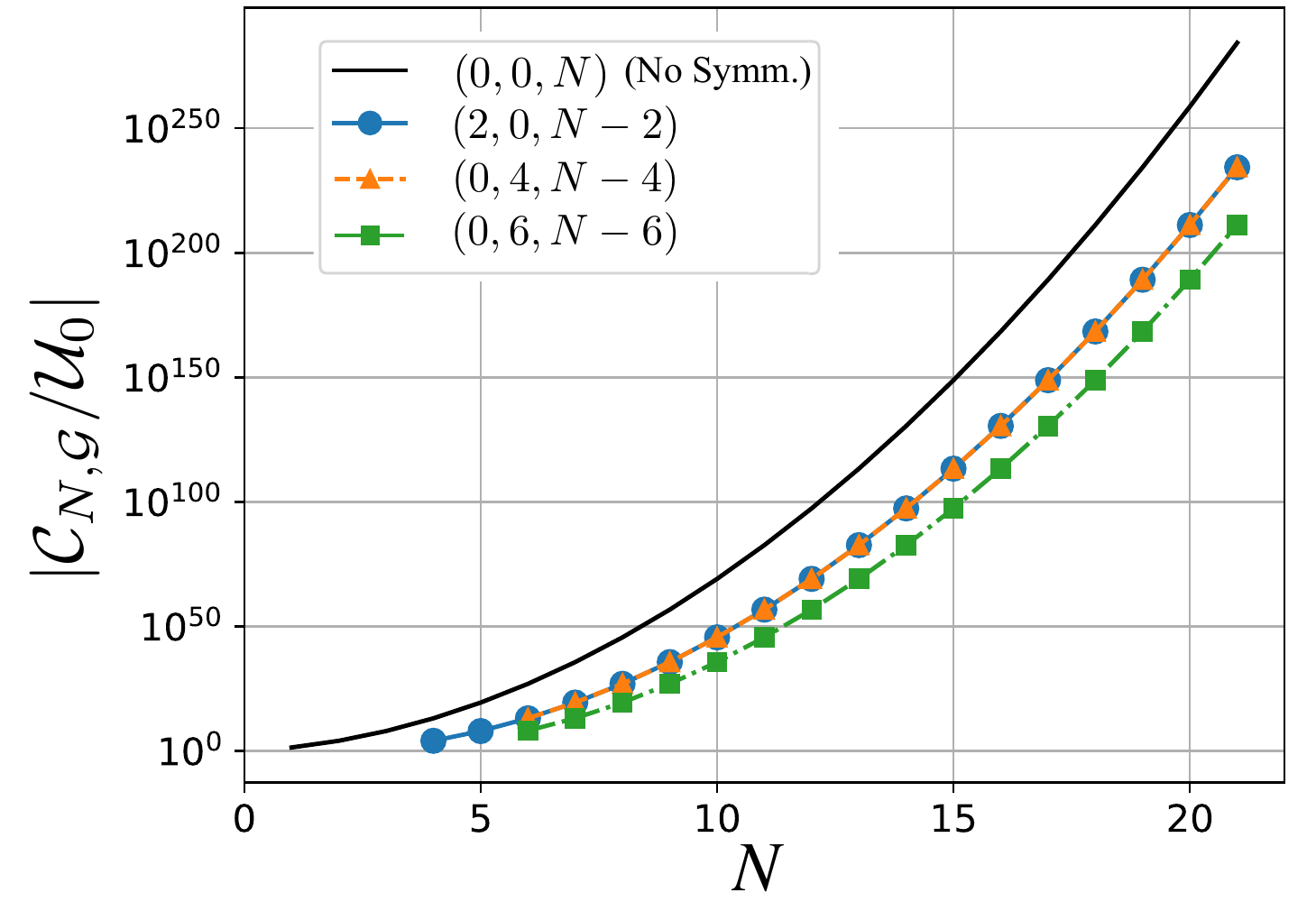}
\caption{
Size of symmetric Clifford groups $|\mathcal{C}_{N, \mathcal{G}} / \mathcal{U}_0|$ under various Pauli symmetries. Here, we show the scaling for symmetries such that the numbers of qubits in A$_1$, A$_2$, A$_3$ are given as $(N_1, N_2, N_3) = (2, 0, N-2), (0, 4, N-4),$ and $(0, 6, N-6)$. 
}
\label{fig:symm_clifford_group_size}
\end{center}
\end{figure}

\section{U(1) and SU(2)-symmetric Clifford groups}

As prominent examples of non-Pauli symmetries, we clarify the property of the $\mathcal{G}$-symmetric Clifford group when $\mathcal{G}$ is given by Eq.~\eqref{eq:example2} or \eqref{eq:example3} on multiple qubits. 
Concretely, in these cases, the $\mathcal{G}$-symmetric Clifford group $\mathcal{C}_{N, \mathcal{G}}$ is a $\mathcal{G}$-symmetric unitary $1$-design, but not a $2$-design. 
This property characterizes the randomness of the symmetric Clifford group under U(1) and SU(2) symmetries given by Eqs.~\eqref{eq:example2} and \eqref{eq:example3}. 
We rigorously present this statement as a theorem.

\begin{theorem} \label{thm:1design}
    (Randomness of $\mathrm{U}(1)$ and $\mathrm{SU}(2)$-symmetric Clifford groups.)
	Let $N\geq 2$ and $\mathcal{G}$ be given by Eq.~\eqref{eq:example2} or \eqref{eq:example3}. 
	Then, $\mathcal{C}_{N, \mathcal{G}}$ is a $\mathcal{G}$-symmetric unitary $1$-design but not a $2$-design. 
\end{theorem}

Let us remark that, as for $2$-designs, we can actually show no-go theorems in a more general class of symmetries, which are described as a tensor product of representations of a nontrivial connected Lie subgroup of a unitary group. 
This type of symmetry represents the conservation of the total observables on the system. 
In the proof of $1$-designs and the disproof of $2$-designs, we use the proof idea of the ``if'' part and the ``only if'' part in the proof of Theorem~\ref{thm:main}, respectively.

\begin{figure}[tb]
\begin{center}
\includegraphics[width=\columnwidth]{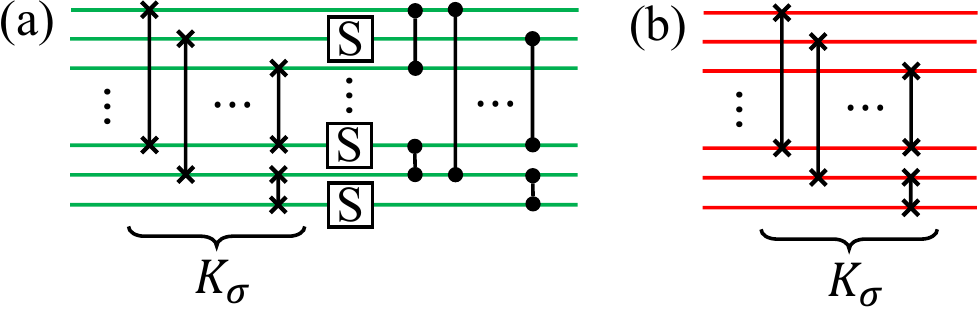}
\caption{
Circuit representations of (a) U(1)-symmetric and (b) SU(2)-symmetric Clifford operators. 
}
\label{fig:nonPauli_circuit}
\end{center}
\end{figure}

By using the result and the proof idea of Theorem~\ref{thm:expression}, we have also found a complete and unique expression for the $\mathcal{G}$-symmetric Clifford operators when the symmetry is given by Eq.~\eqref{eq:example2} or \eqref{eq:example3}. 
Concretely, in the case of Eq.~\eqref{eq:example2}, every $\mathcal{G}$-symmetric Clifford operator $U$ is uniquely expressed as Fig.~\ref{fig:nonPauli_circuit} (a) as
\begin{align}\label{eq:u1_symmetric_clifford}
    U=c\left(\prod_{1\leq j<k\leq N} \mathrm{CZ}_{(j, k)}^{\nu_{j, k}}\right)\left(\prod_{j=1}^N \mathrm{S}_{(j)}^{\mu_j}\right)K_\sigma
\end{align}
with $\mu_j\in\{0, 1, 2, 3\}$, $\nu_{j, k}\in\{0, 1\}$, $\sigma\in\mathfrak{S}_N$ and $c\in\mathcal{U}_0$, where $K_\sigma$ is the permutation operator that brings the $j$th qubit to the $\sigma(j)$th qubit. 
It follows that the size of the quotient group of the symmetric Clifford group $\mathcal{C}_{N, \mathcal{G}}$ divided by the freedom of phase is
\begin{align}
    |\mathcal{C}_{N, \mathcal{G}}/\mathcal{U}_0| =& 2^{N(N-1)/2}\cdot 4^N \cdot N! \nonumber\\
    \sim& 2^{2(N/2)^2+ 3(N/2)+(N+1/2)\log_2(N/e)}.
\end{align}
In the case of Eq.~\eqref{eq:example3}, the $\mathcal{G}$-symmetric Clifford operators are restricted to $c K_\sigma$ with $c\in\mathcal{U}_0$ and $\sigma\in\mathfrak{S}_N$ as expressed in Fig.~\ref{fig:nonPauli_circuit} (b).
The size of $\mathcal{U}_{N, \mathcal{G}}/\mathcal{U}_0$ is $N!$. 
See Theorem~\ref{SMthm:nonPauli_Clifford_expression} in Appendix~\ref{SMsec:construction} for details.

\section{Unitary 4-designs} \label{sec:4design}

We can show that the $\mathcal{G}$-symmetric Clifford group $\mathcal{C}_{N, \mathcal{G}}$ is not a $\mathcal{G}$-symmetric unitary $4$-design except for the trivial case when the $\mathcal{G}$-symmetric unitary subgroup $\mathcal{U}_{N, \mathcal{G}}$ has only scalar multiples of $I$.

\begin{theorem} \label{thm:4design}
    (No-go theorem for symmetric unitary $4$-designs.) 
    Let $\mathcal{G}$ be a subgroup of $\mathcal{U}_N$. 
	Then, $\mathcal{C}_{N, \mathcal{G}}$ is a $\mathcal{G}$-symmetric unitary $4$-design if and only if $\mathcal{U}_{N, \mathcal{G}}=\mathcal{U}_0 I$. 
\end{theorem}

This theorem and Theorem~\ref{thm:main} imply that under a nontrivial Pauli symmetry, the symmetric Clifford group is a symmetric unitary $3$-design but not a $4$-design. 
We can prove this theorem by using the proof idea used in the ``only if'' part of Theorem~\ref{thm:main} as follows:

\begin{figure}[t]
    \begin{center}
        \resizebox{0.98\hsize}{!}{\includegraphics{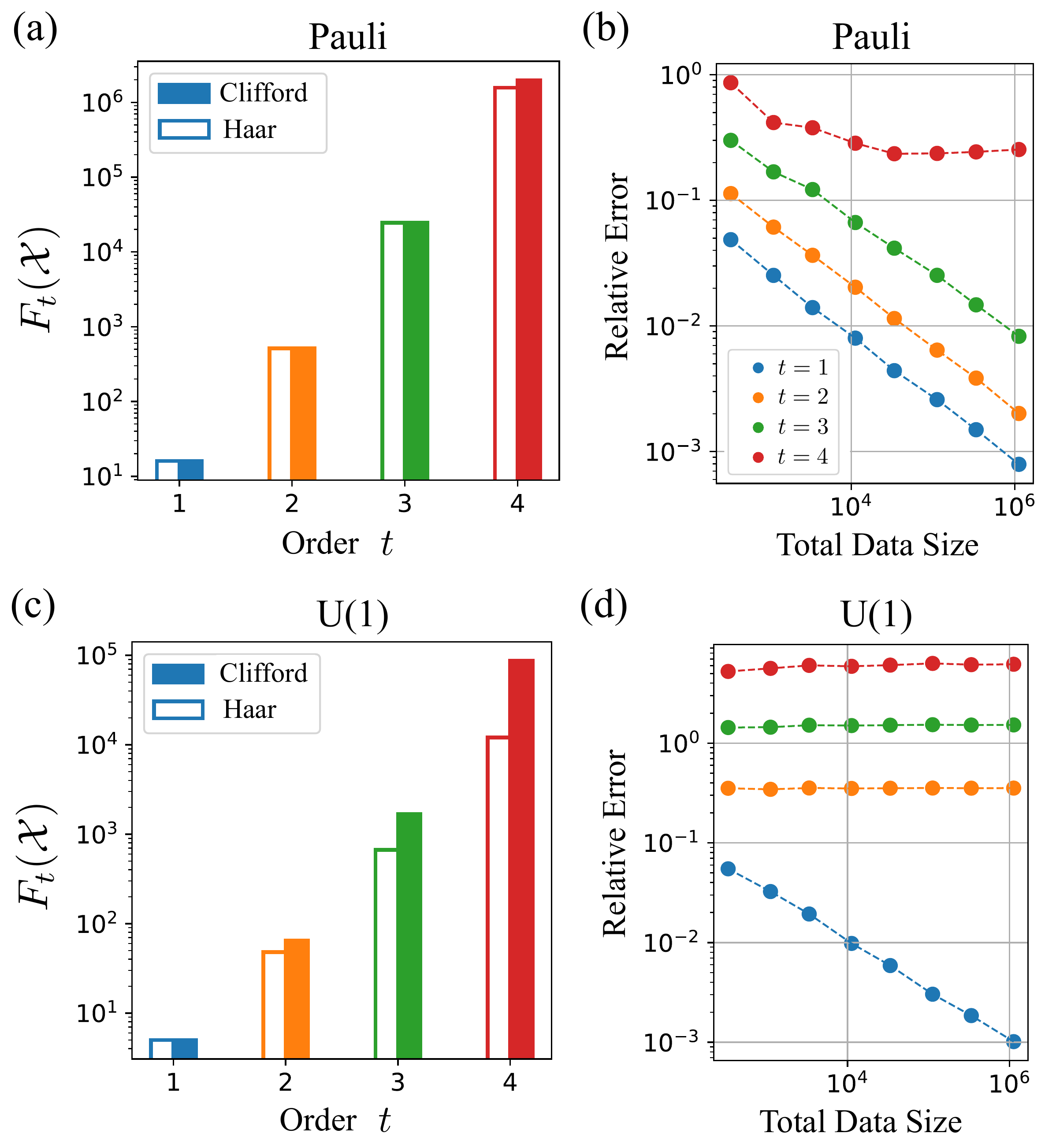}}
        \caption{
        Frame potentials $F_t(\mathcal{X})$ of $\mathcal{G}$-symmetric groups $\mathcal{X}=\mathcal{C}_{N, \mathcal{G}}, \mathcal{U}_{N, \mathcal{G}}$ computed by numerically taking the average over randomly generated unitaries. 
        Here we compare the results for the Pauli symmetry $\mathcal{R}$ given by Eq.~\eqref{eq:standard_Pauli} with $N_1=1$, $N_2=2$, $N_3=3$ and the U(1) symmetry given by Eq.~\eqref{eq:example2} with $N=4$. 
        (a) Frame potentials $F_t(\mathcal{X})$ for the case of  the Pauli symmetry computed by taking the average over $10^6$ samples. 
        (b) Size scaling of the relative error of $F_t(\mathcal{C}_{N, \mathcal{G}})$ against the theoretical lower bound $F_t(\mathcal{U}_{N, \mathcal{G}})$. 
        Here, we independently generate $M$ samples for $U$ and $U'$ respectively, and compute the mean value of $|\mathrm{tr}(UU'^\dag)|^{2t}$. 
        As we increase the total data size $M^2$, the errors become smaller for $t \leq 3$, while they remain finite for $t= 4$. 
        Panels (c) and (d) show the results for the case of U(1) symmetry on a $4$-qubit system. 
        Here the $\mathcal{G}$-symmetric Clifford group is only a $\mathcal{G}$-symmetric unitary $1$-design, and not a $2$-design. 
        The numerical simulation is performed using the library Qulacs~\cite{suzuki2021qulacs}.
        }
        \label{fig:frame_potential}
    \end{center}
\end{figure}

\begin{proof}
	Since the ``if'' part is trivial, it is sufficient to prove the ``only if'' part. 
	Suppose that $\mathcal{C}_{N, \mathcal{G}}$ is a $\mathcal{G}$-symmetric unitary $4$-design. 
	We define $L\in\mathcal{L}(\mathcal{H}^{\otimes 4})$ by $L:=\sum_{P\in\mathcal{P}_N^+} P^{\otimes 4}$, where $\mathcal{P}_N^+:=\{\mathrm{I}, \mathrm{X}, \mathrm{Y}, \mathrm{Z}\}^{\otimes N}$. 
	We take arbitrary $U\in\mathcal{C}_{N, \mathcal{G}}$. 
    By Lemma~\ref{SMlem:Clifford_action_on_Pauli} in Appendix~\ref{SMsec:technical}, we can take a function $s_U:\mathcal{P}_N^+\to\{\pm 1\}$ and a bijection $h_U$ on $\mathcal{P}_N^+$ such that $UPU^\dag=s_U(P) h_U(P)$ for all $P\in\mathcal{P}_N^+$. 
    We note that $s_U$ and $h_U$ are dependent on $U$. 
	By using the definitions of $L$, $s_U$ and $h_U$, we get 
	\begin{align}
		U^{\otimes 4}LU^{\dag\otimes 4}
		=&\sum_{P\in\mathcal{P}_N^+} (UPU^\dag)^{\otimes 4} \nonumber\\
		=&\sum_{P\in\mathcal{P}_N^+} (s_U(P) h_U(P))^{\otimes 4} \nonumber\\
		=&\sum_{P\in\mathcal{P}_N^+} s_U(P)^4 h_U(P)^{\otimes 4} \nonumber\\
		=&\sum_{P\in\mathcal{P}_N^+} h_U(P)^{\otimes 4} \nonumber\\
		=&\sum_{P\in\mathcal{P}_N^+} P^{\otimes 4} \nonumber\\
		=&L. 
	\end{align}
	Since this holds for all $U\in\mathcal{C}_{N, \mathcal{G}}$ and $\mathcal{C}_{N, \mathcal{G}}$ is a $\mathcal{G}$-symmetric unitary $4$-design, by Lemma~\ref{SMlem:unitary_design_commutativity} in Appendix~\ref{SMsec:3design}, we have $U^{\otimes 4}LU^{\dag\otimes 4}=L$ for all $U\in\mathcal{U}_{N, \mathcal{G}}$. 
	We therefore get $UPU^\dag=P$ for all $U\in\mathcal{U}_{N, \mathcal{G}}$ and $P\in\mathcal{P}_N^+$ by Lemma~\ref{SMlem:3_design_discreteness} in Appendix~\ref{SMsec:3design}. 
	This implies that any $U\in\mathcal{U}_{N, \mathcal{G}}$ satisfies $PUP=U$ for all $P\in\mathcal{P}_N^+$, which is equivalent to $U=e^{i\theta}I$ with some $\theta\in\mathbb{R}$. 
	This means that $\mathcal{U}_{N, \mathcal{G}}\subset\mathcal{U}_0 I$. 
	Since $\mathcal{U}_{N, \mathcal{G}}\supset\mathcal{U}_0 I$ always holds, we get $\mathcal{U}_{N, \mathcal{G}}=\mathcal{U}_0 I$. 
\end{proof}

\section{Verification via frame potentials}

We can give a numerical evidence to Theorems~\ref{thm:main} and \ref{thm:expression} by computing the frame potentials, which are defined for a subgroup $\mathcal{X}$ of $\mathcal{U}_{N, \mathcal{G}}$ as~\cite{roberts2017chaos}
\begin{align}
	F_t(\mathcal{X}):=\int_{U, U'\in\mathcal{X}} |\mathrm{tr}(UU'^\dag)|^{2t} d\mu_{\mathcal{X}}(U)d\mu_{\mathcal{X}}(U'), \label{eq:framepot}
\end{align}
where the integral is replaced by a summation if  $\mathcal{X}$ is a finite set up to phase. 
Similarly to the conventional case~\cite{roberts2017chaos}, we can measure the distance between the $t$-fold twirling channels $\Phi_{t, \mathcal{X}}$ and $\Phi_{t, \mathcal{U}_{N, \mathcal{G}}}$ by the frame potentials. 
It is therefore straightforward to show that
$F_t(\mathcal{X}) \geq F_t(\mathcal{U}_{N, \mathcal{G}})$, where the equality holds if and only if $\mathcal{X}$ is a $\mathcal{G}$-symmetric unitary $t$-design.

Figure~\ref{fig:frame_potential} (a)(b) clearly shows that the Clifford group $\mathcal{C}_{N, \mathcal{G}}$ under Pauli symmetry is a $\mathcal{G}$-symmetric unitary 3-design but not a 4-design. 
In sharp contrast, Fig.~\ref{fig:frame_potential} (c)(d) shows that, when the Clifford group is constrained by $\mathcal{G} = \{(e^{i \theta Z})^{\otimes N} | \theta\in \mathbb{R}\}$, which is isomorphic to U(1), then the $\mathcal{G}$-symmetric Clifford group $\mathcal{C}_{N, \mathcal{G}}$ is only a $\mathcal{G}$-symmetric unitary 1-design and not a $2$-design.

We remark that, in order to compute frame potentials numerically, we have sampled symmetric Clifford operators uniformly from the entire $\mathcal{C}_{N, \mathcal{G}}$. 
For instance, in the case of Pauli-symmetric case,
we have utilized the complete and unique construction provided in Theorem~\ref{thm:expression} (or Fig.~\ref{fig:circuit});
we have randomly chosen the parameters $\mu_j\in\{0,1,2,3\}, \nu_{j,k}\in\{0,1\}, P_j \in \{\mathrm{I}, \mathrm{X}, \mathrm{Y}, \mathrm{Z}\}^{\otimes N_3}$ with a uniform probability, and have also employed the method in Ref.~\cite{bravyi2021hadamard} in order to uniformly choose Clifford operators $V\in \mathcal{C}_{N_3}$ acting on the subsystem A$_3$.
This can also be done in a similar way for U(1) and SU(2)-symmetric cases as well based on Fig.~\ref{fig:nonPauli_circuit} and  Eq.~\eqref{eq:u1_symmetric_clifford}.

It is beneficial to mention that the frame potentials of the symmetric unitary groups can be written with those of the unitary groups of several dimensions. 
In order to state the result, we explain the irreducible decomposition of group representations~\cite{bartlett2007reference}. 
We consider the regular representation $\rho$ of the group $\mathcal{G}$, i.e., $\rho(G):=G$ for all $G\in\mathcal{G}$. 
Since $\rho$ is a unitary representation, $\rho$ is completely reducible, and thus there exist some set of pairs $\{(\mathcal{I}_\lambda, \mathcal{J}_\lambda)\}$ of spaces such that the Hilbert space $\mathcal{H}$ of the $N$ qubits is decomposed into 
\begin{align}
    \mathcal{H}=\bigoplus_\lambda \mathcal{I}_\lambda\otimes\mathcal{J}_\lambda
\end{align}
and 
\begin{align}
    \rho(G)=\bigoplus_\lambda \rho_\lambda(G)\otimes I \label{eq:representation_decomposition}
\end{align}
with inequivalent irreducible representations $\rho_\lambda$ of $\mathcal{G}$ on $\mathcal{I}_\lambda$ and the identity operator $I$ on $\mathcal{J}_\lambda$. 
This implies that $\lambda$ is the index for inequivalent irreducible representations, $\mathcal{I}_\lambda$ is the representation space, and $\mathcal{J}_\lambda$ is the multiplicity space.

\begin{theorem} \label{thm:frame_potential}
    (Formula for frame potentials of symmetric unitary groups.) 
    Let $\mathcal{G}$ be a subgroup of $\mathcal{U}_N$ and the regular representation $\rho$ of $\mathcal{G}$ be irreducibly decomposed in the form of Eq.~\eqref{eq:representation_decomposition}. 
    Then, the frame potential $F_t(\mathcal{U}_{N, \mathcal{G}})$ of the $\mathcal{G}$-symmetric unitary group $\mathcal{U}_{N, \mathcal{G}}$ is given by 
    \begin{align}
        F_t(\mathcal{U}_{N, \mathcal{G}})=(t!)^2\sum_{(t_\lambda)\in S_t} \prod_\lambda \frac{\mathrm{dim}(\mathcal{\mathcal{I}_\lambda})^{2t_\lambda}}{(t_\lambda!)^2}F_{t_\lambda}(\mathcal{U}(\mathcal{J}_\lambda)) \label{eq:symmetric_frame_potential}
    \end{align}
    with $S_t:=\{(t_\lambda)|\sum_\lambda t_\lambda=t,\ t_\lambda\in\mathbb{Z},\ t_\lambda\geq 0\}$ and the unitary group $\mathcal{U}(\mathcal{J}_\lambda)$ on $\mathcal{J}_\lambda$. 
\end{theorem}

We illustrate the result of Theorem~\ref{thm:frame_potential} in the case of the Pauli symmetries $\mathcal{Q}$, which are unitarily equivalent to the Pauli subgroup $\mathcal{R}$ given by Eq.~\eqref{eq:standard_Pauli}, as we have explained in Sec.~\ref{subsec:construction}. 
In this case, $\lambda$ is the index for specifying the sequence of the eigenvalues of $\mathrm{Z}_j$ for $j\in\Gamma_2$. We take $\mathcal{I}_\lambda$ as the Hilbert space of the subsystems $\mathrm{A}_1$ and $\mathrm{A}_2$ which is the simultaneous eigenspace of $(\mathrm{Z}_j)_{j\in\Gamma_2}$ with the eigenvalues specified by $\lambda$, and also take $\mathcal{J}_\lambda$ as the Hilbert space of the subsystem $\mathrm{A}_3$. 
This means that there are $2^{N_2}$ choices for $\lambda$, and for each of them, the dimensions of $\mathcal{I}_\lambda$ and $\mathcal{J}_\lambda$ are $2^{N_1}$ and $2^{N_3}$, respectively. When $t\leq 2^{N_3}$, Eq.~\eqref{eq:symmetric_frame_potential} gives the following simple form: 
\begin{align}
    F_t(\mathcal{U}_{N, \mathcal{Q}})
    =&(t!)^2\sum_{(t_\lambda)\in S_t} \prod_\lambda \frac{\left(2^{N_1}\right)^{2t_\lambda}}{(t_\lambda!)^2}\cdot t_\lambda! \nonumber\\
    =&2^{2N_1 t}\cdot t!\sum_{(t_\lambda)\in S_t}\frac{t!}{\prod_\lambda t_\lambda!} \nonumber\\
    =&2^{2N_1 t}\cdot t!\cdot \left(2^{N_2}\right)^t \nonumber\\
    =&2^{(2N_1+N_2)t}\cdot t!, 
\end{align}
where we used $F_{t_\lambda}(\mathcal{U}(\mathcal{J}_\lambda))=t_\lambda!$ in the first line, and the multinomial theorem in the third line.

\begin{proof}
    We are going to prove Eq.~\eqref{eq:symmetric_frame_potential} by considering the generating function of $F_t(\mathcal{U}_{N, \mathcal{G}})$. 
    For a unitary subgroup $\mathcal{X}$, we define $f_\mathcal{\mathcal{X}}$ by 
    \begin{align}
        f_\mathcal{\mathcal{X}}(z):=\int_{U\in\mathcal{X}} \mathrm{det}\left(e^{z(U+U^\dag)}\right) d\mu_\mathcal{X}(U)\ \forall z\in\mathbb{C}. \label{eq:thm:frame_potential1}
    \end{align}
    Then, $f_\mathcal{\mathcal{X}}$ satisfies 
    \begin{align}
        &f_\mathcal{\mathcal{X}}(z) \nonumber\\
        =&\int_{U\in\mathcal{X}} e^{\mathrm{tr}(z(U+U^\dag))} d\mu_\mathcal{X}(U) \nonumber\\
        =&\int_{U\in\mathcal{X}} \left(\sum_{t=0}^\infty \frac{z^t}{t!}(\mathrm{tr}(U))^t\right)\left(\sum_{t'=0}^\infty \frac{z^{t'}}{t'!}(\mathrm{tr}(U^\dag))^{t'}\right) d\mu_\mathcal{X}(U) \nonumber\\
        =&\sum_{t=0}^\infty \sum_{t'=0}^\infty \frac{z^{t+t'}}{t!t'!}\int_{U\in\mathcal{X}} (\mathrm{tr}(U))^t (\mathrm{tr}(U^\dag))^{t'} d\mu_\mathcal{X}(U) \nonumber\\
        =&\sum_{t=0}^\infty \sum_{t'=0}^\infty \frac{z^{t+t'}}{t!t'!}\delta_{t, t'}F_t(\mathcal{X}) \nonumber\\
        =&\sum_{t=0}^\infty \frac{z^{2t}}{(t!)^2}F_t(\mathcal{X}), \label{eq:thm:frame_potential2}
    \end{align}
    where we used 
    \begin{align}
        &\int_{U\in\mathcal{X}} (\mathrm{tr}(U))^t (\mathrm{tr}(U^\dag))^{t'} d\mu_\mathcal{X}(U) \nonumber\\
        =&\frac{1}{2\pi}\int_0^{2\pi} \left(\int_{U\in\mathcal{X}} (\mathrm{tr}(e^{i\theta}U))^t (\mathrm{tr}((e^{i\theta}U)^\dag))^{t'} d\mu_\mathcal{X}(U)\right)d\theta \nonumber\\
        =&\left(\frac{1}{2\pi}\int_0^{2\pi} e^{i(t-t')\theta} d\theta\right) \nonumber\\
        &\ \ \ \ \ \times\left(\int_{U\in\mathcal{X}} \left(\mathrm{tr}(U)\right)^t\left(\mathrm{tr}(U^\dag)\right)^{t'} d\mu_\mathcal{X}(U)\right) \nonumber\\
        =&\delta_{t, t'}F_t(\mathcal{X}). \label{eq:thm:frame_potential3}
    \end{align}
    in the second to last line, and we note that the definition of the frame potential (Eq.~\eqref{eq:framepot}) is equivalent to $F_t(\mathcal{X})=\int_{U\in\mathcal{X}} |\mathrm{tr}(U)|^{2t} d\mu_\mathcal{X}(U)$ by the left invariance of the Haar measure. 
    By Eq.~(2.26) of Ref.~\cite{bartlett2007reference}, every $U\in\mathcal{U}_{N, \mathcal{G}}$ can be written as 
    \begin{align}
        U=\bigoplus_\lambda I\otimes U_\lambda \label{eq:thm:frame_potential4}
    \end{align}
    with the identity operator $I$ on $\mathcal{I}_\lambda$ and some $U_\lambda\in\mathcal{U}(\mathcal{J}_\lambda)$, and we have 
    \begin{align}
        \mathrm{det}\left(e^{z(U+U^\dag)}\right)
        =&\mathrm{det}\left(\bigoplus_\lambda\left(I\otimes e^{z(U_\lambda+U_\lambda^\dag)}\right)\right) \nonumber\\
        =&\prod_\lambda \left(\mathrm{det}\left(e^{z(U_\lambda+U_\lambda^\dag)}\right)\right)^{\mathrm{dim}(\mathcal{I}_\lambda)} \nonumber\\
        =&\prod_\lambda \mathrm{det}\left(e^{\mathrm{dim}(\mathcal{I}_\lambda)z(U_\lambda+U_\lambda^\dag)}\right). \label{eq:thm:frame_potential5}
    \end{align}
    Since Eq.~\eqref{eq:thm:frame_potential4} gives one-to-one correspondence between $\mathcal{U}_{N, \mathcal{G}}$ and the set of $\mathcal{U}(\mathcal{J}_\lambda)$, by Eqs.~\eqref{eq:thm:frame_potential1} and \eqref{eq:thm:frame_potential5}, we get  
    \begin{align}
        &f_{\mathcal{U}_{N, \mathcal{G}}}(z) \nonumber\\
        =&\int_{U\in\mathcal{U}_{N, \mathcal{G}}} \mathrm{det}\left(e^{z(U+U^\dag)}\right) d\mu_{\mathcal{U}_{N, \mathcal{G}}}(U) \nonumber\\
        =&\prod_\lambda \int_{U_\lambda\in\mathcal{U}(\mathcal{J}_\lambda)} \mathrm{det}\left(e^{\mathrm{dim}(\mathcal{I}_\lambda)z(U_\lambda+U_\lambda^\dag)}\right) d\mu_{\mathcal{U}(\mathcal{J}_\lambda)}(U_\lambda) \nonumber\\
        =&\prod_\lambda f_{\mathcal{U}(\mathcal{J}_\lambda)}(\mathrm{dim}(\mathcal{I}_\lambda)z). \label{eq:thm:frame_potential6}
    \end{align}
    By using Eq.~\eqref{eq:thm:frame_potential2} and comparing the coefficients of $z^{2t}$ of both sides of Eq.~\eqref{eq:thm:frame_potential6}, we get Eq.~\eqref{eq:symmetric_frame_potential}. 
\end{proof}

\section{Discussion}

In this paper, we have generalized the unitary 3-design property of the multiqubit Clifford group into symmetric cases.
We have rigorously shown that a symmetric Clifford group is a symmetric unitary 3-design if and only if the symmetry is given by some Pauli subgroup (Theorem~\ref{thm:main}), and also have provided a way to generate all the elements without redundancy (Theorem~\ref{thm:expression}).
Furthermore, we have also proven that two physically important class of U(1) and SU(2) symmetries, which cannot be reduced to Pauli subgroups, only yields symmetric unitary 1-designs.
Finally, for numerical validation, we have computed the frame potentials by randomly sampling symmetric unitaries, and have confirmed that our findings indeed hold.

We can derive another property of the symmetric Clifford group with respect to locality, from the results about the construction method of the symmetric Clifford operators. 
As we have shown in Theorem~\ref{thm:expression} and after Theorem~\ref{thm:1design}, under the Pauli, U(1), and SU(2) symmetries, all the symmetric Clifford operators can be constructed with $2$-qubit local symmetric Clifford operators. 
It can be seen as a symmetric generalization of the fact that all the Clifford operators can be constructed with $2$-qubit local Clifford operators~\cite{selinger2015generators, bravyi2021hadamard}. 
This stands in contrast to the result in Ref.~\cite{marvian2022restrictions}, which shows that while all the unitary operators can be constructed with local unitary operators, some symmetric unitary operators \textit{cannot} be constructed with local symmetric unitary operators.

We envision two important future directions.
First, it is theoretically crucial to reveal the requirement to achieve symmetric unitary designs in the approximate sense. While it is known for the non-symmetric case that the approximate $t$-designs require only polynomial gate depth with respect to both the qubit count $N$ and order $t$~\cite{brandao2016local, mezher2019efficient}, it is far from trivial whether this is also true for the symmetric case.
Second, it is practically intriguing to develop a constructive way to generate symmetric Clifford circuits under general symmetry $\mathcal{G}$ that cannot be described by some Pauli subgroup.
While we provide such an example for both U(1) and SU(2) symmetry, it is important to construct circuits in an automated way for general situations, in particular when one is interested in performing Clifford gate simulation for the purpose of quantum many-body simulation~\cite{ravi2022cafqa}.\\

\section*{Acknowledgements}

The authors wish to thank Hidetaka Manabe and Takahiro Sagawa for insightful discussions. 
This work was supported by JST Grant Number JPMJPF2221, JST ERATO Grant Number JPMJER2302, and JST CREST Grant Number JPMJCR23I4, Japan. 
Y.M. is supported by World-leading Innovative Graduate Study Program for Materials Research, Information, and Technology (MERIT-WINGS) of the University of Tokyo.
Y.M. is also supported by JSPS KAKENHI Grant No. JP23KJ0421.
N.Y. wishes to thank JST PRESTO No. JPMJPR2119, 
and the support from IBM Quantum.

\bibliography{bib.bib}

\begin{thebibliography}{49}%
\makeatletter
\providecommand \@ifxundefined [1]{%
 \@ifx{#1\undefined}
}%
\providecommand \@ifnum [1]{%
 \ifnum #1\expandafter \@firstoftwo
 \else \expandafter \@secondoftwo
 \fi
}%
\providecommand \@ifx [1]{%
 \ifx #1\expandafter \@firstoftwo
 \else \expandafter \@secondoftwo
 \fi
}%
\providecommand \natexlab [1]{#1}%
\providecommand \enquote  [1]{``#1''}%
\providecommand \bibnamefont  [1]{#1}%
\providecommand \bibfnamefont [1]{#1}%
\providecommand \citenamefont [1]{#1}%
\providecommand \href@noop [0]{\@secondoftwo}%
\providecommand \href [0]{\begingroup \@sanitize@url \@href}%
\providecommand \@href[1]{\@@startlink{#1}\@@href}%
\providecommand \@@href[1]{\endgroup#1\@@endlink}%
\providecommand \@sanitize@url [0]{\catcode `\\12\catcode `\$12\catcode
  `\&12\catcode `\#12\catcode `\^12\catcode `\_12\catcode `\%12\relax}%
\providecommand \@@startlink[1]{}%
\providecommand \@@endlink[0]{}%
\providecommand \url  [0]{\begingroup\@sanitize@url \@url }%
\providecommand \@url [1]{\endgroup\@href {#1}{\urlprefix }}%
\providecommand \urlprefix  [0]{URL }%
\providecommand \Eprint [0]{\href }%
\providecommand \doibase [0]{https://doi.org/}%
\providecommand \selectlanguage [0]{\@gobble}%
\providecommand \bibinfo  [0]{\@secondoftwo}%
\providecommand \bibfield  [0]{\@secondoftwo}%
\providecommand \translation [1]{[#1]}%
\providecommand \BibitemOpen [0]{}%
\providecommand \bibitemStop [0]{}%
\providecommand \bibitemNoStop [0]{.\EOS\space}%
\providecommand \EOS [0]{\spacefactor3000\relax}%
\providecommand \BibitemShut  [1]{\csname bibitem#1\endcsname}%
\let\auto@bib@innerbib\@empty
%</preamble>
\bibitem [{\citenamefont {Dupuis}\ \emph {et~al.}(2014)\citenamefont {Dupuis},
  \citenamefont {Berta}, \citenamefont {Wullschleger},\ and\ \citenamefont
  {Renner}}]{dupuis2014oneshot}%
  \BibitemOpen
  \bibfield  {author} {\bibinfo {author} {\bibfnamefont {F.}~\bibnamefont
  {Dupuis}}, \bibinfo {author} {\bibfnamefont {M.}~\bibnamefont {Berta}},
  \bibinfo {author} {\bibfnamefont {J.}~\bibnamefont {Wullschleger}},\ and\
  \bibinfo {author} {\bibfnamefont {R.}~\bibnamefont {Renner}},\ }\bibfield
  {title} {\bibinfo {title} {One-shot decoupling},\ }\href
  {https://doi.org/10.1007/s00220-014-1990-4} {\bibfield  {journal} {\bibinfo
  {journal} {Communications in Mathematical Physics}\ }\textbf {\bibinfo
  {volume} {328}},\ \bibinfo {pages} {251} (\bibinfo {year}
  {2014})}\BibitemShut {NoStop}%
\bibitem [{\citenamefont {Roberts}\ and\ \citenamefont
  {Yoshida}(2017)}]{roberts2017chaos}%
  \BibitemOpen
  \bibfield  {author} {\bibinfo {author} {\bibfnamefont {D.~A.}\ \bibnamefont
  {Roberts}}\ and\ \bibinfo {author} {\bibfnamefont {B.}~\bibnamefont
  {Yoshida}},\ }\bibfield  {title} {\bibinfo {title} {Chaos and complexity by
  design},\ }\href@noop {} {\bibfield  {journal} {\bibinfo  {journal} {Journal
  of High Energy Physics}\ }\textbf {\bibinfo {volume} {2017}},\ \bibinfo
  {pages} {1} (\bibinfo {year} {2017})}\BibitemShut {NoStop}%
\bibitem [{\citenamefont {Popescu}\ \emph {et~al.}(2006)\citenamefont
  {Popescu}, \citenamefont {Short},\ and\ \citenamefont
  {Winter}}]{popescu2006entanglement}%
  \BibitemOpen
  \bibfield  {author} {\bibinfo {author} {\bibfnamefont {S.}~\bibnamefont
  {Popescu}}, \bibinfo {author} {\bibfnamefont {A.~J.}\ \bibnamefont {Short}},\
  and\ \bibinfo {author} {\bibfnamefont {A.}~\bibnamefont {Winter}},\
  }\bibfield  {title} {\bibinfo {title} {Entanglement and the foundations of
  statistical mechanics},\ }\href {https://doi.org/10.1038/nphys444} {\bibfield
   {journal} {\bibinfo  {journal} {Nature Physics}\ }\textbf {\bibinfo {volume}
  {2}},\ \bibinfo {pages} {754} (\bibinfo {year} {2006})}\BibitemShut {NoStop}%
\bibitem [{\citenamefont {Rio}\ \emph {et~al.}(2011)\citenamefont {Rio},
  \citenamefont {{\AA}berg}, \citenamefont {Renner}, \citenamefont {Dahlsten},\
  and\ \citenamefont {Vedral}}]{rio2011thermodynamic}%
  \BibitemOpen
  \bibfield  {author} {\bibinfo {author} {\bibfnamefont {L.~d.}\ \bibnamefont
  {Rio}}, \bibinfo {author} {\bibfnamefont {J.}~\bibnamefont {{\AA}berg}},
  \bibinfo {author} {\bibfnamefont {R.}~\bibnamefont {Renner}}, \bibinfo
  {author} {\bibfnamefont {O.}~\bibnamefont {Dahlsten}},\ and\ \bibinfo
  {author} {\bibfnamefont {V.}~\bibnamefont {Vedral}},\ }\bibfield  {title}
  {\bibinfo {title} {The thermodynamic meaning of negative entropy},\ }\href
  {https://doi.org/10.1038/nature10123} {\bibfield  {journal} {\bibinfo
  {journal} {Nature}\ }\textbf {\bibinfo {volume} {474}},\ \bibinfo {pages}
  {61} (\bibinfo {year} {2011})}\BibitemShut {NoStop}%
\bibitem [{\citenamefont {Landsman}\ \emph {et~al.}(2019)\citenamefont
  {Landsman}, \citenamefont {Figgatt}, \citenamefont {Schuster}, \citenamefont
  {Linke}, \citenamefont {Yoshida}, \citenamefont {Yao},\ and\ \citenamefont
  {Monroe}}]{landsman2019verified}%
  \BibitemOpen
  \bibfield  {author} {\bibinfo {author} {\bibfnamefont {K.~A.}\ \bibnamefont
  {Landsman}}, \bibinfo {author} {\bibfnamefont {C.}~\bibnamefont {Figgatt}},
  \bibinfo {author} {\bibfnamefont {T.}~\bibnamefont {Schuster}}, \bibinfo
  {author} {\bibfnamefont {N.~M.}\ \bibnamefont {Linke}}, \bibinfo {author}
  {\bibfnamefont {B.}~\bibnamefont {Yoshida}}, \bibinfo {author} {\bibfnamefont
  {N.~Y.}\ \bibnamefont {Yao}},\ and\ \bibinfo {author} {\bibfnamefont
  {C.}~\bibnamefont {Monroe}},\ }\bibfield  {title} {\bibinfo {title} {Verified
  quantum information scrambling},\ }\href
  {https://doi.org/10.1038/s41586-019-0952-6} {\bibfield  {journal} {\bibinfo
  {journal} {Nature}\ }\textbf {\bibinfo {volume} {567}},\ \bibinfo {pages}
  {61} (\bibinfo {year} {2019})}\BibitemShut {NoStop}%
\bibitem [{\citenamefont {Mi}\ \emph {et~al.}(2021)\citenamefont {Mi},
  \citenamefont {Roushan}, \citenamefont {Quintana}, \citenamefont {Mandra},
  \citenamefont {Marshall}, \citenamefont {Neill}, \citenamefont {Arute},
  \citenamefont {Arya}, \citenamefont {Atalaya}, \citenamefont {Babbush} \emph
  {et~al.}}]{mi2021information}%
  \BibitemOpen
  \bibfield  {author} {\bibinfo {author} {\bibfnamefont {X.}~\bibnamefont
  {Mi}}, \bibinfo {author} {\bibfnamefont {P.}~\bibnamefont {Roushan}},
  \bibinfo {author} {\bibfnamefont {C.}~\bibnamefont {Quintana}}, \bibinfo
  {author} {\bibfnamefont {S.}~\bibnamefont {Mandra}}, \bibinfo {author}
  {\bibfnamefont {J.}~\bibnamefont {Marshall}}, \bibinfo {author}
  {\bibfnamefont {C.}~\bibnamefont {Neill}}, \bibinfo {author} {\bibfnamefont
  {F.}~\bibnamefont {Arute}}, \bibinfo {author} {\bibfnamefont
  {K.}~\bibnamefont {Arya}}, \bibinfo {author} {\bibfnamefont {J.}~\bibnamefont
  {Atalaya}}, \bibinfo {author} {\bibfnamefont {R.}~\bibnamefont {Babbush}},
  \emph {et~al.},\ }\bibfield  {title} {\bibinfo {title} {Information
  scrambling in quantum circuits},\ }\href@noop {} {\bibfield  {journal}
  {\bibinfo  {journal} {Science}\ }\textbf {\bibinfo {volume} {374}},\ \bibinfo
  {pages} {1479} (\bibinfo {year} {2021})}\BibitemShut {NoStop}%
\bibitem [{\citenamefont {Devetak}\ and\ \citenamefont
  {Winter}(2004)}]{devetak2004relating}%
  \BibitemOpen
  \bibfield  {author} {\bibinfo {author} {\bibfnamefont {I.}~\bibnamefont
  {Devetak}}\ and\ \bibinfo {author} {\bibfnamefont {A.}~\bibnamefont
  {Winter}},\ }\bibfield  {title} {\bibinfo {title} {Relating quantum privacy
  and quantum coherence: An operational approach},\ }\href
  {https://doi.org/10.1103/PhysRevLett.93.080501} {\bibfield  {journal}
  {\bibinfo  {journal} {Phys. Rev. Lett.}\ }\textbf {\bibinfo {volume} {93}},\
  \bibinfo {pages} {080501} (\bibinfo {year} {2004})}\BibitemShut {NoStop}%
\bibitem [{\citenamefont {Horodecki}\ \emph {et~al.}(2005)\citenamefont
  {Horodecki}, \citenamefont {Oppenheim},\ and\ \citenamefont
  {Winter}}]{horodecki2005partial}%
  \BibitemOpen
  \bibfield  {author} {\bibinfo {author} {\bibfnamefont {M.}~\bibnamefont
  {Horodecki}}, \bibinfo {author} {\bibfnamefont {J.}~\bibnamefont
  {Oppenheim}},\ and\ \bibinfo {author} {\bibfnamefont {A.}~\bibnamefont
  {Winter}},\ }\bibfield  {title} {\bibinfo {title} {Partial quantum
  information},\ }\href {https://doi.org/10.1038/nature03909} {\bibfield
  {journal} {\bibinfo  {journal} {Nature}\ }\textbf {\bibinfo {volume} {436}},\
  \bibinfo {pages} {673} (\bibinfo {year} {2005})}\BibitemShut {NoStop}%
\bibitem [{\citenamefont {Ambainis}\ \emph {et~al.}(2009)\citenamefont
  {Ambainis}, \citenamefont {Bouda},\ and\ \citenamefont
  {Winter}}]{ambainis2009nonmalleable}%
  \BibitemOpen
  \bibfield  {author} {\bibinfo {author} {\bibfnamefont {A.}~\bibnamefont
  {Ambainis}}, \bibinfo {author} {\bibfnamefont {J.}~\bibnamefont {Bouda}},\
  and\ \bibinfo {author} {\bibfnamefont {A.}~\bibnamefont {Winter}},\
  }\bibfield  {title} {\bibinfo {title} {Nonmalleable encryption of quantum
  information},\ }\href@noop {} {\bibfield  {journal} {\bibinfo  {journal}
  {Journal of Mathematical Physics}\ }\textbf {\bibinfo {volume} {50}},\
  \bibinfo {pages} {042106} (\bibinfo {year} {2009})}\BibitemShut {NoStop}%
\bibitem [{\citenamefont {Chau}(2005)}]{chau2005unconditionally}%
  \BibitemOpen
  \bibfield  {author} {\bibinfo {author} {\bibfnamefont {H.~F.}\ \bibnamefont
  {Chau}},\ }\bibfield  {title} {\bibinfo {title} {Unconditionally secure key
  distribution in higher dimensions by depolarization},\ }\href@noop {}
  {\bibfield  {journal} {\bibinfo  {journal} {IEEE Transactions on Information
  Theory}\ }\textbf {\bibinfo {volume} {51}},\ \bibinfo {pages} {1451}
  (\bibinfo {year} {2005})}\BibitemShut {NoStop}%
\bibitem [{\citenamefont {DiVincenzo}\ \emph {et~al.}(2002)\citenamefont
  {DiVincenzo}, \citenamefont {Leung},\ and\ \citenamefont
  {Terhal}}]{divincenzo2002quantum}%
  \BibitemOpen
  \bibfield  {author} {\bibinfo {author} {\bibfnamefont {D.}~\bibnamefont
  {DiVincenzo}}, \bibinfo {author} {\bibfnamefont {D.}~\bibnamefont {Leung}},\
  and\ \bibinfo {author} {\bibfnamefont {B.}~\bibnamefont {Terhal}},\
  }\bibfield  {title} {\bibinfo {title} {Quantum data hiding},\ }\href
  {https://doi.org/10.1109/18.985948} {\bibfield  {journal} {\bibinfo
  {journal} {IEEE Transactions on Information Theory}\ }\textbf {\bibinfo
  {volume} {48}},\ \bibinfo {pages} {580} (\bibinfo {year} {2002})}\BibitemShut
  {NoStop}%
\bibitem [{\citenamefont {Ambainis}\ and\ \citenamefont
  {Emerson}(2007)}]{ambainis2007quantum}%
  \BibitemOpen
  \bibfield  {author} {\bibinfo {author} {\bibfnamefont {A.}~\bibnamefont
  {Ambainis}}\ and\ \bibinfo {author} {\bibfnamefont {J.}~\bibnamefont
  {Emerson}},\ }\bibfield  {title} {\bibinfo {title} {Quantum t-designs: t-wise
  independence in the quantum world},\ }\href@noop {} {\bibfield  {journal}
  {\bibinfo  {journal} {Twenty-Second Annual IEEE Conference on Computational
  Complexity (CCC'07)}\ ,\ \bibinfo {pages} {129}} (\bibinfo {year}
  {2007})}\BibitemShut {NoStop}%
\bibitem [{\citenamefont {Matthews}\ \emph {et~al.}(2009)\citenamefont
  {Matthews}, \citenamefont {Wehner},\ and\ \citenamefont
  {Winter}}]{matthews2009distinguishability}%
  \BibitemOpen
  \bibfield  {author} {\bibinfo {author} {\bibfnamefont {W.}~\bibnamefont
  {Matthews}}, \bibinfo {author} {\bibfnamefont {S.}~\bibnamefont {Wehner}},\
  and\ \bibinfo {author} {\bibfnamefont {A.}~\bibnamefont {Winter}},\
  }\bibfield  {title} {\bibinfo {title} {Distinguishability of quantum states
  under restricted families of measurements with an application to quantum data
  hiding},\ }\href {https://doi.org/10.1007/s00220-009-0890-5} {\bibfield
  {journal} {\bibinfo  {journal} {Communications in Mathematical Physics}\
  }\textbf {\bibinfo {volume} {291}},\ \bibinfo {pages} {813} (\bibinfo {year}
  {2009})}\BibitemShut {NoStop}%
\bibitem [{\citenamefont {Bouland}\ \emph {et~al.}(2018)\citenamefont
  {Bouland}, \citenamefont {Fefferman}, \citenamefont {Nirkhe},\ and\
  \citenamefont {Vazirani}}]{bouland2018quantum}%
  \BibitemOpen
  \bibfield  {author} {\bibinfo {author} {\bibfnamefont {A.}~\bibnamefont
  {Bouland}}, \bibinfo {author} {\bibfnamefont {B.}~\bibnamefont {Fefferman}},
  \bibinfo {author} {\bibfnamefont {C.}~\bibnamefont {Nirkhe}},\ and\ \bibinfo
  {author} {\bibfnamefont {U.}~\bibnamefont {Vazirani}},\ }\bibfield  {title}
  {\bibinfo {title} {Quantum supremacy and the complexity of random circuit
  sampling},\ }\href@noop {} {\bibfield  {journal} {\bibinfo  {journal} {arXiv
  preprint arXiv:1803.04402}\ } (\bibinfo {year} {2018})}\BibitemShut {NoStop}%
\bibitem [{\citenamefont {Boixo}\ \emph {et~al.}(2018)\citenamefont {Boixo},
  \citenamefont {Isakov}, \citenamefont {Smelyanskiy}, \citenamefont {Babbush},
  \citenamefont {Ding}, \citenamefont {Jiang}, \citenamefont {Bremner},
  \citenamefont {Martinis},\ and\ \citenamefont
  {Neven}}]{boixo2018characterizing}%
  \BibitemOpen
  \bibfield  {author} {\bibinfo {author} {\bibfnamefont {S.}~\bibnamefont
  {Boixo}}, \bibinfo {author} {\bibfnamefont {S.~V.}\ \bibnamefont {Isakov}},
  \bibinfo {author} {\bibfnamefont {V.~N.}\ \bibnamefont {Smelyanskiy}},
  \bibinfo {author} {\bibfnamefont {R.}~\bibnamefont {Babbush}}, \bibinfo
  {author} {\bibfnamefont {N.}~\bibnamefont {Ding}}, \bibinfo {author}
  {\bibfnamefont {Z.}~\bibnamefont {Jiang}}, \bibinfo {author} {\bibfnamefont
  {M.~J.}\ \bibnamefont {Bremner}}, \bibinfo {author} {\bibfnamefont {J.~M.}\
  \bibnamefont {Martinis}},\ and\ \bibinfo {author} {\bibfnamefont
  {H.}~\bibnamefont {Neven}},\ }\bibfield  {title} {\bibinfo {title}
  {Characterizing quantum supremacy in near-term devices},\ }\href@noop {}
  {\bibfield  {journal} {\bibinfo  {journal} {Nature Physics}\ }\textbf
  {\bibinfo {volume} {14}},\ \bibinfo {pages} {595} (\bibinfo {year}
  {2018})}\BibitemShut {NoStop}%
\bibitem [{\citenamefont {Arute}\ \emph {et~al.}(2019)\citenamefont {Arute},
  \citenamefont {Arya}, \citenamefont {Babbush}, \citenamefont {Bacon},
  \citenamefont {Bardin}, \citenamefont {Barends}, \citenamefont {Biswas},
  \citenamefont {Boixo}, \citenamefont {Brandao}, \citenamefont {Buell},
  \citenamefont {Burkett}, \citenamefont {Chen}, \citenamefont {Chen},
  \citenamefont {Chiaro}, \citenamefont {Collins}, \citenamefont {Courtney},
  \citenamefont {Dunsworth}, \citenamefont {Farhi}, \citenamefont {Foxen},
  \citenamefont {Fowler}, \citenamefont {Gidney}, \citenamefont {Giustina},
  \citenamefont {Graff}, \citenamefont {Guerin}, \citenamefont {Habegger},
  \citenamefont {Harrigan}, \citenamefont {Hartmann}, \citenamefont {Ho},
  \citenamefont {Hoffmann}, \citenamefont {Huang}, \citenamefont {Humble},
  \citenamefont {Isakov}, \citenamefont {Jeffrey}, \citenamefont {Jiang},
  \citenamefont {Kafri}, \citenamefont {Kechedzhi}, \citenamefont {Kelly},
  \citenamefont {Klimov}, \citenamefont {Knysh}, \citenamefont {Korotkov},
  \citenamefont {Kostritsa}, \citenamefont {Landhuis}, \citenamefont
  {Lindmark}, \citenamefont {Lucero}, \citenamefont {Lyakh}, \citenamefont
  {Mandr{\`a}}, \citenamefont {McClean}, \citenamefont {McEwen}, \citenamefont
  {Megrant}, \citenamefont {Mi}, \citenamefont {Michielsen}, \citenamefont
  {Mohseni}, \citenamefont {Mutus}, \citenamefont {Naaman}, \citenamefont
  {Neeley}, \citenamefont {Neill}, \citenamefont {Niu}, \citenamefont {Ostby},
  \citenamefont {Petukhov}, \citenamefont {Platt}, \citenamefont {Quintana},
  \citenamefont {Rieffel}, \citenamefont {Roushan}, \citenamefont {Rubin},
  \citenamefont {Sank}, \citenamefont {Satzinger}, \citenamefont {Smelyanskiy},
  \citenamefont {Sung}, \citenamefont {Trevithick}, \citenamefont
  {Vainsencher}, \citenamefont {Villalonga}, \citenamefont {White},
  \citenamefont {Yao}, \citenamefont {Yeh}, \citenamefont {Zalcman},
  \citenamefont {Neven},\ and\ \citenamefont {Martinis}}]{arute2019quantum}%
  \BibitemOpen
  \bibfield  {author} {\bibinfo {author} {\bibfnamefont {F.}~\bibnamefont
  {Arute}}, \bibinfo {author} {\bibfnamefont {K.}~\bibnamefont {Arya}},
  \bibinfo {author} {\bibfnamefont {R.}~\bibnamefont {Babbush}}, \bibinfo
  {author} {\bibfnamefont {D.}~\bibnamefont {Bacon}}, \bibinfo {author}
  {\bibfnamefont {J.~C.}\ \bibnamefont {Bardin}}, \bibinfo {author}
  {\bibfnamefont {R.}~\bibnamefont {Barends}}, \bibinfo {author} {\bibfnamefont
  {R.}~\bibnamefont {Biswas}}, \bibinfo {author} {\bibfnamefont
  {S.}~\bibnamefont {Boixo}}, \bibinfo {author} {\bibfnamefont {F.~G. S.~L.}\
  \bibnamefont {Brandao}}, \bibinfo {author} {\bibfnamefont {D.~A.}\
  \bibnamefont {Buell}}, \bibinfo {author} {\bibfnamefont {B.}~\bibnamefont
  {Burkett}}, \bibinfo {author} {\bibfnamefont {Y.}~\bibnamefont {Chen}},
  \bibinfo {author} {\bibfnamefont {Z.}~\bibnamefont {Chen}}, \bibinfo {author}
  {\bibfnamefont {B.}~\bibnamefont {Chiaro}}, \bibinfo {author} {\bibfnamefont
  {R.}~\bibnamefont {Collins}}, \bibinfo {author} {\bibfnamefont
  {W.}~\bibnamefont {Courtney}}, \bibinfo {author} {\bibfnamefont
  {A.}~\bibnamefont {Dunsworth}}, \bibinfo {author} {\bibfnamefont
  {E.}~\bibnamefont {Farhi}}, \bibinfo {author} {\bibfnamefont
  {B.}~\bibnamefont {Foxen}}, \bibinfo {author} {\bibfnamefont
  {A.}~\bibnamefont {Fowler}}, \bibinfo {author} {\bibfnamefont
  {C.}~\bibnamefont {Gidney}}, \bibinfo {author} {\bibfnamefont
  {M.}~\bibnamefont {Giustina}}, \bibinfo {author} {\bibfnamefont
  {R.}~\bibnamefont {Graff}}, \bibinfo {author} {\bibfnamefont
  {K.}~\bibnamefont {Guerin}}, \bibinfo {author} {\bibfnamefont
  {S.}~\bibnamefont {Habegger}}, \bibinfo {author} {\bibfnamefont {M.~P.}\
  \bibnamefont {Harrigan}}, \bibinfo {author} {\bibfnamefont {M.~J.}\
  \bibnamefont {Hartmann}}, \bibinfo {author} {\bibfnamefont {A.}~\bibnamefont
  {Ho}}, \bibinfo {author} {\bibfnamefont {M.}~\bibnamefont {Hoffmann}},
  \bibinfo {author} {\bibfnamefont {T.}~\bibnamefont {Huang}}, \bibinfo
  {author} {\bibfnamefont {T.~S.}\ \bibnamefont {Humble}}, \bibinfo {author}
  {\bibfnamefont {S.~V.}\ \bibnamefont {Isakov}}, \bibinfo {author}
  {\bibfnamefont {E.}~\bibnamefont {Jeffrey}}, \bibinfo {author} {\bibfnamefont
  {Z.}~\bibnamefont {Jiang}}, \bibinfo {author} {\bibfnamefont
  {D.}~\bibnamefont {Kafri}}, \bibinfo {author} {\bibfnamefont
  {K.}~\bibnamefont {Kechedzhi}}, \bibinfo {author} {\bibfnamefont
  {J.}~\bibnamefont {Kelly}}, \bibinfo {author} {\bibfnamefont {P.~V.}\
  \bibnamefont {Klimov}}, \bibinfo {author} {\bibfnamefont {S.}~\bibnamefont
  {Knysh}}, \bibinfo {author} {\bibfnamefont {A.}~\bibnamefont {Korotkov}},
  \bibinfo {author} {\bibfnamefont {F.}~\bibnamefont {Kostritsa}}, \bibinfo
  {author} {\bibfnamefont {D.}~\bibnamefont {Landhuis}}, \bibinfo {author}
  {\bibfnamefont {M.}~\bibnamefont {Lindmark}}, \bibinfo {author}
  {\bibfnamefont {E.}~\bibnamefont {Lucero}}, \bibinfo {author} {\bibfnamefont
  {D.}~\bibnamefont {Lyakh}}, \bibinfo {author} {\bibfnamefont
  {S.}~\bibnamefont {Mandr{\`a}}}, \bibinfo {author} {\bibfnamefont {J.~R.}\
  \bibnamefont {McClean}}, \bibinfo {author} {\bibfnamefont {M.}~\bibnamefont
  {McEwen}}, \bibinfo {author} {\bibfnamefont {A.}~\bibnamefont {Megrant}},
  \bibinfo {author} {\bibfnamefont {X.}~\bibnamefont {Mi}}, \bibinfo {author}
  {\bibfnamefont {K.}~\bibnamefont {Michielsen}}, \bibinfo {author}
  {\bibfnamefont {M.}~\bibnamefont {Mohseni}}, \bibinfo {author} {\bibfnamefont
  {J.}~\bibnamefont {Mutus}}, \bibinfo {author} {\bibfnamefont
  {O.}~\bibnamefont {Naaman}}, \bibinfo {author} {\bibfnamefont
  {M.}~\bibnamefont {Neeley}}, \bibinfo {author} {\bibfnamefont
  {C.}~\bibnamefont {Neill}}, \bibinfo {author} {\bibfnamefont {M.~Y.}\
  \bibnamefont {Niu}}, \bibinfo {author} {\bibfnamefont {E.}~\bibnamefont
  {Ostby}}, \bibinfo {author} {\bibfnamefont {A.}~\bibnamefont {Petukhov}},
  \bibinfo {author} {\bibfnamefont {J.~C.}\ \bibnamefont {Platt}}, \bibinfo
  {author} {\bibfnamefont {C.}~\bibnamefont {Quintana}}, \bibinfo {author}
  {\bibfnamefont {E.~G.}\ \bibnamefont {Rieffel}}, \bibinfo {author}
  {\bibfnamefont {P.}~\bibnamefont {Roushan}}, \bibinfo {author} {\bibfnamefont
  {N.~C.}\ \bibnamefont {Rubin}}, \bibinfo {author} {\bibfnamefont
  {D.}~\bibnamefont {Sank}}, \bibinfo {author} {\bibfnamefont {K.~J.}\
  \bibnamefont {Satzinger}}, \bibinfo {author} {\bibfnamefont {V.}~\bibnamefont
  {Smelyanskiy}}, \bibinfo {author} {\bibfnamefont {K.~J.}\ \bibnamefont
  {Sung}}, \bibinfo {author} {\bibfnamefont {M.~D.}\ \bibnamefont
  {Trevithick}}, \bibinfo {author} {\bibfnamefont {A.}~\bibnamefont
  {Vainsencher}}, \bibinfo {author} {\bibfnamefont {B.}~\bibnamefont
  {Villalonga}}, \bibinfo {author} {\bibfnamefont {T.}~\bibnamefont {White}},
  \bibinfo {author} {\bibfnamefont {Z.~J.}\ \bibnamefont {Yao}}, \bibinfo
  {author} {\bibfnamefont {P.}~\bibnamefont {Yeh}}, \bibinfo {author}
  {\bibfnamefont {A.}~\bibnamefont {Zalcman}}, \bibinfo {author} {\bibfnamefont
  {H.}~\bibnamefont {Neven}},\ and\ \bibinfo {author} {\bibfnamefont {J.~M.}\
  \bibnamefont {Martinis}},\ }\bibfield  {title} {\bibinfo {title} {Quantum
  supremacy using a programmable superconducting processor},\ }\href
  {https://doi.org/10.1038/s41586-019-1666-5} {\bibfield  {journal} {\bibinfo
  {journal} {Nature}\ }\textbf {\bibinfo {volume} {574}},\ \bibinfo {pages}
  {505} (\bibinfo {year} {2019})}\BibitemShut {NoStop}%
\bibitem [{\citenamefont {Gottesman}(1998)}]{gottesman1998heisenberg}%
  \BibitemOpen
  \bibfield  {author} {\bibinfo {author} {\bibfnamefont {D.}~\bibnamefont
  {Gottesman}},\ }\bibfield  {title} {\bibinfo {title} {The heisenberg
  representation of quantum computers},\ }\href@noop {} {\bibfield  {journal}
  {\bibinfo  {journal} {arXiv preprint quant-ph/9807006}\ } (\bibinfo {year}
  {1998})}\BibitemShut {NoStop}%
\bibitem [{\citenamefont {Aaronson}\ and\ \citenamefont
  {Gottesman}(2004)}]{aaronson2004improved}%
  \BibitemOpen
  \bibfield  {author} {\bibinfo {author} {\bibfnamefont {S.}~\bibnamefont
  {Aaronson}}\ and\ \bibinfo {author} {\bibfnamefont {D.}~\bibnamefont
  {Gottesman}},\ }\bibfield  {title} {\bibinfo {title} {Improved simulation of
  stabilizer circuits},\ }\href@noop {} {\bibfield  {journal} {\bibinfo
  {journal} {Physical Review A}\ }\textbf {\bibinfo {volume} {70}},\ \bibinfo
  {pages} {052328} (\bibinfo {year} {2004})}\BibitemShut {NoStop}%
\bibitem [{\citenamefont {Huang}\ \emph {et~al.}(2020)\citenamefont {Huang},
  \citenamefont {Kueng},\ and\ \citenamefont {Preskill}}]{huang2020predicting}%
  \BibitemOpen
  \bibfield  {author} {\bibinfo {author} {\bibfnamefont {H.-Y.}\ \bibnamefont
  {Huang}}, \bibinfo {author} {\bibfnamefont {R.}~\bibnamefont {Kueng}},\ and\
  \bibinfo {author} {\bibfnamefont {J.}~\bibnamefont {Preskill}},\ }\bibfield
  {title} {\bibinfo {title} {Predicting many properties of a quantum system
  from very few measurements},\ }\href
  {https://doi.org/10.1038/s41567-020-0932-7} {\bibfield  {journal} {\bibinfo
  {journal} {Nature Physics}\ }\textbf {\bibinfo {volume} {16}},\ \bibinfo
  {pages} {1050} (\bibinfo {year} {2020})}\BibitemShut {NoStop}%
\bibitem [{\citenamefont {Emerson}\ \emph {et~al.}(2005)\citenamefont
  {Emerson}, \citenamefont {Alicki},\ and\ \citenamefont
  {{\.Z}yczkowski}}]{emerson2005scalable}%
  \BibitemOpen
  \bibfield  {author} {\bibinfo {author} {\bibfnamefont {J.}~\bibnamefont
  {Emerson}}, \bibinfo {author} {\bibfnamefont {R.}~\bibnamefont {Alicki}},\
  and\ \bibinfo {author} {\bibfnamefont {K.}~\bibnamefont {{\.Z}yczkowski}},\
  }\bibfield  {title} {\bibinfo {title} {Scalable noise estimation with random
  unitary operators},\ }\href@noop {} {\bibfield  {journal} {\bibinfo
  {journal} {Journal of Optics B: Quantum and Semiclassical Optics}\ }\textbf
  {\bibinfo {volume} {7}},\ \bibinfo {pages} {S347} (\bibinfo {year}
  {2005})}\BibitemShut {NoStop}%
\bibitem [{\citenamefont {Knill}\ \emph {et~al.}(2008)\citenamefont {Knill},
  \citenamefont {Leibfried}, \citenamefont {Reichle}, \citenamefont {Britton},
  \citenamefont {Blakestad}, \citenamefont {Jost}, \citenamefont {Langer},
  \citenamefont {Ozeri}, \citenamefont {Seidelin},\ and\ \citenamefont
  {Wineland}}]{knill2008randomized}%
  \BibitemOpen
  \bibfield  {author} {\bibinfo {author} {\bibfnamefont {E.}~\bibnamefont
  {Knill}}, \bibinfo {author} {\bibfnamefont {D.}~\bibnamefont {Leibfried}},
  \bibinfo {author} {\bibfnamefont {R.}~\bibnamefont {Reichle}}, \bibinfo
  {author} {\bibfnamefont {J.}~\bibnamefont {Britton}}, \bibinfo {author}
  {\bibfnamefont {R.~B.}\ \bibnamefont {Blakestad}}, \bibinfo {author}
  {\bibfnamefont {J.~D.}\ \bibnamefont {Jost}}, \bibinfo {author}
  {\bibfnamefont {C.}~\bibnamefont {Langer}}, \bibinfo {author} {\bibfnamefont
  {R.}~\bibnamefont {Ozeri}}, \bibinfo {author} {\bibfnamefont
  {S.}~\bibnamefont {Seidelin}},\ and\ \bibinfo {author} {\bibfnamefont
  {D.~J.}\ \bibnamefont {Wineland}},\ }\bibfield  {title} {\bibinfo {title}
  {Randomized benchmarking of quantum gates},\ }\href
  {https://doi.org/10.1103/PhysRevA.77.012307} {\bibfield  {journal} {\bibinfo
  {journal} {Phys. Rev. A}\ }\textbf {\bibinfo {volume} {77}},\ \bibinfo
  {pages} {012307} (\bibinfo {year} {2008})}\BibitemShut {NoStop}%
\bibitem [{\citenamefont {Magesan}\ \emph {et~al.}(2011)\citenamefont
  {Magesan}, \citenamefont {Gambetta},\ and\ \citenamefont
  {Emerson}}]{magesan2011scalable}%
  \BibitemOpen
  \bibfield  {author} {\bibinfo {author} {\bibfnamefont {E.}~\bibnamefont
  {Magesan}}, \bibinfo {author} {\bibfnamefont {J.~M.}\ \bibnamefont
  {Gambetta}},\ and\ \bibinfo {author} {\bibfnamefont {J.}~\bibnamefont
  {Emerson}},\ }\bibfield  {title} {\bibinfo {title} {Scalable and robust
  randomized benchmarking of quantum processes},\ }\href
  {https://doi.org/10.1103/PhysRevLett.106.180504} {\bibfield  {journal}
  {\bibinfo  {journal} {Phys. Rev. Lett.}\ }\textbf {\bibinfo {volume} {106}},\
  \bibinfo {pages} {180504} (\bibinfo {year} {2011})}\BibitemShut {NoStop}%
\bibitem [{\citenamefont {Nakata}\ \emph {et~al.}(2021)\citenamefont {Nakata},
  \citenamefont {Zhao}, \citenamefont {Okuda}, \citenamefont {Bannai},
  \citenamefont {Suzuki}, \citenamefont {Tamiya}, \citenamefont {Heya},
  \citenamefont {Yan}, \citenamefont {Zuo}, \citenamefont {Tamate},
  \citenamefont {Tabuchi},\ and\ \citenamefont {Nakamura}}]{nakata2021quantum}%
  \BibitemOpen
  \bibfield  {author} {\bibinfo {author} {\bibfnamefont {Y.}~\bibnamefont
  {Nakata}}, \bibinfo {author} {\bibfnamefont {D.}~\bibnamefont {Zhao}},
  \bibinfo {author} {\bibfnamefont {T.}~\bibnamefont {Okuda}}, \bibinfo
  {author} {\bibfnamefont {E.}~\bibnamefont {Bannai}}, \bibinfo {author}
  {\bibfnamefont {Y.}~\bibnamefont {Suzuki}}, \bibinfo {author} {\bibfnamefont
  {S.}~\bibnamefont {Tamiya}}, \bibinfo {author} {\bibfnamefont
  {K.}~\bibnamefont {Heya}}, \bibinfo {author} {\bibfnamefont {Z.}~\bibnamefont
  {Yan}}, \bibinfo {author} {\bibfnamefont {K.}~\bibnamefont {Zuo}}, \bibinfo
  {author} {\bibfnamefont {S.}~\bibnamefont {Tamate}}, \bibinfo {author}
  {\bibfnamefont {Y.}~\bibnamefont {Tabuchi}},\ and\ \bibinfo {author}
  {\bibfnamefont {Y.}~\bibnamefont {Nakamura}},\ }\bibfield  {title} {\bibinfo
  {title} {Quantum circuits for exact unitary $t$-designs and applications to
  higher-order randomized benchmarking},\ }\href
  {https://doi.org/10.1103/PRXQuantum.2.030339} {\bibfield  {journal} {\bibinfo
   {journal} {PRX Quantum}\ }\textbf {\bibinfo {volume} {2}},\ \bibinfo {pages}
  {030339} (\bibinfo {year} {2021})}\BibitemShut {NoStop}%
\bibitem [{\citenamefont {Graydon}\ \emph {et~al.}(2021)\citenamefont
  {Graydon}, \citenamefont {Skanes-Norman},\ and\ \citenamefont
  {Wallman}}]{graydon2021clifford}%
  \BibitemOpen
  \bibfield  {author} {\bibinfo {author} {\bibfnamefont {M.~A.}\ \bibnamefont
  {Graydon}}, \bibinfo {author} {\bibfnamefont {J.}~\bibnamefont
  {Skanes-Norman}},\ and\ \bibinfo {author} {\bibfnamefont {J.~J.}\
  \bibnamefont {Wallman}},\ }\bibfield  {title} {\bibinfo {title} {Clifford
  groups are not always 2-designs},\ }\href@noop {} {\bibfield  {journal}
  {\bibinfo  {journal} {arXiv preprint arXiv:2108.04200}\ } (\bibinfo {year}
  {2021})}\BibitemShut {NoStop}%
\bibitem [{\citenamefont {D\"ur}\ \emph {et~al.}(2005)\citenamefont {D\"ur},
  \citenamefont {Hein}, \citenamefont {Cirac},\ and\ \citenamefont
  {Briegel}}]{dur2005standard}%
  \BibitemOpen
  \bibfield  {author} {\bibinfo {author} {\bibfnamefont {W.}~\bibnamefont
  {D\"ur}}, \bibinfo {author} {\bibfnamefont {M.}~\bibnamefont {Hein}},
  \bibinfo {author} {\bibfnamefont {J.~I.}\ \bibnamefont {Cirac}},\ and\
  \bibinfo {author} {\bibfnamefont {H.-J.}\ \bibnamefont {Briegel}},\
  }\bibfield  {title} {\bibinfo {title} {Standard forms of noisy quantum
  operations via depolarization},\ }\href
  {https://doi.org/10.1103/PhysRevA.72.052326} {\bibfield  {journal} {\bibinfo
  {journal} {Phys. Rev. A}\ }\textbf {\bibinfo {volume} {72}},\ \bibinfo
  {pages} {052326} (\bibinfo {year} {2005})}\BibitemShut {NoStop}%
\bibitem [{\citenamefont {Gross}\ \emph {et~al.}(2007)\citenamefont {Gross},
  \citenamefont {Audenaert},\ and\ \citenamefont {Eisert}}]{gross2007evenly}%
  \BibitemOpen
  \bibfield  {author} {\bibinfo {author} {\bibfnamefont {D.}~\bibnamefont
  {Gross}}, \bibinfo {author} {\bibfnamefont {K.}~\bibnamefont {Audenaert}},\
  and\ \bibinfo {author} {\bibfnamefont {J.}~\bibnamefont {Eisert}},\
  }\bibfield  {title} {\bibinfo {title} {Evenly distributed unitaries: On the
  structure of unitary designs},\ }\href@noop {} {\bibfield  {journal}
  {\bibinfo  {journal} {Journal of mathematical physics}\ }\textbf {\bibinfo
  {volume} {48}},\ \bibinfo {pages} {052104} (\bibinfo {year}
  {2007})}\BibitemShut {NoStop}%
\bibitem [{\citenamefont {Webb}(2016)}]{webb2016clifford}%
  \BibitemOpen
  \bibfield  {author} {\bibinfo {author} {\bibfnamefont {Z.}~\bibnamefont
  {Webb}},\ }\bibfield  {title} {\bibinfo {title} {The clifford group forms a
  unitary 3-design},\ }\href
  {https://dl.acm.org/doi/abs/10.5555/3179439.3179447} {\bibfield  {journal}
  {\bibinfo  {journal} {Quantum Info. Comput.}\ }\textbf {\bibinfo {volume}
  {16}},\ \bibinfo {pages} {1379} (\bibinfo {year} {2016})}\BibitemShut
  {NoStop}%
\bibitem [{\citenamefont {Zhu}(2017)}]{zhu2017multiqubit}%
  \BibitemOpen
  \bibfield  {author} {\bibinfo {author} {\bibfnamefont {H.}~\bibnamefont
  {Zhu}},\ }\bibfield  {title} {\bibinfo {title} {Multiqubit clifford groups
  are unitary 3-designs},\ }\href {https://doi.org/10.1103/PhysRevA.96.062336}
  {\bibfield  {journal} {\bibinfo  {journal} {Phys. Rev. A}\ }\textbf {\bibinfo
  {volume} {96}},\ \bibinfo {pages} {062336} (\bibinfo {year}
  {2017})}\BibitemShut {NoStop}%
\bibitem [{\citenamefont {Li}\ \emph {et~al.}(2019)\citenamefont {Li},
  \citenamefont {Chen},\ and\ \citenamefont {Fisher}}]{Li2019measurement}%
  \BibitemOpen
  \bibfield  {author} {\bibinfo {author} {\bibfnamefont {Y.}~\bibnamefont
  {Li}}, \bibinfo {author} {\bibfnamefont {X.}~\bibnamefont {Chen}},\ and\
  \bibinfo {author} {\bibfnamefont {M.~P.~A.}\ \bibnamefont {Fisher}},\
  }\bibfield  {title} {\bibinfo {title} {Measurement-driven entanglement
  transition in hybrid quantum circuits},\ }\href
  {https://doi.org/10.1103/PhysRevB.100.134306} {\bibfield  {journal} {\bibinfo
   {journal} {Phys. Rev. B}\ }\textbf {\bibinfo {volume} {100}},\ \bibinfo
  {pages} {134306} (\bibinfo {year} {2019})}\BibitemShut {NoStop}%
\bibitem [{\citenamefont {Farshi}\ \emph {et~al.}(2022)\citenamefont {Farshi},
  \citenamefont {Toniolo}, \citenamefont {Gonz\'{a}lez-Guill\'{e}n},
  \citenamefont {Alhambra},\ and\ \citenamefont {Masanes}}]{Farshi2022mixing}%
  \BibitemOpen
  \bibfield  {author} {\bibinfo {author} {\bibfnamefont {T.}~\bibnamefont
  {Farshi}}, \bibinfo {author} {\bibfnamefont {D.}~\bibnamefont {Toniolo}},
  \bibinfo {author} {\bibfnamefont {C.~E.}\ \bibnamefont
  {Gonz\'{a}lez-Guill\'{e}n}}, \bibinfo {author} {\bibfnamefont {A.~M.}\
  \bibnamefont {Alhambra}},\ and\ \bibinfo {author} {\bibfnamefont
  {L.}~\bibnamefont {Masanes}},\ }\bibfield  {title} {\bibinfo {title} {Mixing
  and localization in random time-periodic quantum circuits of clifford
  unitaries},\ }\href {https://doi.org/10.1063/5.0054863} {\bibfield  {journal}
  {\bibinfo  {journal} {J. Math. Phys.}\ }\textbf {\bibinfo {volume} {63}},\
  \bibinfo {pages} {032201} (\bibinfo {year} {2022})}\BibitemShut {NoStop}%
\bibitem [{\citenamefont {Farshi}\ \emph {et~al.}(2023)\citenamefont {Farshi},
  \citenamefont {Richter}, \citenamefont {Toniolo}, \citenamefont {Pal},\ and\
  \citenamefont {Masanes}}]{Farshi2023absence}%
  \BibitemOpen
  \bibfield  {author} {\bibinfo {author} {\bibfnamefont {T.}~\bibnamefont
  {Farshi}}, \bibinfo {author} {\bibfnamefont {J.}~\bibnamefont {Richter}},
  \bibinfo {author} {\bibfnamefont {D.}~\bibnamefont {Toniolo}}, \bibinfo
  {author} {\bibfnamefont {A.}~\bibnamefont {Pal}},\ and\ \bibinfo {author}
  {\bibfnamefont {L.}~\bibnamefont {Masanes}},\ }\bibfield  {title} {\bibinfo
  {title} {Absence of localization in two-dimensional clifford circuits},\
  }\href {https://doi.org/10.1103/PRXQuantum.4.030302} {\bibfield  {journal}
  {\bibinfo  {journal} {PRX Quantum}\ }\textbf {\bibinfo {volume} {4}},\
  \bibinfo {pages} {030302} (\bibinfo {year} {2023})}\BibitemShut {NoStop}%
\bibitem [{\citenamefont {Liu}\ \emph {et~al.}(2023)\citenamefont {Liu},
  \citenamefont {Smith}, \citenamefont {Knap},\ and\ \citenamefont
  {Pollmann}}]{Liu2023model-independent}%
  \BibitemOpen
  \bibfield  {author} {\bibinfo {author} {\bibfnamefont {Y.-J.}\ \bibnamefont
  {Liu}}, \bibinfo {author} {\bibfnamefont {A.}~\bibnamefont {Smith}}, \bibinfo
  {author} {\bibfnamefont {M.}~\bibnamefont {Knap}},\ and\ \bibinfo {author}
  {\bibfnamefont {F.}~\bibnamefont {Pollmann}},\ }\bibfield  {title} {\bibinfo
  {title} {Model-independent learning of quantum phases of matter with quantum
  convolutional neural networks},\ }\href
  {https://doi.org/10.1103/PhysRevLett.130.220603} {\bibfield  {journal}
  {\bibinfo  {journal} {Phys. Rev. Lett.}\ }\textbf {\bibinfo {volume} {130}},\
  \bibinfo {pages} {220603} (\bibinfo {year} {2023})}\BibitemShut {NoStop}%
\bibitem [{\citenamefont {Scott}(2008)}]{scott2008optimizing}%
  \BibitemOpen
  \bibfield  {author} {\bibinfo {author} {\bibfnamefont {A.~J.}\ \bibnamefont
  {Scott}},\ }\bibfield  {title} {\bibinfo {title} {Optimizing quantum process
  tomography with unitary 2-designs},\ }\href
  {https://dx.doi.org/10.1088/1751-8113/41/5/055308} {\bibfield  {journal}
  {\bibinfo  {journal} {J. Phys. A: Math. Theor.}\ }\textbf {\bibinfo {volume}
  {41}},\ \bibinfo {pages} {055308} (\bibinfo {year} {2008})}\BibitemShut
  {NoStop}%
\bibitem [{\citenamefont {Selinger}(2015)}]{selinger2015generators}%
  \BibitemOpen
  \bibfield  {author} {\bibinfo {author} {\bibfnamefont {P.}~\bibnamefont
  {Selinger}},\ }\bibfield  {title} {\bibinfo {title} {Generators and relations
  for n-qubit clifford operators},\ }\href@noop {} {\bibfield  {journal}
  {\bibinfo  {journal} {Logical Methods in Computer Science}\ }\textbf
  {\bibinfo {volume} {11}} (\bibinfo {year} {2015})}\BibitemShut {NoStop}%
\bibitem [{\citenamefont {Bravyi}\ and\ \citenamefont
  {Maslov}(2021)}]{bravyi2021hadamard}%
  \BibitemOpen
  \bibfield  {author} {\bibinfo {author} {\bibfnamefont {S.}~\bibnamefont
  {Bravyi}}\ and\ \bibinfo {author} {\bibfnamefont {D.}~\bibnamefont
  {Maslov}},\ }\bibfield  {title} {\bibinfo {title} {Hadamard-free circuits
  expose the structure of the clifford group},\ }\href
  {https://doi.org/10.1109/TIT.2021.3081415} {\bibfield  {journal} {\bibinfo
  {journal} {IEEE Transactions on Information Theory}\ }\textbf {\bibinfo
  {volume} {67}},\ \bibinfo {pages} {4546} (\bibinfo {year}
  {2021})}\BibitemShut {NoStop}%
\bibitem [{\citenamefont {Calderbank}\ \emph {et~al.}(1998)\citenamefont
  {Calderbank}, \citenamefont {Rains}, \citenamefont {Shor},\ and\
  \citenamefont {Sloane}}]{Calderbank1998}%
  \BibitemOpen
  \bibfield  {author} {\bibinfo {author} {\bibfnamefont {A.~R.}\ \bibnamefont
  {Calderbank}}, \bibinfo {author} {\bibfnamefont {E.~M.}\ \bibnamefont
  {Rains}}, \bibinfo {author} {\bibfnamefont {P.~M.}\ \bibnamefont {Shor}},\
  and\ \bibinfo {author} {\bibfnamefont {N.~J.}\ \bibnamefont {Sloane}},\
  }\bibfield  {title} {\bibinfo {title} {Quantum error correction via codes
  over gf (4)},\ }\href@noop {} {\bibfield  {journal} {\bibinfo  {journal}
  {IEEE Transactions on Information Theory}\ }\textbf {\bibinfo {volume}
  {44}},\ \bibinfo {pages} {1369} (\bibinfo {year} {1998})}\BibitemShut
  {NoStop}%
\bibitem [{\citenamefont {Bravyi}\ \emph {et~al.}(2017)\citenamefont {Bravyi},
  \citenamefont {Gambetta}, \citenamefont {Mezzacapo},\ and\ \citenamefont
  {Temme}}]{bravyi2017tapering}%
  \BibitemOpen
  \bibfield  {author} {\bibinfo {author} {\bibfnamefont {S.}~\bibnamefont
  {Bravyi}}, \bibinfo {author} {\bibfnamefont {J.~M.}\ \bibnamefont
  {Gambetta}}, \bibinfo {author} {\bibfnamefont {A.}~\bibnamefont
  {Mezzacapo}},\ and\ \bibinfo {author} {\bibfnamefont {K.}~\bibnamefont
  {Temme}},\ }\bibfield  {title} {\bibinfo {title} {Tapering off qubits to
  simulate fermionic hamiltonians},\ }\href@noop {} {\bibfield  {journal}
  {\bibinfo  {journal} {arXiv preprint arXiv:1701.08213}\ } (\bibinfo {year}
  {2017})}\BibitemShut {NoStop}%
\bibitem [{\citenamefont {Setia}\ \emph {et~al.}(2020)\citenamefont {Setia},
  \citenamefont {Chen}, \citenamefont {Rice}, \citenamefont {Mezzacapo},
  \citenamefont {Pistoia},\ and\ \citenamefont
  {Whitfield}}]{setia2020reducing}%
  \BibitemOpen
  \bibfield  {author} {\bibinfo {author} {\bibfnamefont {K.}~\bibnamefont
  {Setia}}, \bibinfo {author} {\bibfnamefont {R.}~\bibnamefont {Chen}},
  \bibinfo {author} {\bibfnamefont {J.~E.}\ \bibnamefont {Rice}}, \bibinfo
  {author} {\bibfnamefont {A.}~\bibnamefont {Mezzacapo}}, \bibinfo {author}
  {\bibfnamefont {M.}~\bibnamefont {Pistoia}},\ and\ \bibinfo {author}
  {\bibfnamefont {J.~D.}\ \bibnamefont {Whitfield}},\ }\bibfield  {title}
  {\bibinfo {title} {Reducing qubit requirements for quantum simulations using
  molecular point group symmetries},\ }\href
  {https://doi.org/10.1021/acs.jctc.0c00113} {\bibfield  {journal} {\bibinfo
  {journal} {Journal of Chemical Theory and Computation}\ }\textbf {\bibinfo
  {volume} {16}},\ \bibinfo {pages} {6091} (\bibinfo {year}
  {2020})}\BibitemShut {NoStop}%
\bibitem [{\citenamefont {Suzuki}\ \emph {et~al.}(2021)\citenamefont {Suzuki},
  \citenamefont {Kawase}, \citenamefont {Masumura}, \citenamefont {Hiraga},
  \citenamefont {Nakadai}, \citenamefont {Chen}, \citenamefont {Nakanishi},
  \citenamefont {Mitarai}, \citenamefont {Imai}, \citenamefont {Tamiya} \emph
  {et~al.}}]{suzuki2021qulacs}%
  \BibitemOpen
  \bibfield  {author} {\bibinfo {author} {\bibfnamefont {Y.}~\bibnamefont
  {Suzuki}}, \bibinfo {author} {\bibfnamefont {Y.}~\bibnamefont {Kawase}},
  \bibinfo {author} {\bibfnamefont {Y.}~\bibnamefont {Masumura}}, \bibinfo
  {author} {\bibfnamefont {Y.}~\bibnamefont {Hiraga}}, \bibinfo {author}
  {\bibfnamefont {M.}~\bibnamefont {Nakadai}}, \bibinfo {author} {\bibfnamefont
  {J.}~\bibnamefont {Chen}}, \bibinfo {author} {\bibfnamefont {K.~M.}\
  \bibnamefont {Nakanishi}}, \bibinfo {author} {\bibfnamefont {K.}~\bibnamefont
  {Mitarai}}, \bibinfo {author} {\bibfnamefont {R.}~\bibnamefont {Imai}},
  \bibinfo {author} {\bibfnamefont {S.}~\bibnamefont {Tamiya}}, \emph
  {et~al.},\ }\bibfield  {title} {\bibinfo {title} {Qulacs: a fast and
  versatile quantum circuit simulator for research purpose},\ }\href@noop {}
  {\bibfield  {journal} {\bibinfo  {journal} {Quantum}\ }\textbf {\bibinfo
  {volume} {5}},\ \bibinfo {pages} {559} (\bibinfo {year} {2021})}\BibitemShut
  {NoStop}%
\bibitem [{\citenamefont {Bartlett}\ \emph {et~al.}(2007)\citenamefont
  {Bartlett}, \citenamefont {Rudolph},\ and\ \citenamefont
  {Spekkens}}]{bartlett2007reference}%
  \BibitemOpen
  \bibfield  {author} {\bibinfo {author} {\bibfnamefont {S.~D.}\ \bibnamefont
  {Bartlett}}, \bibinfo {author} {\bibfnamefont {T.}~\bibnamefont {Rudolph}},\
  and\ \bibinfo {author} {\bibfnamefont {R.~W.}\ \bibnamefont {Spekkens}},\
  }\bibfield  {title} {\bibinfo {title} {Reference frames, superselection
  rules, and quantum information},\ }\href
  {https://doi.org/10.1103/RevModPhys.79.555} {\bibfield  {journal} {\bibinfo
  {journal} {Rev. Mod. Phys.}\ }\textbf {\bibinfo {volume} {79}},\ \bibinfo
  {pages} {555} (\bibinfo {year} {2007})}\BibitemShut {NoStop}%
\bibitem [{\citenamefont {Marvian}(2022)}]{marvian2022restrictions}%
  \BibitemOpen
  \bibfield  {author} {\bibinfo {author} {\bibfnamefont {I.}~\bibnamefont
  {Marvian}},\ }\bibfield  {title} {\bibinfo {title} {Restrictions on
  realizable unitary operations imposed by symmetry and locality},\ }\href
  {https://doi.org/10.1038/s41567-021-01464-0} {\bibfield  {journal} {\bibinfo
  {journal} {Nat. Phys.}\ }\textbf {\bibinfo {volume} {18}},\ \bibinfo {pages}
  {283} (\bibinfo {year} {2022})}\BibitemShut {NoStop}%
\bibitem [{\citenamefont {Brandao}\ \emph {et~al.}(2016)\citenamefont
  {Brandao}, \citenamefont {Harrow},\ and\ \citenamefont
  {Horodecki}}]{brandao2016local}%
  \BibitemOpen
  \bibfield  {author} {\bibinfo {author} {\bibfnamefont {F.~G.}\ \bibnamefont
  {Brandao}}, \bibinfo {author} {\bibfnamefont {A.~W.}\ \bibnamefont
  {Harrow}},\ and\ \bibinfo {author} {\bibfnamefont {M.}~\bibnamefont
  {Horodecki}},\ }\bibfield  {title} {\bibinfo {title} {Local random quantum
  circuits are approximate polynomial-designs},\ }\href@noop {} {\bibfield
  {journal} {\bibinfo  {journal} {Communications in Mathematical Physics}\
  }\textbf {\bibinfo {volume} {346}},\ \bibinfo {pages} {397} (\bibinfo {year}
  {2016})}\BibitemShut {NoStop}%
\bibitem [{\citenamefont {Mezher}\ \emph {et~al.}(2019)\citenamefont {Mezher},
  \citenamefont {Ghalbouni}, \citenamefont {Dgheim},\ and\ \citenamefont
  {Markham}}]{mezher2019efficient}%
  \BibitemOpen
  \bibfield  {author} {\bibinfo {author} {\bibfnamefont {R.}~\bibnamefont
  {Mezher}}, \bibinfo {author} {\bibfnamefont {J.}~\bibnamefont {Ghalbouni}},
  \bibinfo {author} {\bibfnamefont {J.}~\bibnamefont {Dgheim}},\ and\ \bibinfo
  {author} {\bibfnamefont {D.}~\bibnamefont {Markham}},\ }\bibfield  {title}
  {\bibinfo {title} {Efficient approximate unitary t-designs from partially
  invertible universal sets and their application to quantum speedup},\
  }\href@noop {} {\bibfield  {journal} {\bibinfo  {journal} {arXiv preprint
  arXiv:1905.01504}\ } (\bibinfo {year} {2019})}\BibitemShut {NoStop}%
\bibitem [{\citenamefont {Ravi}\ \emph {et~al.}(2022)\citenamefont {Ravi},
  \citenamefont {Gokhale}, \citenamefont {Ding}, \citenamefont {Kirby},
  \citenamefont {Smith}, \citenamefont {Baker}, \citenamefont {Love},
  \citenamefont {Hoffmann}, \citenamefont {Brown},\ and\ \citenamefont
  {Chong}}]{ravi2022cafqa}%
  \BibitemOpen
  \bibfield  {author} {\bibinfo {author} {\bibfnamefont {G.~S.}\ \bibnamefont
  {Ravi}}, \bibinfo {author} {\bibfnamefont {P.}~\bibnamefont {Gokhale}},
  \bibinfo {author} {\bibfnamefont {Y.}~\bibnamefont {Ding}}, \bibinfo {author}
  {\bibfnamefont {W.}~\bibnamefont {Kirby}}, \bibinfo {author} {\bibfnamefont
  {K.}~\bibnamefont {Smith}}, \bibinfo {author} {\bibfnamefont {J.~M.}\
  \bibnamefont {Baker}}, \bibinfo {author} {\bibfnamefont {P.~J.}\ \bibnamefont
  {Love}}, \bibinfo {author} {\bibfnamefont {H.}~\bibnamefont {Hoffmann}},
  \bibinfo {author} {\bibfnamefont {K.~R.}\ \bibnamefont {Brown}},\ and\
  \bibinfo {author} {\bibfnamefont {F.~T.}\ \bibnamefont {Chong}},\ }\bibfield
  {title} {\bibinfo {title} {Cafqa: A classical simulation bootstrap for
  variational quantum algorithms},\ }\href@noop {} {\bibfield  {journal}
  {\bibinfo  {journal} {Proceedings of the 28th ACM International Conference on
  Architectural Support for Programming Languages and Operating Systems, Volume
  1}\ ,\ \bibinfo {pages} {15}} (\bibinfo {year} {2022})}\BibitemShut {NoStop}%
\bibitem [{knapp2002lie()}]{knapp2002lie}%
  \BibitemOpen
  \bibinfo {note} {A. W. Knapp, {\it Lie Groups Beyond an Introduction}. 2nd
  ed. (Birkh{\"a}user, Boston, 2002)}\BibitemShut {NoStop}%
\bibitem [{\citenamefont {Marvian}\ and\ \citenamefont
  {Spekkens}(2014)}]{marvian2014generalization}%
  \BibitemOpen
  \bibfield  {author} {\bibinfo {author} {\bibfnamefont {I.}~\bibnamefont
  {Marvian}}\ and\ \bibinfo {author} {\bibfnamefont {R.~W.}\ \bibnamefont
  {Spekkens}},\ }\bibfield  {title} {\bibinfo {title} {A generalization of
  schur--weyl duality with applications in quantum estimation},\ }\href
  {https://doi.org/https://doi.org/10.1007/s00220-014-2059-0} {\bibfield
  {journal} {\bibinfo  {journal} {Communications in Mathematical Physics}\
  }\textbf {\bibinfo {volume} {331}},\ \bibinfo {pages} {431} (\bibinfo {year}
  {2014})}\BibitemShut {NoStop}%
\bibitem [{watrous2018theory()}]{watrous2018theory}%
  \BibitemOpen
  \bibinfo {note} {J. Watrous, {\it The Theory of Quantum Information},
  (Cambridge University Press, Cambridge, 2018)}\BibitemShut {NoStop}%
\bibitem [{hall2003lie()}]{hall2003lie}%
  \BibitemOpen
  \bibinfo {note} {B. C. Hall, {\it Lie Groups, Lie Algebras, and
  Representations: An Elementary Introduction}, (Springer-Verlag, New York,
  2003)}\BibitemShut {NoStop}%
\bibitem [{raczka1986theory()}]{raczka1986theory}%
  \BibitemOpen
  \bibinfo {note} {A. O. Barut and R. R\c{a}czka, {\it Theory of Group
  Representations and Applications}, (World Scientific, Singapore,
  1986)}\BibitemShut {NoStop}%
\end{thebibliography}%

\let\addcontentsline\oldaddcontentsline% Restore \addcontentsline

\onecolumngrid
\appendix

{
\hypersetup{linkcolor=blue}
\tableofcontents
}

\section{General Remarks for Detailed Proofs}

In the following, we present the detailed proofs of the theorems in the main text. 
In Appendix~\ref{SMsec:3design}, we prove Theorem~\ref{thm:main} in the main text, which states that the symmetric Clifford group is a symmetric unitary $3$-design if and only if the symmetry constraint can be written as the commutativity with a Pauli subgroup. 
In Appendix~\ref{SMsec:construction}, we prove Theorem~\ref{thm:expression} in the main text, which gives a one-to-one correspondence from the sets of elementary gates to the symmetric Clifford gates. 
In Appendix~\ref{SMsec:1design}, we prove the former half of Theorem~\ref{thm:1design} in the main text. 
We show that the symmetric Clifford group is a symmetric unitary $1$-design under the U(1) and SU(2) symmetries, which are not Pauli symmetries. 
In Appendix~\ref{SMsec:2design}, we prove the latter half of Theorem~\ref{thm:1design} in the main text in a more general form. 
We take a larger class of symmetries than the one in Appendix~\ref{SMsec:1design}, and show that the symmetric Clifford group is not a symmetric unitary $2$-design for those symmetries. 
In Appendix~\ref{SMsec:technical}, we present the technical lemmas that we use in the proofs of the statements above.

Before going into the details, we introduce the notations in the following appendices. 
For general Hilbert spaces $\mathcal{K}$ and $\mathcal{K}'$, we denote the set of all linear operators from $\mathcal{K}$ to $\mathcal{K}'$, all linear operators on $\mathcal{K}$, and all unitary operators on $\mathcal{K}$ by $\mathcal{L}(\mathcal{K}\to\mathcal{K}')$, $\mathcal{L}(\mathcal{K})$ and $\mathcal{U}(\mathcal{K})$, respectively. 
We denote the Hilbert space for $N$ qubits by $\mathcal{H}$. 
We denote the unitary group on $N$ qubits by $\mathcal{U}_N$. 
We denote the Pauli group on $N$ qubits by $\mathcal{P}_N$. 
We define the Clifford group on $N$ qubits as the normalizer of $\mathcal{P}_N$ and denote it by $\mathcal{C}_N$. 
As for Clifford operators on a single qubit, we denote the Pauli-X, Y and Z, Hadamard and S operators on the $j$th qubit by $\mathrm{X}_j$, $\mathrm{Y}_j$, $\mathrm{Z}_j$, $\mathrm{H}_j$ and $\mathrm{S}_j$, respectively. 
As for Clifford operators on two qubits, we denote the controlled-NOT, the controlled-Z and the SWAP operators on the $j$th and $k$th qubits by $\mathrm{CNOT}_{j, k}$, $\mathrm{CZ}_{j, k}$ and $\mathrm{SWAP}_{j, k}$, respectively, where the $j$th qubit is the control qubit and the $k$th qubit is the target qubit of $\mathrm{CNOT}_{j, k}$. 
We denote the set of Pauli operators without phase by $\mathcal{P}_N^+:=\{\mathrm{I}, \mathrm{X}, \mathrm{Y}, \mathrm{Z}\}^{\otimes N}$, where $\mathcal{A}^{\otimes n}:=\mathcal{A}\otimes\mathcal{A}^{\otimes n-1}$ and $\mathcal{A}\otimes \mathcal{B}:=\{A\otimes B | A\in\mathcal{A}, B\in\mathcal{B}\}$ for general operator sets $\mathcal{A}$ and $\mathcal{B}$. 
For convenience, we formally define $\mathcal{U}_0:=\{e^{i\theta}\ |\ \theta\in\mathbb{R}\}$ and $\mathcal{P}_0:=\{\pm1, \pm i\}$. 
We denote the symmetric group of degree $M$ by $\mathfrak{S}_M$. 
We denote by $\braket{\mathcal{O}}$ the group generated by the operators in a set $\mathcal{O}$. 
We denote $a\equiv b$ (mod $r$) when $a-b$ is an integer multiple of $r$.

\section{Proof of Theorem~\ref{thm:main} (Unitary $3$-designs)} \label{SMsec:3design}

In this appendix, we define the notion of symmetric Clifford group and symmetric unitary design, and prove Theorem~\ref{thm:main} in the main text, namely the statement that the symmetric Clifford group is a symmetric unitary $3$-design if and only if the symmetry constraint is described by some Pauli subgroup. \\

First, we define the symmetric Clifford group as the symmetric subgroup of the conventional Clifford group.

\begin{definition}
    (Restatement of Definition~\ref{def:symmetric_Clifford}.) 
	Let $\mathcal{G}$ be a subgroup of $\mathcal{U}_N$. 
	The $\mathcal{G}$-symmetric Clifford group $\mathcal{C}_{N, \mathcal{G}}$ is defined by 
	\begin{align}
		\mathcal{C}_{N, \mathcal{G}}:=\mathcal{C}_N\cap\mathcal{U}_{N, \mathcal{G}}, \label{SMeq:C_G_def}
	\end{align}
	where 
	\begin{align}
		\mathcal{U}_{N, \mathcal{G}}:=\{U\in\mathcal{U}_N\ |\ \forall G\in\mathcal{G}\ [U, G]=0\}. 
	\end{align}
	An operator $U\in\mathcal{U}_N$ is called a $\mathcal{G}$-symmetric Clifford operator if $U\in\mathcal{C}_{N, \mathcal{G}}$. 
\end{definition}

This definition includes the conventional Clifford group $\mathcal{C}_N$ as the special case when $\mathcal{G}=\{I\}$. \\

Next, we define the notion of symmetric unitary design. 
For a subgroup $\mathcal{G}$ of $\mathcal{U}_N$ and $t\in\mathbb{N}$, we define $\mathcal{G}$-symmetric unitary $t$-designs as the group that approximate $\mathcal{U}_{N, \mathcal{G}}$.

\begin{definition} \label{SMdef:symmetric_unweighted_unitary_design}
	(Restatement of Definition~\ref{def:symmetric_unweighted_unitary_design}.) 
    Let $t\in\mathbb{N}$ and $\mathcal{G}$ be a subgroup of $\mathcal{U}_N$. 
	A subgroup $\mathcal{X}$ of $\mathcal{U}_N$ is a $\mathcal{G}$-symmetric unitary $t$-design if 
	\begin{align}
		\Phi_{t, \mathcal{X}}=\Phi_{t, \mathcal{U}_{N, \mathcal{G}}}, \label{SMeq:symmetric_unweighted_unitary_design_def}
	\end{align}
	where $\Phi_{t, \mathcal{X}}$ is the $t$-fold twirling channel defined by 
	\begin{align}
		\Phi_{t, \mathcal{X}}:=\int_{U\in\mathcal{X}} \mathcal{E}_{t, U} d\mu_\mathcal{X}(U),  \label{SMeq:Phi_def}
	\end{align}
	with the normalized Haar measure on $\mathcal{X}$ and $\mathcal{E}_{t, U}$ is the $t$-fold unitary conjugation map on $\mathcal{L}(\mathcal{H}^{\otimes t})$ defined by 
	\begin{align}
		\mathcal{E}_{t, U}(L):=U^{\otimes t}LU^{\dag\otimes t}\ \forall L\in\mathcal{L}(\mathcal{H}^{\otimes t}). \label{SMeq:t_fold_unitary_action_def}
	\end{align}
\end{definition}

This definition includes the standard unitary designs as the special case when $\mathcal{G}=\{I\}$. 
This type of unitary designs are sometimes called unweighted unitary designs in comparison with weighted unitary designs, where non-uniform mixtures of $\mathcal{E}_{t, U}$ are considered (see Definition~\ref{SMdef:symmetric_weighted_unitary_design}). 
As we show in Theorem~\ref{SMthm:weighted_unitary_design}, the conditions for a symmetric Clifford group being unweighted and weighted unitary designs are equivalent to each other. 
It is therefore sufficient to focus only on unweighted unitary designs, and we express them simply as unitary designs in the following. \\

We are going to prove that $\mathcal{C}_{N, \mathcal{G}}$ is a $\mathcal{G}$-symmetric unitary 3-design if and only if the symmetry condition can be described by the commutativity with some Pauli subgroup. 
This can be rigorously stated as follows:

\begin{theorem} \label{SMthm:main}
    (Restatement of Theorem~\ref{thm:main}.) 
	Let $\mathcal{G}$ be a subgroup of $\mathcal{U}_N$. 
	Then, $\mathcal{C}_{N, \mathcal{G}}$ is a $\mathcal{G}$-symmetric unitary $3$-design if and only if $\mathcal{U}_{N, \mathcal{G}}=\mathcal{U}_{N, \mathcal{Q}}$ with some subgroup $\mathcal{Q}$ of the Pauli group $\mathcal{P}_N$. 
\end{theorem}

We present the overall structure of the proof of this theorem in Fig.~\ref{SMfig:proof_structure}. 
We prove the ``if'' part in Proposition~\ref{SMprop:if_part}, and the ``only if'' part in Proposition~\ref{SMprop:only_if_part}. \\

	\begin{figure}
	\begin{center}
	\includegraphics[width=0.85\columnwidth]{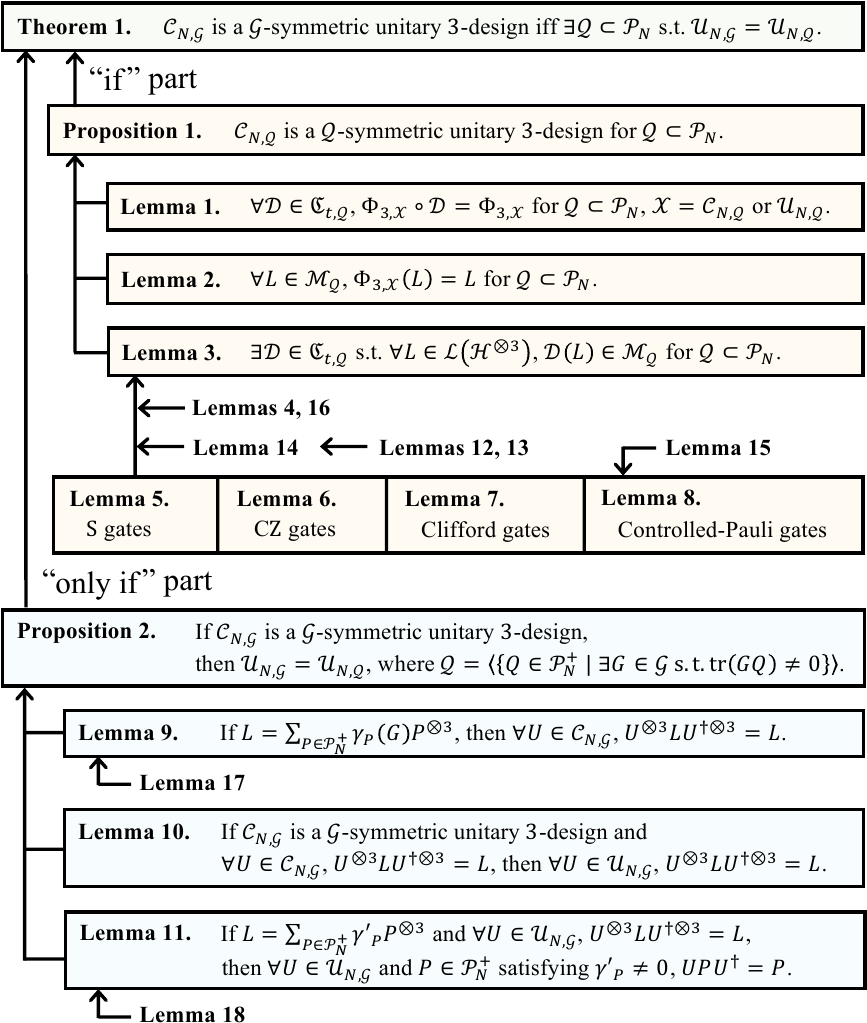}
	\caption{
	Overall structure of the proof of Theorem~\ref{thm:main}. 
	}
	\label{SMfig:proof_structure}
	\end{center}
	\end{figure}

\subsection{Proof of the ``if'' part of Theorem~\ref{thm:main} (Theorem~\ref{SMthm:main})} \label{SMsec:if_part}

The ``if'' part of Theorem~\ref{SMthm:main} is equivalent to the statement that $\mathcal{C}_{N, \mathcal{Q}}$ is a $\mathcal{Q}$-symmetric unitary $3$-design for all Pauli subgroups $\mathcal{Q}$, which we show in the following proposition. 
This is because if $\mathcal{U}_{N, \mathcal{G}}=\mathcal{U}_{N, \mathcal{Q}}$, then $\mathcal{C}_{N, \mathcal{G}}=\mathcal{C}_{N, \mathcal{Q}}$ and $\mathcal{G}$-symmetric unitary $3$-designs are the same as $\mathcal{Q}$-symmetric unitary $3$-designs.

\begin{proposition} \label{SMprop:if_part}
	Let $\mathcal{Q}$ be a subgroup of $\mathcal{P}_N$. 
	Then, $\mathcal{C}_{N, \mathcal{Q}}$ is a $\mathcal{Q}$-symmetric unitary $3$-design. 
\end{proposition}

By the definition of unitary $t$-designs, the goal is to prove $\Phi_{\mathcal{C}_{N, \mathcal{Q}}}=\Phi_{\mathcal{U}_{N, \mathcal{Q}}}$. 
In the following, we introduce two useful properties of $\Phi_{t, \mathcal{X}}$ for $\mathcal{X}=\mathcal{C}_{N, \mathcal{Q}}$ and $\mathcal{U}_{N, \mathcal{Q}}$. \\

As the first property, we cannot distinguish whether there is a symmetric unitary conjugation action before the action of $\Phi_{t, \mathcal{X}}$. 
This can be directly proven by the right invariance of $\mu_\mathcal{X}$. 
We also use this lemma in the proofs of Theorems~\ref{SMthm:1design} and \ref{SMthm:weighted_unitary_design}.

\begin{lemma} \label{SMlem:uniform_mixture_property1}
	Let $N\in\mathbb{N}$, $t\in\mathbb{N}$, $\mathcal{G}$ be a subgroup of $\mathcal{U}_N$, $\mathcal{X}=\mathcal{C}_{N, \mathcal{G}}$ or $\mathcal{U}_{N, \mathcal{G}}$, $\Phi_{t, \mathcal{X}}$ be defined by Eq.~\eqref{SMeq:Phi_def}, and $\mathfrak{C}_{t, \mathcal{G}}$ be the set of all $t$-fold $\mathcal{G}$-symmetric Clifford conjugation mixture maps defined by 
	\begin{align}
		\mathfrak{C}_{t, \mathcal{G}}:=\left\{\sum_{j=1}^n \lambda_j\mathcal{E}_{t, U_j}\ \middle|\ n\in\mathbb{N},\ \lambda_1, \lambda_2, ..., \lambda_n\in\mathbb{R},\ U_1, U_2, ..., U_n\in\mathcal{C}_{N, \mathcal{G}},\ \sum_{j=1}^n \lambda_j=1\right\}. \label{SMeq:symmetric_Clifford_mixture_def}
  	\end{align}
	Then, for any $t$-fold $\mathcal{G}$-symmetric Clifford mixture map $\mathcal{D}\in\mathfrak{C}_{t, \mathcal{G}}$, 
	\begin{align}
		\Phi_{t, \mathcal{X}}\circ\mathcal{D}=\Phi_{t, \mathcal{X}}. 
	\end{align}
\end{lemma}

\begin{proof}
	Since $\mathcal{D}\in\mathfrak{C}_{t, \mathcal{G}}$, $\mathcal{D}$ can be written as $\mathcal{D}=\sum_{j=1}^n \lambda_j\mathcal{E}_{t, U_j}$ with some $\lambda_1, \lambda_2, ..., \lambda_n\in\mathbb{R}$ and $U_1, U_2, ..., U_n\in\mathcal{C}_{N, \mathcal{G}}$ satisfying $\sum_{j=1}^n \lambda_j=1$. 
	For any $j\in\{1, 2, ..., n\}$, we get 
	\begin{align}
		\Phi_{t, \mathcal{X}}\circ\mathcal{E}_{t, U_j} 
		=\int_{U\in\mathcal{X}} \mathcal{E}_{t, U}\circ\mathcal{E}_{t, U_j} d\mu_\mathcal{X}(U) 
		=\int_{U\in\mathcal{X}} \mathcal{E}_{t, UU_j} d\mu_\mathcal{X}(U) 
		=\int_{U\in\mathcal{X}} \mathcal{E}_{t, U} d\mu_\mathcal{X}(U) 
		=\Phi_{t, \mathcal{X}}, 
	\end{align}
	where we used the right invariance of $\mu_\mathcal{X}$. 
	We note that the Haar measure on a compact Lie group $\mathcal{X}$ is right-invariant by Corollary~8.31 of Ref.~\cite{knapp2002lie}.
	We therefore get 
	\begin{align}
		\Phi_{t, \mathcal{X}}\circ\mathcal{D} 
		=\sum_{j=1}^n \lambda_j\Phi_{t, \mathcal{X}}\circ\mathcal{E}_{t, U_j} 
		=\sum_{j=1}^n \lambda_j\Phi_{t, \mathcal{X}} 
		=\Phi_{t, \mathcal{X}}. 
	\end{align}
\end{proof}

As the second property, we introduce trivial fixed-points of $\Phi_{t, \mathcal{X}}$ for $\mathcal{X}=\mathcal{C}_{N, \mathcal{Q}}$ and $\mathcal{U}_{N, \mathcal{Q}}$ in an explicit form.

\begin{lemma} \label{SMlem:fixed_points_specific}
	Let $N\in\mathbb{N}$, $\mathcal{G}$ be a subgroup of $\mathcal{U}_N$, $\mathcal{X}=\mathcal{C}_{N, \mathcal{G}}$ or $\mathcal{U}_{N, \mathcal{G}}$, $\Phi_{3, \mathcal{X}}$ be defined by Eq.~\eqref{SMeq:Phi_def}, and $\mathcal{M}_\mathcal{G}$ be the linear subspace of $\mathcal{L}(\mathcal{H}^{\otimes 3})$ defined by 
	\begin{align}
		\mathcal{M}_\mathcal{G}:=\mathrm{span}\left(\left\{V_\sigma(G_1\otimes G_2\otimes G_3)\ |\ \sigma\in\mathfrak{S}_3,\ G_1, G_2, G_3\in\mathcal{G}\right\}\right), \label{SMeq:invariant_subspace}
	\end{align}
	where the span is taken over the field $\mathbb{C}$ and $V_\sigma\in\mathcal{U}(\mathcal{H}^{\otimes 3})$ is the permutation operator that brings the $j$th copy of qubits to the $\sigma(j)$th qubits, i.e., 
	\begin{align}
		V_\sigma\left(\ket{\Psi_1}\otimes\ket{\Psi_2}\otimes\ket{\Psi_3}\right)
		:=\ket{\Psi_{\sigma^{-1}(1)}}\otimes\ket{\Psi_{\sigma^{-1}(2)}}\otimes\ket{\Psi_{\sigma^{-1}(3)}} \label{SMeq:total_permutation}
	\end{align}
	for all $\ket{\Psi_1}, \ket{\Psi_2}, \ket{\Psi_3}\in\mathcal{H}$. 
	Then, $\Phi_{3, \mathcal{X}}(L)=L$ for all $L\in\mathcal{M}_\mathcal{G}$. 
\end{lemma}

As for Eq.~\eqref{SMeq:total_permutation}, we note that the state of the $j$th copy of qubits after the action of $V_\sigma$ is the same as that of the $\sigma^{-1}(j)$th copy of qubits before the action, because $V_\sigma$ brings $\sigma^{-1}(j)$th copy of qubits to the $j$th copy of qubits.

\begin{proof}
	Since $\Phi_{3, \mathcal{X}}$ is a linear map and $\mathcal{M}_\mathcal{G}$ is a linear subspace spanned by $\{V_\sigma(G_1\otimes G_2\otimes G_3)\ |\ \sigma\in\mathfrak{S}_3, G_1, G_2, G_3\in\mathcal{G}\}$, it is sufficient to show that $\Phi_{3, \mathcal{X}}(V_\sigma(G_1\otimes G_2\otimes G_3))=V_\sigma(G_1\otimes G_2\otimes G_3)$ for all $\sigma\in\mathfrak{S}_3$, $G_1, G_2, G_3\in\mathcal{G}$. 
	Since $V_\sigma$ commutes with $U^{\otimes 3}$, and $G_1$, $G_2$ and $G_3$ commute with $U$ for all $U\in\mathcal{U}_{N, \mathcal{G}}$, we have $[V_\sigma(G_1\otimes G_2\otimes G_3), U^{\otimes 3}]=0$ for all $U\in\mathcal{U}_{N, \mathcal{G}}$. 
	We therefore get 
	\begin{align}
		\Phi_{3, \mathcal{X}}(V_\sigma(G_1\otimes G_2\otimes G_3))
		=&\int_{U\in\mathcal{X}} U^{\otimes 3}V_\sigma(G_1\otimes G_2\otimes G_3)U^{\dag\otimes 3} d\mu_\mathcal{X}(U) \nonumber\\
		=&\int_{U\in\mathcal{X}} V_\sigma(G_1\otimes G_2\otimes G_3) d\mu_\mathcal{X}(U) \nonumber\\
		=&V_\sigma(G_1\otimes G_2\otimes G_3). 
	\end{align}
\end{proof}

Although we only require the fact that all the points in $\mathcal{M}_\mathcal{G}$ are fixed-points of $\Phi_{3, \mathcal{U}_{N, \mathcal{G}}}$ in our proof, we can prove that the set of all the fixed-points of $\Phi_{3, \mathcal{U}_{N, \mathcal{G}}}$ corresponds with $\mathcal{M}_{\mathcal{U}_{N, \mathcal{U}_{N, \mathcal{G}}}}$ by using the result of Ref.~\cite{marvian2014generalization}. \\

In order to connect Lemmas~\ref{SMlem:uniform_mixture_property1} and \ref{SMlem:fixed_points_specific} to the proof of Proposition~\ref{SMprop:if_part}, it is sufficient to find a map $\mathcal{D}\in\mathfrak{C}_{3, \mathcal{Q}}$ satisfying $\mathcal{D}(L)\in\mathcal{M}_\mathcal{Q}$ for all $L\in\mathcal{L}(\mathcal{H}^{\otimes 3})$. 
If we can construct such a map $\mathcal{D}$, we can explicitly compute the $t$-fold uniform unitary mixture map $\Phi_{t, \mathcal{X}}$ as $\Phi_{3, \mathcal{X}}=\Phi_{3, \mathcal{X}}\circ\mathcal{D}=\mathcal{D}$ for $\mathcal{X}=\mathcal{C}_{N, \mathcal{Q}}$ and $\mathcal{U}_{N, \mathcal{Q}}$, which implies that $\Phi_{3, \mathcal{C}_{N, \mathcal{Q}}}=\Phi_{3, \mathcal{U}_{N, \mathcal{Q}}}$. 
We present the existence of such a map $\mathcal{D}$ as a lemma.

\begin{lemma} \label{SMlem:symmetrization}
	Let $N\in\mathbb{N}$, $\mathcal{Q}$ be a subgroup of $\mathcal{P}_N$, $\mathfrak{C}_{3, \mathcal{Q}}$ be the set all $t$-fold $\mathcal{Q}$-symmetric Clifford conjugation mixture maps defined by Eq.~\eqref{SMeq:symmetric_Clifford_mixture_def} and $\mathcal{M}_\mathcal{Q}$ be defined by Eq.~\eqref{SMeq:invariant_subspace}. 
	Then, there exists a map $\mathcal{D}\in\mathfrak{C}_{3, \mathcal{Q}}$ such that $\mathcal{D}(L)\in\mathcal{M}_\mathcal{Q}$ for all $L\in\mathcal{L}(\mathcal{H}^{\otimes 3})$. 
\end{lemma}

In order to simplify the proof of this lemma, we show that the statements of this lemma for two symmetry groups are equivalent if the two groups can be transformed into each other by some Clifford conjugation action up to phase.

\begin{lemma} \label{SMlem:unitary_design_equivalence}
	Let $N\in\mathbb{N}$, $\mathcal{G}$ and $\mathcal{G}'$ be subgroups of $\mathcal{U}_N$ satisfying $\mathcal{U}_0\mathcal{G}'=\mathcal{U}_0 W\mathcal{G}W^\dag$ with some $W\in\mathcal{C}_N$, and $\mathfrak{C}_{t, \mathcal{G}}$ and $\mathfrak{C}_{t, \mathcal{G}'}$ be the sets of all $t$-fold $\mathcal{G}$- and $\mathcal{G}'$-symmetric Clifford conjugation mixture maps defined by Eq.~\eqref{SMeq:symmetric_Clifford_mixture_def}. 
	Then, there exists a map $\mathcal{D}'\in\mathfrak{C}_{3, \mathcal{G}'}$ such that $\mathcal{D}'(L)\in\mathcal{M}_{\mathcal{G}'}$ for all $L\in\mathcal{L}(\mathcal{H}^{\otimes 3})$ if and only if there exists a map $\mathcal{D}\in\mathfrak{C}_{3, \mathcal{G}}$ such that $\mathcal{D}(L)\in\mathcal{M}_{\mathcal{G}}$ for all $L\in\mathcal{L}(\mathcal{H}^{\otimes 3})$. 
\end{lemma}

\begin{proof}
	Since $\mathcal{U}_0\mathcal{G}'=\mathcal{U}_0 W\mathcal{G}W^\dag$ is equivalent to $\mathcal{U}_0\mathcal{G}=\mathcal{U}_0 W^\dag\mathcal{G}'W$, it is sufficient only to prove the ``if'' part. 
	We suppose that there exists a map $\mathcal{D}\in\mathfrak{C}_{3, \mathcal{G}}$ such that $\mathcal{D}(L)\in\mathcal{M}_\mathcal{G}$ for all $L\in\mathcal{L}(\mathcal{H}^{\otimes 3})$. 
	By the definition of $\mathfrak{C}_{3, \mathcal{G}}$, $\mathcal{D}$ can be written as 
	\begin{align}
		\mathcal{D}=\sum_{j=1}^n \lambda_j\mathcal{E}_{3, U_j} 
	\end{align}
	with some $n\in\mathbb{N}$, $U_1, U_2, ..., U_n\in\mathcal{C}_{N, \mathcal{G}}$, and $\lambda_1, \lambda_2, ..., \lambda_n\in\mathbb{R}$ satisfying $\sum_{j=1}^n \lambda_j=1$. 
	We define 
	\begin{align}
		\mathcal{D}':=\sum_{j=1}^n \lambda_j\mathcal{E}_{3, W U_j W^\dag}. 
	\end{align}
	Then, $\mathcal{D}'\in\mathfrak{C}_{3, \mathcal{G}}$ by noting that $WU_j W^\dag\in\mathcal{C}_{N, \mathcal{G}'}$ for all $j\in\{1, 2, ..., n\}$. 
	By this definition, we also know that for any $L\in\mathcal{L}(\mathcal{H}^{\otimes t})$, 
	\begin{align}
		\mathcal{D}'(L)
		=&W^{\otimes 3}\mathcal{D}(W^{\dag\otimes 3}LW^{\otimes 3})W^{\dag\otimes 3} \nonumber\\
		\in& W^{\otimes 3}\mathcal{M}_\mathcal{G}W^{\dag\otimes 3} \nonumber\\
		=&\mathrm{span}\left(\left\{W^{\otimes 3}V_\sigma W^{\dag\otimes 3}(WG_1 W^\dag\otimes WG_2 W^\dag\otimes WG_3 W^\dag)\ |\ \sigma\in\mathfrak{S}_3, G_1, G_2, G_3\in\mathcal{G})\right\}\right) \nonumber\\
		=&\mathrm{span}\left(\left\{V_\sigma(G'_1\otimes G'_2\otimes G'_3)\ |\ \sigma\in\mathfrak{S}_3, G'_1, G'_2, G'_3\in\mathcal{G}')\right\}\right) \nonumber\\
		=&\mathcal{M}_{\mathcal{G}'}. 
	\end{align}
\end{proof}

Now we note that any Pauli subgroup $\mathcal{Q}$ can generally be transformed into a Pauli subgroup $\mathcal{R}$ in the form of 
\begin{align}
	\mathcal{R}:=\mathcal{P}_0\{\mathrm{I}, \mathrm{X}, \mathrm{Y}, \mathrm{Z}\}^{\otimes N_1}\otimes\{\mathrm{I}, \mathrm{Z}\}^{\otimes N_2}\otimes\{\mathrm{I}\}^{N_3} \label{SMeq:Pauli_subgroup_standard}
\end{align}
with some $N_1, N_2, N_3\geq0$ up to phase via some Clifford conjugation action, which we prove in Lemma~\ref{SMlem:Pauli_subgroup_equivalence} in Appendix~\ref{SMsec:technical}. 
By combining this property and Lemma~\ref{SMlem:unitary_design_equivalence}, we know that it is sufficient only to prove Lemma~\ref{SMlem:symmetrization} when $\mathcal{Q}$ is given as $\mathcal{R}$ in the form of Eq.~\eqref{SMeq:Pauli_subgroup_standard}.

Since we are going to deal with three copies of the system each of which is decomposed into three subsystems, we define the notations for explicit presentation of the Hilbert space on which an operator acts or in which a vector exists. 
When we explicitly show that an operator $O$ acts on a Hilbert space $\mathcal{K}$ and a vector $\ket{\Psi}$ exists in $\mathcal{K}$, we denote $O^{(\mathcal{K})}$ and $\ket{\Psi}^{(\mathcal{K})}$, respectively. 
The notations for Hilbert spaces are as follows: 
In order to distinguish the Hilbert spaces $\mathcal{H}$ associated with the $3$ copies of the $N$ qubits that we consider in the context of unitary $3$-designs, we denote the $3$ Hilbert spaces by $\mathcal{H}^1$, $\mathcal{H}^2$ and $\mathcal{H}^3$ (see Fig.~\ref{SMfig:space_decomposition} (a)). 
The symmetry $\mathcal{R}$ induces a natural decomposition of each Hilbert space $\mathcal{H}^j$ into three parts $\mathcal{H}_1^j$, $\mathcal{H}_2^j$ and $\mathcal{H}_3^j$ of $N_1$ $N_2$ and $N_3$ qubits, correspondingly to the representation of $\mathcal{R}$. 
We also denote the Hilbert space of the $l$th qubit in $\mathcal{H}_k^j$ by $\mathcal{H}_{k, l}^j$. (see Fig.~\ref{SMfig:space_decomposition} (b)). 
We denote the tensor product of the three spaces of $\mathcal{H}_k^1$, $\mathcal{H}_k^2$ and $\mathcal{H}_k^3$ by $\mathcal{H}_k^\mathrm{tot}$ (see Fig.~\ref{SMfig:space_decomposition} (c)). 
We may refer to $\mathcal{H}_k^j$ simply as $\mathcal{H}_k$ when we need not specify $j$.

In the proof of Lemma~\ref{SMlem:symmetrization}, we focus on the following four types of $\mathcal{R}$-symmetric Clifford operators; the S gates on a qubit in $\mathcal{H}_2$, the controlled-Z gates on two qubits in $\mathcal{H}_2$, the Clifford gates on qubits in $\mathcal{H}_3$, and the controlled-Pauli gates with a control qubit in $\mathcal{H}_2$ and target qubits in $\mathcal{H}_3$ (see Fig.~\ref{SMfig:space_decomposition} (d)). 
We are going to see their properties one by one in the four lemmas below.

	\begin{figure}
	\begin{center}
	\includegraphics[width=0.85\columnwidth]{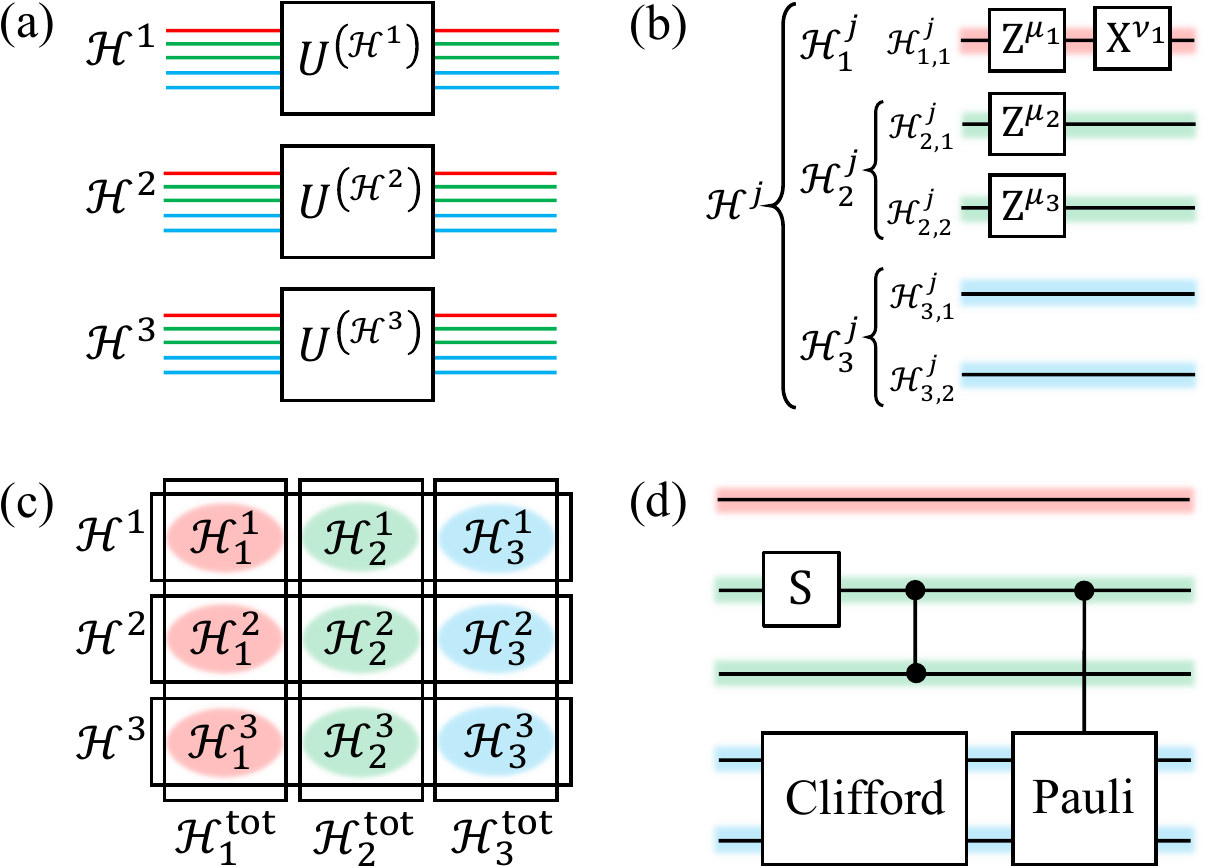}
	\caption{Setup of Proposition~\ref{SMprop:if_part} and the notations of the Hilbert spaces in the proof. 
	(a) In the proof of unitary $3$-designs, we consider unitary operations $U$ on $3$ copies of a Hilbert space, which we denote by $\mathcal{H}^1$, $\mathcal{H}^2$ and $\mathcal{H}^3$. 
	When we explicitly show that a unitary operator $U$ acts on $\mathcal{H}^j$, we denote $U^{(\mathcal{H}^j)}$. 
	(b) The symmetry $\mathcal{R}$ decomposes each Hilbert space $\mathcal{H}^j$ into three parts; $\mathcal{H}_1^j$, $\mathcal{H}_2^j$ and $\mathcal{H}_3^j$ for the $N_1$, $N_2$ and $N_3$ qubits, correspondingly to the representation of $\mathcal{R}$. 
	The figure is for the case when $N_1=1$, $N_2=2$ and $N_3=2$. 
	We denote the Hilbert space for the $l$th qubit in $\mathcal{H}_k^j$ by $\mathcal{H}_{k, l}^j$. 
	(c) We define $\mathcal{H}_k^\mathrm{tot}:=\mathcal{H}_k^1 \otimes\mathcal{H}_k^2\otimes\mathcal{H}_k^3$ for $k=1, 2, 3$. 
	We note that the total Hilbert space $\mathcal{H}^{\otimes 3}$ can be expressed in two ways as $\mathcal{H}^1\otimes\mathcal{H}^2\otimes\mathcal{H}^3$ and as $\mathcal{H}_1^\mathrm{tot}\otimes\mathcal{H}_2^\mathrm{tot}\otimes\mathcal{H}_3^\mathrm{tot}$. 
	(d) As $\mathcal{R}$-symmetric Clifford operations, we focus on four types of gates; the S gates on qubits in $\mathcal{H}_2$, the controlled-Z gates on two qubits in $\mathcal{H}_2$, the Clifford gates on qubits in $\mathcal{H}_3$, and the controlled-Pauli gates with a control qubit in $\mathcal{H}_2$ and target qubits in $\mathcal{H}_3$. 
	}
	\label{SMfig:space_decomposition}
	\end{center}
	\end{figure}

First, we prove the property of the mixture of the $\mathrm{S}$ gates on $\mathcal{H}_2$. 
In the following, we denote the Pauli Z-basis of the qubit in $\mathcal{H}_{2, l}^j$ by $\ket{x_{j, l}}$ with $x_{j, l}\in\{0, 1\}$, and define $\ket{\bm{x}}:=\bigotimes_{j\in\{1, 2, 3\}, l\in\{1, 2, ..., N_2\}, } \ket{x_{j, l}}$ for $\bm{x}:=(x_{j, l})_{j\in\{1, 2, 3\}, l\in\{1, 2, ..., N_2\}}$.

\begin{lemma} \label{SMlem:S_gates}
	Let $N\in\mathbb{N}$, $\mathcal{R}$ be defined by Eq.~\eqref{SMeq:Pauli_subgroup_standard}, $\mathfrak{C}_{3, \mathcal{R}}$ be the set of all $3$-fold $\mathcal{R}$-symmetric Clifford conjugation mixture maps defined by Eq.~\eqref{SMeq:symmetric_Clifford_mixture_def}, $m\in\{1, 2, ..., N_2\}$, $\mathcal{D}_{1, m}$ be defined by 
	\begin{align}
		\mathcal{D}_{1, m}(L):=\frac{1}{4}\sum_{\mu=0}^3 \left(\bigotimes_{j=1}^3 
		{\mathrm{S}^\mu}^{\left(\mathcal{H}_{2, m}^j\right)}\right) L\left(\bigotimes_{j=1}^3 {\mathrm{S}^{\mu\dag}}^{\left(\mathcal{H}_{2, m}^j\right)}\right) \label{SMeq:S^B_mixture_def}
	\end{align}
	and $K\in\mathcal{L}(\mathcal{H}^{\otimes 3})$ be in the form of 
	\begin{align}
		K=P^{\left(\mathcal{H}_1^\mathrm{tot}\right)}\otimes\ket{\bm{x}}\bra{\bm{y}}^{\left(\mathcal{H}_2^\mathrm{tot}\right)}\otimes O^{\left(\mathcal{H}_3^\mathrm{tot}\right)} \label{SMeq:basis_total}
	\end{align}
	with some $P\in\mathcal{P}_{3N_1}^+$, $\bm{x}, \bm{y}\in\{0, 1\}^{3N_2}$ and $O\in\mathcal{L}(\mathcal{H}_3^\mathrm{tot})$. 
	Then, we have $\mathcal{D}_{1, m}\in\mathfrak{C}_{3, \mathcal{R}}$ and 
	\begin{align}
		\mathcal{D}_{1, m}(K)
		=\left\{
		\begin{array}{ll}
			K\ &(\textrm{if}\ \sum_{j=1}^3 x_{j, m}=\sum_{j=1}^3 y_{j, m})\\
			0\ &(\textrm{otherwise}). 
		\end{array} 
		\right. 
	\end{align}
\end{lemma}

\begin{proof}
	Since $\mathrm{S}^{\left(\mathcal{H}_{2, m}^j\right)}\in\mathcal{C}_{N, \mathcal{R}}$ and $\mathcal{D}_{1, m}$ is an affine combination of the conjugation actions of $I^{\otimes 3}$, $\mathrm{S}^{\otimes 3}$, $(\mathrm{S}^2)^{\otimes 3}$ and $(\mathrm{S}^3)^{\otimes 3}$, we can confirm that $\mathcal{D}_{1, m}\in\mathfrak{C}_{3, \mathcal{R}}$. 
	By noting that $\mathrm{S}=\sum_{z\in\{0, 1\}} i^z\ket{z}\bra{z}$, we get 
	\begin{align}
		\mathcal{D}_{1, m}(K)
		=&\frac{1}{4}\sum_{\mu=0}^3 P^{\left(\mathcal{H}_1^\mathrm{tot}\right)}\otimes \left(i^{\mu\sum_{j=1}^3 x_{j, m}}\ket{\bm{x}}\bra{\bm{y}}^{\left(\mathcal{H}_2^\mathrm{tot}\right)}i^{-\mu\sum_{j=1}^3 y_{j, m}}\right)\otimes O^{\left(\mathcal{H}_3^\mathrm{tot}\right)} \nonumber\\
		=&\left(\frac{1}{4}\sum_{\mu=0}^3 i^{\mu(\sum_{j=1}^3 x_{j, m}-\sum_{j=1}^3 y_{j, m})}\right) P^{\left(\mathcal{H}_1^\mathrm{tot}\right)}\otimes\ket{\bm{x}}\bra{\bm{y}}^{\left(\mathcal{H}_2^\mathrm{tot}\right)}\otimes O^{\left(\mathcal{H}_3^\mathrm{tot}\right)} \nonumber\\
		=&\left\{
		\begin{array}{ll}
			K\ &(\textrm{if}\ \sum_{j=1}^3 x_{j, m}\equiv\sum_{j=1}^3 y_{j, m}\ (\mathrm{mod}\ 4))\\
			0\ &(\textrm{otherwise}). 
		\end{array} 
		\right. 
	\end{align}
	Since $x_{j, m}$ and $y_{j, m}$ take only $0$ or $1$, $\sum_{j=1}^3 x_{j, m}\equiv\sum_{j=1}^3 y_{j, m}$ (mod $4$) is equivalent to $\sum_{j=1}^3 x_{j, m}=\sum_{j=1}^3 y_{j, m}$. 
\end{proof}

Second, we prove the property of the mixture of the controlled-Z gates on $\mathcal{H}_2$.

\begin{lemma} \label{SMlem:CZ^BB_mixture}
	Let $N\in\mathbb{N}$, $\mathcal{R}$ be defined by Eq.~\eqref{SMeq:Pauli_subgroup_standard}, $\mathfrak{C}_{3, \mathcal{R}}$ be the set of all $3$-fold $\mathcal{R}$-symmetric Clifford conjugation mixture maps defined by Eq.~\eqref{SMeq:symmetric_Clifford_mixture_def}, $m, m'\in\{1, 2, ..., N_2\}$ satisfy $m\neq m'$, $\mathcal{D}_{2, m, m'}$ be defined by 
	\begin{align}
		\mathcal{D}_{2, m, m'}(L):=\frac{1}{2}\sum_{\nu=0}^1 \left(\bigotimes_{j=1}^3 {\mathrm{CZ}^\nu}^{\left(\mathcal{H}_{2, m}^j, \mathcal{H}_{2, m'}^j\right)}\right) L\left(\bigotimes_{j=1}^3 {\mathrm{CZ}^{\nu\dag}}^{\left(\mathcal{H}_{2, m}^j, \mathcal{H}_{2, m'}^j\right)}\right) \label{SMeq:CZ^B_mixture_def}
	\end{align}
	for all $L\in\mathcal{L}(\mathcal{H}^{\otimes 3})$, and $K\in\mathcal{L}(\mathcal{H}^{\otimes 3})$ be in the form of Eq.~\eqref{SMeq:basis_total} with some $P\in\mathcal{P}_{3N_1}^+$, $\bm{x}, \bm{y}\in\{0, 1\}^{3N_2}$ and $O\in\mathcal{L}(\mathcal{H}_3^\mathrm{tot})$. 
	Then, $\mathcal{D}_{2, m, m'}\in\mathfrak{C}_{3, \mathcal{R}}$ and 
	\begin{align}
		\mathcal{D}_{2, m ,m'}(K)=
		\left\{
		\begin{array}{ll}
			K\ &(\textrm{if}\ \sum_{j=1}^3 x_{j, m} x_{j, m'} \equiv\sum_{j=1}^3 y_{j, m} y_{j, m'}\ (\mathrm{mod}\ 2))\\
			0\ &(\textrm{otherwise}). 
		\end{array}
		\right. 
	\end{align}
\end{lemma}

\begin{proof}
	Since $\mathrm{CZ}^{\left(\mathcal{H}_{2, m}^j, \mathcal{H}_{2, m'}^j\right)}\in\mathcal{C}_{N, \mathcal{R}}$ and $\mathcal{D}_{2, m, m'}$ is an affine combination of the conjugation actions of $I^{\otimes 3}$ and $\mathrm{CZ}^{\otimes 3}$, we can confirm that $\mathcal{D}_{2, m, m'}\in\mathfrak{C}_{3, \mathcal{R}}$. 
	By noting that $\mathrm{CZ}=\sum_{z, w\in\{0, 1\}} (-1)^{zw}\ket{zw}\bra{zw}$, we get 
	\begin{align}
		\mathcal{D}_{2, m, m'}(K)
		=&\frac{1}{2}\sum_{\nu=0}^1 P^{\left(\mathcal{H}_1^\mathrm{tot}\right)}\otimes \left(i^{\nu\sum_{j=1}^3 x_{j, m}x_{j, m'}}\ket{\bm{x}}\bra{\bm{y}}^{\left(\mathcal{H}_2^\mathrm{tot}\right)}i^{-\nu\sum_{j=1}^3 y_{j, m}y_{j, m'}}\right)\otimes O^{\left(\mathcal{H}_3^\mathrm{tot}\right)} \nonumber\\
		=&\left(\frac{1}{2}\sum_{\nu=0}^1 i^{\nu(\sum_{j=1}^3 x_{j, m}x_{j, m'}-\sum_{j=1}^3 y_{j, m}y_{j, m'})}\right) P^{\left(\mathcal{H}_1^\mathrm{tot}\right)}\otimes\ket{\bm{x}}\bra{\bm{y}}^{\left(\mathcal{H}_2^\mathrm{tot}\right)}\otimes O^{\left(\mathcal{H}_3^\mathrm{tot}\right)} \nonumber\\
		=&\left\{
		\begin{array}{ll}
			K\ &(\textrm{if}\ \sum_{j=1}^3 x_{j, m}x_{j, m'}\equiv\sum_{j=1}^3 y_{j, m}y_{j, m'}\ (\mathrm{mod}\ 2))\\
			0\ &(\textrm{otherwise}). 
		\end{array} 
		\right. 
	\end{align}
\end{proof}

Third, we prove the property of the mixture of the Clifford gates on $\mathcal{H}_3$. 
We use the result of Refs.~\cite{webb2016clifford, zhu2017multiqubit} stating that the conventional Clifford group is a unitary $3$-design.

\begin{lemma} \label{SMlem:Clifford_3_mixture}
	Let $N\in\mathbb{N}$, $\mathcal{R}$ be defined by Eq.~\eqref{SMeq:Pauli_subgroup_standard}, $\mathfrak{C}_{3, \mathcal{R}}$ be the set of all $3$-fold $\mathcal{R}$-symmetric Clifford conjugation mixture maps defined by Eq.~\eqref{SMeq:symmetric_Clifford_mixture_def}, $\mathcal{D}_3$ be defined by 
	\begin{align}
		\mathcal{D}_3(L):=\frac{1}{|\mathcal{V}|}\sum_{U\in\mathcal{V}} \left(U^{(\mathcal{H}_3^1)}\otimes U^{(\mathcal{H}_3^2)}\otimes U^{(\mathcal{H}_3^3)}\right)L\left(U^{\dag(\mathcal{H}_3^1)}\otimes U^{\dag(\mathcal{H}_3^2)}\otimes U^{\dag(\mathcal{H}_3^3)}\right) \label{SMeq:Clifford^C_mixture_def}
	\end{align}
	for all $L\in\mathcal{L}(\mathcal{H}^{\otimes 3})$ with $\mathcal{V}$ being the set of all the representatives of the equivalence classes in $\mathcal{C}_{N_3}/(\mathcal{U}_0 I)$, and $K\in\mathcal{L}(\mathcal{H}^{\otimes 3})$ be in the form of Eq.~\eqref{SMeq:basis_total} with some $P\in\mathcal{P}_{3N_1}^+$, $\bm{x}, \bm{y}\in\{0, 1\}^{3N_2}$ and $O\in\mathcal{L}(\mathcal{H}_3^\mathrm{tot})$. 
	Then, $\mathcal{D}_3\in\mathfrak{C}_{3, \mathcal{R}}$ and 
	\begin{align}
		\mathcal{D}_3(K)=P^{\left(\mathcal{H}_1^\mathrm{tot}\right)}\otimes\ket{\bm{x}}\bra{\bm{y}}^{\left(\mathcal{H}_2^\mathrm{tot}\right)}\otimes \sum_{\sigma\in\mathfrak{S}_3} \alpha_\sigma T_\sigma^{\left(\mathcal{H}_3^\mathrm{tot}\right)} 
	\end{align}
	with some $\{\alpha_\sigma\}_{\sigma\in\mathfrak{S}_3}\in\mathbb{C}^{\mathfrak{S}_3}$, where $T_\sigma\in\mathcal{U}(\mathcal{H}_3^\mathrm{tot})$ is defined as the permutation operator satisfying 
	\begin{align}
		T_\sigma\left(\ket{\xi_1}^{\left(\mathcal{H}_3^1\right)}\otimes\ket{\xi_2}^{\left(\mathcal{H}_3^2\right)}\otimes\ket{\xi_3}^{\left(\mathcal{H}_3^3\right)}\right) 
		=\ket{\xi_{\sigma^{-1}(1)}}^{\left(\mathcal{H}_3^1\right)}\otimes\ket{\xi_{\sigma^{-1}(2)}}^{\left(\mathcal{H}_3^2\right)}\otimes\ket{\xi_{\sigma^{-1}(3)}}^{\left(\mathcal{H}_3^3\right)} \label{SMeq:permutation_C_def}
	\end{align}
	for all $\ket{\xi_1}, \ket{\xi_2}, \ket{\xi_3}\in\mathcal{H}_3$. 
\end{lemma}

\begin{proof}
	Since $U^{\left(\mathcal{H}_3^j\right)}\in\mathcal{C}_{N, \mathcal{R}}$ for all $U\in\mathcal{V}$ and $\mathcal{D}_3$ is an affine combination of the conjugation actions of $U^{\otimes 3}$ for $U\in\mathcal{V}$, we can confirm that $\mathcal{D}_3\in\mathfrak{C}_{3, \mathcal{R}}$. 
	By the definition of $\mathcal{D}_3$, $\mathcal{D}_3(K)$ is written as 
	\begin{align}
		\mathcal{D}_3(K)=P^{\left(\mathcal{H}_1^\mathrm{tot}\right)}\otimes\ket{\bm{x}}\bra{\bm{y}}^{\left(\mathcal{H}_2^\mathrm{tot}\right)}\otimes O'^{\left(\mathcal{H}_3^\mathrm{tot}\right)} 
	\end{align}
	with 
	\begin{align}
		O':=\int_{U'\in\mathcal{C}_{N_3}} {U'}^{\otimes 3}O{U'}^{\dag\otimes 3} d\mu_{\mathcal{C}_{N_3}}(U'). 
	\end{align}
	Since the Clifford group $\mathcal{C}_{N_3}$ is a unitary $3$-design~\cite{webb2016clifford, zhu2017multiqubit}, $O'$ can also be written as 
	\begin{align}
		O'=\int_{U'\in\mathcal{U}_{N_3}} {U'}^{\otimes 3}O{U'}^{\dag\otimes 3}\ d\mu_{\mathcal{U}_{N_3}}(U'). 
	\end{align}
	For any $U\in\mathcal{U}_{N_3}$, by the left invariance of $\mu_{\mathcal{U}_{N_3}}$, we get 
 	\begin{align}
		U^{\otimes 3}O'U^{\dag\otimes 3}
		=\int_{U'\in\mathcal{U}_{N_3}} (UU')^{\otimes 3}O'(UU')^{\dag\otimes 3}\ d\mu_{\mathcal{U}_{N_3}}(U')
		=\int_{U'\in\mathcal{U}_{N_3}} {U'}^{\otimes 3}O{U'}^{\dag\otimes 3}\ d\mu_{\mathcal{U}_{N_3}}(U')
		=O'. 
	\end{align}
	This implies that $O'$ commutes with $U^{\otimes 3}$ for all $U\in\mathcal{U}_{N_3}$. 
	By Theorem~7.15 of Ref.~\cite{watrous2018theory}, $O'$ can be written as $O'=\sum_{\sigma\in\mathfrak{S}_3} \alpha_\sigma T_\sigma$ with some $\{\alpha_\sigma\}_{\sigma\in\mathfrak{S}_3}\in\mathbb{C}^{\mathfrak{S}_3}$. 
\end{proof}

Finally, we prove the property of the mixture of the controlled-Pauli gates on $\mathcal{H}_2$ and $\mathcal{H}_3$. 
Here we fix the control qubit as the $m$th qubit in $\mathcal{H}_2$.

\begin{lemma} \label{SMlem:CPauli^BC_mixture}
	Let $N\in\mathbb{N}$, $\mathcal{R}$ be defined by Eq.~\eqref{SMeq:Pauli_subgroup_standard}, $\mathfrak{C}_{3, \mathcal{R}}$ be the set of all $3$-fold $\mathcal{R}$-symmetric Clifford conjugation mixture maps defined by Eq.~\eqref{SMeq:symmetric_Clifford_mixture_def}, $m\in\{1, 2, ..., N_2\}$, $\mathcal{D}_{4, m}$ be defined by 
	\begin{align}
		\mathcal{D}_{4, m}(L):=\frac{1}{4^{N_3}}\sum_{Q\in\mathcal{P}_{N_3}^+} \left(\bigotimes_{j=1}^3 \mathrm{C}(Q)^{\left(\mathcal{H}_{2, m}^j, \mathcal{H}_3^j\right)}\right)L\left(\bigotimes_{j=1}^3 \mathrm{C}(Q)^{\dag\left(\mathcal{H}_{2, m}^j, \mathcal{H}_3^j\right)}\right) \label{SMeq:CPauli^BC_mixture_def}
	\end{align}
	for all $L\in\mathcal{L}(\mathcal{H}^{\otimes 3})$, where $\mathrm{C}(Q)^{\left(\mathcal{H}_{2, m}^j, \mathcal{H}_3^j\right)}$ is the controlled-$Q$ operator defined for $Q\in\mathcal{P}_{N_3}^+$ by 
	\begin{align}
		\mathrm{C}(Q)^{\left(\mathcal{H}_{2, m}^j, \mathcal{H}_3^j\right)} 
		=\ket{0}\bra{0}^{\left(\mathcal{H}_{2, m}^j\right)}\otimes I^{\left(\mathcal{H}_3^j\right)}
		+\ket{1}\bra{1}^{\left(\mathcal{H}_{2, m}^j\right)}\otimes Q^{\left(\mathcal{H}_3^j\right)}, 
	\end{align}
	and $J\in\mathcal{L}(\mathcal{H}^{\otimes 3})$ be in the form of 
	\begin{align}
		J=P^{\left(\mathcal{H}_1^\mathrm{tot}\right)}\otimes\ket{\bm{x}}\bra{\bm{y}}^{\left(\mathcal{H}_2^\mathrm{tot}\right)}\otimes T_\sigma^{\left(\mathcal{H}_3^\mathrm{tot}\right)} \label{SMeq:intermediate_form}
	\end{align}
	with some $P\in\mathcal{P}_{3N_1}^+$, $\bm{x}, \bm{y}\in\{0, 1\}^{3N_2}$, and $\sigma\in\mathfrak{S}_3$ satisfying 
	\begin{align}
		&\sum_{j=1}^3 x_{j, l}=\sum_{j=1}^3 y_{j, l}\ \mathrm{for\ all}\ l\in\{1, 2, ..., N_2\}, \label{SMeq:bitstring_cond1}\\
		&\sum_{j=1}^3 x_{j, l} x_{j, l'}\equiv\sum_{j=1}^3 y_{j, l}y_{j, l'}\ (\mathrm{mod}\ 2)\ \mathrm{for\ all}\ l, l'\in\{1, 2, ..., N_2\}, \label{SMeq:bitstring_cond2}
	\end{align}
	and $x_{\sigma(j), m'}=y_{j, m'}$ for all $j\in\{1, 2, 3\}$ and $m'\in\{1, 2, ..., m-1\}$, where $T_\sigma$ is defined by Eq.~\eqref{SMeq:permutation_C_def}. 
	Then, $\mathcal{D}_{4, m}\in\mathfrak{C}_{3, \mathcal{R}}$ and 
	\begin{align}
		\mathcal{D}_{4, m}(J)=c P^{\left(\mathcal{H}_1^\mathrm{tot}\right)}\otimes\ket{\bm{x}}\bra{\bm{y}}^{\left(\mathcal{H}_2^\mathrm{tot}\right)}\otimes T_{\sigma'}^{\left(\mathcal{H}_3^\mathrm{tot}\right)} 
	\end{align}
	with some $c\in\mathbb{R}$ and $\sigma'\in\mathfrak{S}_3$ satisfying $x_{\sigma'(j), m'}=y_{j, m'}$ for all $j\in\{1, 2, 3\}$ and $m'\in\{1, 2, ..., m\}$. 
\end{lemma}

\begin{proof}
	Since the controlled-Pauli operators can be expressed as products of the controlled-X, Y, and Z operators, we can confirm that $\mathrm{C}(Q)^{\left(\mathcal{H}_{2, m}^j, \mathcal{H}_3^j\right)}\in\mathcal{C}_{N, \mathcal{R}}$. 
	$\mathcal{D}_{4, m}$ is an affine combination of the conjugation actions of $\mathrm{C}(Q)^{\otimes 3}$ for $Q\in\mathcal{P}_{N_3}^+$, and thus we can also confirm that $\mathcal{D}_{4, m}\in\mathfrak{C}_{3, \mathcal{R}}$. 
	By noting that 
	\begin{align}
		\mathrm{C}(Q)^{\left(\mathcal{H}_{2, m}^j, \mathcal{H}_3^j\right)}=\sum_{z\in\{0, 1\}} \ket{z}\bra{z}^{\left(\mathcal{H}_{2, m}^j\right)}\otimes {Q^z}^{\left(\mathcal{H}_3^j\right)}, 
	\end{align}
	we know that 
	\begin{align}
		\mathcal{D}_{4, m}(J)
		=\frac{1}{4^{N_3}}\sum_{Q\in\mathcal{P}_{N_3}^+} P^{\left(\mathcal{H}_1^\mathrm{tot}\right)}\otimes\ket{\bm{x}}\bra{\bm{y}}^{\left(\mathcal{H}_2^\mathrm{tot}\right)}\otimes
		\left(\bigotimes_{j=1}^3 \mbox{$Q^{x_{j, m}}$}^{\left(\mathcal{H}_3^j\right)}\right)T_\sigma^{\left(\mathcal{H}_3^\mathrm{tot}\right)}
		\left(\bigotimes_{j=1}^3 \mbox{$Q^{y_{j, m}\dag}$}^{\left(\mathcal{H}_3^j\right)}\right). \label{SMeq:SMlem:CPauli^BC_mixture1}
	\end{align}
	We note that 
	\begin{align}
		\left(\bigotimes_{j=1}^3 \mbox{$Q^{x_{j, m}}$}^{\left(\mathcal{H}_3^j\right)}\right)T_\sigma^{\left(\mathcal{H}_3^\mathrm{tot}\right)}
		=&T_\sigma^{\left(\mathcal{H}_3^\mathrm{tot}\right)}\cdot
		T_{\sigma^{-1}}^{\left(\mathcal{H}_3^\mathrm{tot}\right)}
		\left(\bigotimes_{j=1}^3 \mbox{$Q^{x_{j, m}}$}^{\left(\mathcal{H}_3^j\right)}\right)T_\sigma^{\left(\mathcal{H}_3^\mathrm{tot}\right)} \nonumber\\
		=&T_\sigma^{\left(\mathcal{H}_3^\mathrm{tot}\right)}
		\left(\bigotimes_{j=1}^3 \mbox{$Q^{x_{j, m}}$}^{\left(\mathcal{H}_3^{\sigma^{-1}(j)}\right)}\right) \nonumber\\
		=&T_\sigma^{\left(\mathcal{H}_3^\mathrm{tot}\right)}
		\left(\bigotimes_{j=1}^3 \mbox{$Q^{x_{\sigma(j), m}}$}^{\left(\mathcal{H}_3^j\right)}\right). \label{SMeq:SMlem:CPauli^BC_mixture2}
	\end{align}
	By Eqs.~\eqref{SMeq:SMlem:CPauli^BC_mixture1} and \eqref{SMeq:SMlem:CPauli^BC_mixture2}, we get 
	\begin{align}
		\mathcal{D}_{4, m}(J)
		=&\frac{1}{4^{N_3}}\sum_{Q\in\mathcal{P}_{N_3}^+} P^{\left(\mathcal{H}_1^\mathrm{tot}\right)}\otimes\ket{\bm{x}}\bra{\bm{y}}^{\left(\mathcal{H}_2^\mathrm{tot}\right)}\otimes T_\sigma^{\left(\mathcal{H}_3^\mathrm{tot}\right)}
		\left(\bigotimes_{j=1}^3 \mbox{$Q^{x_{\sigma(j), m}}$}^{\left(\mathcal{H}_3^j\right)}\right)
		\left(\bigotimes_{j=1}^3 \mbox{$Q^{-y_{j, m}}$}^{\left(\mathcal{H}_3^j\right)}\right) \nonumber\\
		=&P^{\left(\mathcal{H}_1^\mathrm{tot}\right)}\otimes\ket{\bm{x}}\bra{\bm{y}}^{\left(\mathcal{H}_2^\mathrm{tot}\right)}\otimes T_\sigma^{\left(\mathcal{H}_3^\mathrm{tot}\right)}\left(\frac{1}{4^{N_3}}\sum_{Q\in\mathcal{P}_{N_3}^+} \bigotimes_{j=1}^3 \mbox{$Q^{x_{\sigma(j), m}-y_{j, m}}$}^{\left(\mathcal{H}_3^j\right)}\right). \label{SMeq:SMlem:CPauli^BC_mixture3}
	\end{align}

	First, we consider the case when $x_{\sigma(j), m}=y_{j, m}$ for all $j\in\{1, 2, 3\}$. 
	In this case, we have 
	\begin{align}
		\frac{1}{4^{N_3}}\sum_{Q\in\mathcal{P}_{N_3}^+} \bigotimes_{j=1}^3 \mbox{$Q^{x_{\sigma(j), m}-y_{j, m}}$}^{\left(\mathcal{H}_3^j\right)} 
		=\frac{1}{4^{N_3}}\sum_{Q\in\mathcal{P}_{N_3}^+} \bigotimes_{j=1}^3 I^{\left(\mathcal{H}_3^j\right)} 
		=I^{\left(\mathcal{H}_3^\mathrm{tot}\right)}.  \label{SMeq:SMlem:CPauli^BC_mixture4}
	\end{align}
	By Eqs.~\eqref{SMeq:SMlem:CPauli^BC_mixture4} and \eqref{SMeq:SMlem:CPauli^BC_mixture4}, we get 
	\begin{align}
		\mathcal{D}_{4, m}(J) 
		=P^{\left(\mathcal{H}_1^\mathrm{tot}\right)}\otimes\ket{\bm{x}}\bra{\bm{y}}^{\left(\mathcal{H}_2^\mathrm{tot}\right)}\otimes T_\sigma^{\left(\mathcal{H}_3^\mathrm{tot}\right)}I^{\left(\mathcal{H}_3^\mathrm{tot}\right)} 
		=c P^{\left(\mathcal{H}_1^\mathrm{tot}\right)}\otimes\ket{\bm{x}}\bra{\bm{y}}^{\left(\mathcal{H}_2^\mathrm{tot}\right)}\otimes T_{\sigma'}^{\left(\mathcal{H}_3^\mathrm{tot}\right)}, 
	\end{align}
	where $c:=1$ and $\sigma':=\sigma$. 
	Since $\sigma'=\sigma$, we get $x_{\sigma'(j), m}=x_{\sigma(j), m}=y_{j, m}$ for all $j\in\{1, 2, 3\}$ and $m'\in\{1, 2, ..., m\}$.

	Next, we consider the case when $x_{\sigma(j), m}\neq y_{j, m}$ for some $j\in\{1, 2, 3\}$. 
	Since Eq.~\eqref{SMeq:bitstring_cond1} is satisfied for $k=m$, we can take $p, q\in\{1, 2, 3\}$ that uniquely satisfy $x_{\sigma(p), m}\neq x_{\sigma(j), m}$ for $j\in\{1, 2, 3\}\backslash\{p\}$ and $y_{q, m}\neq y_{j, m}$ for $j\in\{1, 2, 3\}\backslash\{q\}$. 
    Then, such $p$ and $q$ satisfy $p\neq q$ and $x_{\sigma(p), m}=y_{q, m}$. 
	We note that 
	\begin{align}
		\frac{1}{2^{N_3}}\sum_{Q\in\mathcal{P}_{N_3}^+} \bigotimes_{j=1}^3 {Q^{x_{\sigma(j), m}-y_{j, m}}}^{\left(\mathcal{H}_3^j\right)} 
		=&\frac{1}{2^{N_3}}\sum_{Q\in\mathcal{P}_{N_3}^+} Q^{\left(\mathcal{H}_3^p\right)}\otimes Q^{\left(\mathcal{H}_3^q\right)} \nonumber\\
		=&\bigotimes_{l=1}^{N_3} \left(\frac{1}{2}\sum_{Q\in\mathcal{P}_1^+} Q^{\left(\mathcal{H}_{3, l}^p\right)}\otimes Q^{\left(\mathcal{H}_{3, l}^q\right)}\right) \nonumber\\
		=&\bigotimes_{l=1}^{N_3} \mathrm{SWAP}^{\left(\mathcal{H}_{3, l}^p, \mathcal{H}_{3, l}^q\right)} \nonumber\\
		=&T_{\tau_{p, q}}^{\left(\mathcal{H}_3^\mathrm{tot}\right)}, \label{SMeq:SMlem:CPauli^BC_mixture5}
	\end{align}
	where $\mathrm{SWAP}^{\left(\mathcal{H}_{3, l}^p, \mathcal{H}_{3, l}^q\right)}$ is the SWAP operator between the Hilbert spaces $\mathcal{H}_{3, l}^p$ and $\mathcal{H}_{3, l}^q$, and $\tau_{p, q}\in\mathfrak{S}_3$ is the transposition between $p$ and $q$. 
	By Eqs.~\eqref{SMeq:SMlem:CPauli^BC_mixture3} and \eqref{SMeq:SMlem:CPauli^BC_mixture5}, we get 
	\begin{align}
		\mathcal{D}_{4, m}(J) 
		=\frac{1}{2^{N_3}}P^{\left(\mathcal{H}_1^\mathrm{tot}\right)}\otimes\ket{\bm{x}}\bra{\bm{y}}^{\left(\mathcal{H}_2^\mathrm{tot}\right)}\otimes T_\sigma^{\left(\mathcal{H}_3^\mathrm{tot}\right)}T_{\tau_{p, q}}^{\left(\mathcal{H}_3^\mathrm{tot}\right)} 
		=c P^{\left(\mathcal{H}_1^\mathrm{tot}\right)}\otimes\ket{\bm{x}}\bra{\bm{y}}^{\left(\mathcal{H}_2^\mathrm{tot}\right)}\otimes T_{\sigma'}^{\left(\mathcal{H}_3^\mathrm{tot}\right)}, 
	\end{align}
	where $c:=1/2^{N_3}$ and $\sigma':=\sigma\tau_{p, q}$. 
	For any $m'\in\{1, 2, ..., m-1\}$, since Eq.~\eqref{SMeq:bitstring_cond1} is satisfied for $k=m$ and $k=m'$, and Eq.~\eqref{SMeq:bitstring_cond2} is satisfied for $k=m$ and $k'=m'$, we get $x_{\sigma(p), m'}=x_{\sigma(q), m'}$ by Lemma~\ref{SMlem:3_bit_property} in Appendix~\ref{SMsec:technical}. 
	We therefore get $x_{\sigma'(j), m'}=x_{\sigma(\tau_{p, q}(j)), m'}=x_{\sigma(j), m'}=y_{j, m'}$ for $j\in\{1, 2, 3\}$. 
	As for the conclusion in the case of $m'=m$, since $x_{\sigma'(j), m}=x_{\sigma(\tau_{p, q}(j)), m}\equiv x_{\sigma(p), m}+\delta_{\tau_{p, q}(j), p}=y_{q, m}+\delta_{j, \tau_{p, q}(p)}=y_{q, m}+\delta_{j, q}\equiv y_{j, m}$ (mod $2$), we get $x_{\sigma'(j), m}=y_{j, m}$ for $j\in\{1, 2, 3\}$. 
\end{proof}

By combining the five lemmas above, we prove Lemma~\ref{SMlem:symmetrization}. \\

\noindent
\textit{Proof of Lemma~\ref{SMlem:symmetrization}.}
	By Lemma~\ref{SMlem:Pauli_subgroup_equivalence} in Appendix~\ref{SMsec:technical}, we can take $\mathcal{R}$ in the form of Eq.~\eqref{SMeq:Pauli_subgroup_standard} and $W\in\mathcal{C}_N$ such that $\mathcal{U}_0 W\mathcal{Q}W^\dag=\mathcal{U}_0\mathcal{R}$. 
	By Lemma~\ref{SMlem:unitary_design_equivalence}, it is sufficient to show that there exists a map $\mathcal{D}\in\mathfrak{C}_{3, \mathcal{R}}$ such that $\mathcal{D}(L)\in\mathcal{M}_\mathcal{R}$ for all $L\in\mathcal{L}(\mathcal{H}^{\otimes 3})$. 
	We define $\mathcal{D}$ by 
	\begin{align}
		\mathcal{D}=\mathcal{D}_4\circ\mathcal{D}_3\circ\mathcal{D}_2\circ\mathcal{D}_1, 
	\end{align}
	where $\mathcal{D}_1$, $\mathcal{D}_2$, $\mathcal{D}_3$, and $\mathcal{D}_4$ are defined by Eqs.~\eqref{SMeq:S^B_mixture_def}, \eqref{SMeq:CZ^B_mixture_def}, \eqref{SMeq:Clifford^C_mixture_def}, \eqref{SMeq:CPauli^BC_mixture_def}, and 
	\begin{align}
		&\mathcal{D}_1:=\mathcal{D}_{1, N_2}\circ\cdots\circ\mathcal{D}_{1, 2}\circ\mathcal{D}_{1, 1}, \\
		&\mathcal{D}_2:=\mathcal{D}_{2, N_2-1}\circ\cdots\circ\mathcal{D}_{2, 2}\circ\mathcal{D}_{2, 1}, \\
		&\mathcal{D}_{2, m}:=\mathcal{D}_{2, m, N_2}\circ\cdots\circ\mathcal{D}_{2, m, m+2}\circ\mathcal{D}_{2, m, m+1}, \\
		&\mathcal{D}_4:=\mathcal{D}_{4, N_2}\circ\cdots\circ\mathcal{D}_{4, 2}\circ\mathcal{D}_{4, 1}. 
	\end{align}
	By Lemma~\ref{SMeq:symmetric_Clifford_mixture_composition} in Appendix~\ref{SMsec:technical}, we can confirm that $\mathcal{D}\in\mathfrak{C}_{3, \mathcal{R}}$. 
	Take arbitrary $L\in\mathcal{L}(\mathcal{H}^{\otimes 3})$. 
	Since $\mathcal{D}_1$, $\mathcal{D}_2$, $\mathcal{D}_3$ and $\mathcal{D}_4$ are linear maps, $\mathcal{M}_\mathcal{R}$ is a linear subspace, and $L$ can be written as 
	\begin{align}
		L=\sum_{P\in\mathcal{P}_{3N_1}^+, \bm{x}, \bm{y}\in\{0, 1\}^{3N_2}} P^{\left(\mathcal{H}_1^\mathrm{tot}\right)}\otimes\ket{\bm{x}}\bra{\bm{y}}^{\left(\mathcal{H}_2^\mathrm{tot}\right)}\otimes O_{P, \bm{x}, \bm{y}}^{\left(\mathcal{H}_3^\mathrm{tot}\right)} 
	\end{align}
	with some $O_{P, \bm{x}, \bm{y}}\in\mathcal{L}(\mathcal{H}_3^\mathrm{tot})$, it is sufficient to show that 
	\begin{align}
		\mathcal{D}_4\circ\mathcal{D}_3\circ\mathcal{D}_2\circ\mathcal{D}_1(K)\in\mathcal{M}_\mathcal{R} 
	\end{align}
	for all $K\in\mathcal{L}(\mathcal{H}^{\otimes 3})$ in the form of Eq.~\eqref{SMeq:basis_total} with $P\in\mathcal{P}_{3N_1}^+$, $\bm{x}, \bm{y}\in\{0, 1\}^{3N_2}$ and $O\in\mathcal{L}(\mathcal{H}_3^\mathrm{tot})$. 
	By Lemmas~\ref{SMlem:S_gates}, \ref{SMlem:CZ^BB_mixture} and \ref{SMlem:Clifford_3_mixture}, we get 
	\begin{align}
		\mathcal{D}_3\circ\mathcal{D}_2\circ\mathcal{D}_1(K)=
		\left\{
		\begin{array}{ll}
			P^{\left(\mathcal{H}_1^\mathrm{tot}\right)}\otimes\ket{\bm{x}}\bra{\bm{y}}^{\left(\mathcal{H}_2^\mathrm{tot}\right)}\otimes \sum_{\sigma\in\mathfrak{S}_3} \alpha_\sigma T_\sigma^{\left(\mathcal{H}_3^\mathrm{tot}\right)}\ &(\textrm{if}\ \bm{x}\ \mathrm{and}\ \bm{y}\ \mathrm{satisfy\ Eqs.}~\eqref{SMeq:bitstring_cond1}\  \mathrm{and}\ \eqref{SMeq:bitstring_cond2})\\
			0\ &(\textrm{otherwise}) 
		\end{array}
		\right. 
	\end{align}
	with some $\{\alpha_\sigma\}_{\sigma\in\mathfrak{S}_3}\in\mathbb{C}^{\mathfrak{S}_3}$, where we note that Eq.~\eqref{SMeq:bitstring_cond2} in the case of $k=k'$ can be derived from Eq.~\eqref{SMeq:bitstring_cond1}. 
	Since $\mathcal{D}_4$ is a linear map and $\mathcal{M}_\mathcal{R}$ is a linear subspace, it is sufficient to show that 
	\begin{align}
		\mathcal{D}_4(J)\in\mathcal{M}_\mathcal{R} 
	\end{align}
	for all $J\in\mathcal{L}(\mathcal{H}^{\otimes 3})$ in the form of Eq.~\eqref{SMeq:intermediate_form} with $P\in\mathcal{P}_{3N_1}^+$, $\bm{x}, \bm{y}\in\{0, 1\}^{3N_2}$ satisfying Eqs.~\eqref{SMeq:bitstring_cond1} and \eqref{SMeq:bitstring_cond2} and $\sigma\in\mathfrak{S}_3$. 
	By Lemma~\ref{SMlem:CPauli^BC_mixture}, we know that 
	\begin{align}
		\mathcal{D}_4(J)=c P^{\left(\mathcal{H}_1^\mathrm{tot}\right)}\otimes\ket{\bm{x}}\bra{\bm{y}}^{\left(\mathcal{H}_2^\mathrm{tot}\right)}\otimes T_{\sigma'}^{\left(\mathcal{H}_3^\mathrm{tot}\right)}
	\end{align}
	with $c\in\mathbb{R}$ and $\sigma'\in\mathfrak{S}_3$ satisfying $x_{\sigma'(j), l}=y_{j, l}$ for all $j\in\{1, 2, 3\}$ and $l\in\{1, 2, ..., N_2\}$. 
	We therefore get 
	\begin{align}
		\mathcal{D}_4(J)=&c V_{\sigma'} V_{\sigma'^{-1}}\left(P^{\left(\mathcal{H}_1^\mathrm{tot}\right)}\otimes\ket{\bm{x}}\bra{\bm{y}}^{\left(\mathcal{H}_2^\mathrm{tot}\right)}\otimes T_{\sigma'}^{\left(\mathcal{H}_3^\mathrm{tot}\right)}\right) \nonumber\\
		=&c V_{\sigma'}\left(R_{\sigma'^{-1}}^{\left(\mathcal{H}_1^\mathrm{tot}\right)}\otimes S_{\sigma'^{-1}}^{\left(\mathcal{H}_2^\mathrm{tot}\right)}\otimes T_{\sigma'^{-1}}^{\left(\mathcal{H}_3^\mathrm{tot}\right)}\right)\left(P^{\left(\mathcal{H}_1^\mathrm{tot}\right)}\otimes\ket{\bm{x}}\bra{\bm{y}}^{\left(\mathcal{H}_2^\mathrm{tot}\right)}\otimes T_{\sigma'}^{\left(\mathcal{H}_3^\mathrm{tot}\right)}\right) \nonumber\\
		=&c V_{\sigma'} \left(R_{\sigma'^{-1}}^{\left(\mathcal{H}_1^\mathrm{tot}\right)}P^{\left(\mathcal{H}_1^\mathrm{tot}\right)}
		\otimes S_{\sigma'^{-1}}^{\left(\mathcal{H}_2^\mathrm{tot}\right)}\ket{\bm{x}}\bra{\bm{y}}^{\left(\mathcal{H}_2^\mathrm{tot}\right)}
		\otimes T_{\sigma'^{-1}}^{\left(\mathcal{H}_3^\mathrm{tot}\right)}T_{\sigma'}^{\left(\mathcal{H}_3^\mathrm{tot}\right)}\right), \label{SMeq:SMlem:symmetrization1}
	\end{align}
	where $R_\sigma\in\mathcal{U}(\mathcal{H}_1^\mathrm{tot})$ and $S_\sigma\in\mathcal{U}(\mathcal{H}_2^\mathrm{tot})$ are defined as the permutation operators satisfying 
	\begin{align}
		&R_\sigma\left(\ket{\phi_1}^{\left(\mathcal{H}_1^1\right)}\otimes\ket{\phi_2}^{\left(\mathcal{H}_1^2\right)}\otimes\ket{\phi_3}^{\left(\mathcal{H}_1^3\right)}\right) 
		:=\ket{\phi_{\sigma^{-1}(1)}}^{\left(\mathcal{H}_1^1\right)}\otimes\ket{\phi_{\sigma^{-1}(2)}}^{\left(\mathcal{H}_1^2\right)}\otimes\ket{\phi_{\sigma^{-1}(3)}}^{\left(\mathcal{H}_1^3\right)}, \\
		&S_\sigma\left(\ket{\psi_1}^{\left(\mathcal{H}_2^1\right)}\otimes\ket{\psi_2}^{\left(\mathcal{H}_2^2\right)}\otimes\ket{\psi_3}^{\left(\mathcal{H}_2^3\right)}\right) 
		:=\ket{\psi_{\sigma^{-1}(1)}}^{\left(\mathcal{H}_2^1\right)}\otimes\ket{\psi_{\sigma^{-1}(2)}}^{\left(\mathcal{H}_2^2\right)}\otimes\ket{\psi_{\sigma^{-1}(3)}}^{\left(\mathcal{H}_2^3\right)} 
	\end{align}
	for all $\ket{\phi_j}\in\mathcal{H}_1$ and $\ket{\psi_j}\in\mathcal{H}_2$. 
	We note that 
	\begin{align}
		S_{\sigma'^{-1}}^{\left(\mathcal{H}_2^\mathrm{tot}\right)}\ket{\bm{x}}^{\left(\mathcal{H}_2^\mathrm{tot}\right)}
		=S_{\sigma'^{-1}}^{\left(\mathcal{H}_2^\mathrm{tot}\right)}\left(\bigotimes_{j=1}^3 \ket{\bm{x}_j}^{\left(\mathcal{H}_2^j\right)}\right) 
		=\bigotimes_{j=1}^3 \ket{\bm{x}_{\sigma'(j)}}^{\left(\mathcal{H}_2^j\right)} 
		=\bigotimes_{j=1}^3 \ket{\bm{y}_j}^{\left(\mathcal{H}_2^j\right)} 
		=\ket{\bm{y}}^{\left(\mathcal{H}_2^\mathrm{tot}\right)}. \label{SMeq:SMlem:symmetrization2}
	\end{align}
	where $\bm{x}_j:=(x_{j, l})_{l\in\{1, ..., N_2\}}$, $\bm{y}_j:=(y_{j, l})_{l\in\{1, ..., N_2\}}$, $\ket{\bm{x}_j}:=\bigotimes_{l\in\{1, ..., N_2\}} \ket{x_{j, l}}$ and $\ket{\bm{y}_j}:=\bigotimes_{l\in\{1, ..., N_2\}} \ket{y_{j, l}}$ for $j\in\{1, 2, 3\}$. 
	By Eqs.~\eqref{SMeq:SMlem:symmetrization1} and \eqref{SMeq:SMlem:symmetrization2}, we have 
	\begin{align}
		\mathcal{D}_4(J)=c V_{\sigma'} \left(R_{\sigma'^{-1}}^{\left(\mathcal{H}_1^\mathrm{tot}\right)}P^{\left(\mathcal{H}_1^\mathrm{tot}\right)}
		\otimes \ket{\bm{y}}\bra{\bm{y}}^{\left(\mathcal{H}_2^\mathrm{tot}\right)}
		\otimes I^{\left(\mathcal{H}_3^\mathrm{tot}\right)}\right). \label{SMeq:SMlem:symmetrization3}
	\end{align}
	$R_{\sigma'^{-1}}^{\left(\mathcal{H}_1^\mathrm{tot}\right)}P^{\left(\mathcal{H}_1^\mathrm{tot}\right)}$ can be written as 
	\begin{align}
		R_{\sigma'^{-1}}^{\left(\mathcal{H}_1^\mathrm{tot}\right)}P^{\left(\mathcal{H}_1^\mathrm{tot}\right)}=\sum_{P_1, P_2, P_3\in\mathcal{P}_{N_1}^+} \zeta_{P_1, P_2, P_3}\bigotimes_{j=1}^3 P_j^{\left(\mathcal{H}_1^j\right)} \label{SMeq:SMlem:symmetrization4}
	\end{align}
	with some $\zeta_{P_1, P_2, P_3}\in\mathbb{C}$ defined for $P_1, P_2, P_3\in\mathcal{P}_{N_1}^+$. 
	$\ket{\bm{y}}\bra{\bm{y}}^{\left(\mathcal{H}_2^\mathrm{tot}\right)}$ can be written as 
	\begin{align}
		\ket{\bm{y}}\bra{\bm{y}}^{\left(\mathcal{H}_2^\mathrm{tot}\right)}
		=&\bigotimes_{j\in\{1, 2, 3\}, l\in\{1, 2, ..., N_2\}} \ket{y_{j, l}}\bra{y_{j, l}}^{\left(\mathcal{H}_{2, l}^j\right)} \nonumber\\
		=&\bigotimes_{j\in\{1, 2, 3\}, l\in\{1, 2, ..., N_2\}} \frac{1}{2}\sum_{w_{j, l}\in\{0, 1\}} {\left[(-1)^{y_{j, l}}\mathrm{Z}\right]^{w_{j, l}}}^{\left(\mathcal{H}_{2, l}^j\right)} \nonumber\\
		=&\frac{1}{2^{3N_2}}\sum_{\bm{w}\in\{0, 1\}^{3N_2}} \bigotimes_{j\in\{1, 2, 3\}, l\in\{1, 2, ..., N_2\}} (-1)^{y_{j, l}w_{j, l}}{\mathrm{Z}^{w_{j, l}}}^{\left(\mathcal{H}_{2, l}^j\right)} \nonumber\\
		=&\frac{1}{2^{3N_2}}\sum_{\bm{w}\in\{0, 1\}^{3N_2}} (-1)^{\sum_{j, l} y_{j, l}w_{j, l}} \bigotimes_{j=1}^3\left(\bigotimes_{l=1}^{N_2} {\mathrm{Z}^{w_{j, l}}}^{\left(\mathcal{H}_{2, l}^j\right)}\right). \label{SMeq:SMlem:symmetrization5}
	\end{align}
	By plugging Eqs.~\eqref{SMeq:SMlem:symmetrization4} and \eqref{SMeq:SMlem:symmetrization5} into Eq.~\eqref{SMeq:SMlem:symmetrization3}, we get 
	\begin{align}
		\mathcal{D}_4(J)=\frac{c}{2^{3N_2}}\sum_{P_1, P_2, P_3\in\mathcal{P}_{N_1}^+, \bm{w}\in\{0, 1\}^{3N_2}} \zeta_{P_1, P_2, P_3}(-1)^{\sum_{j, l} y_{j, l}w_{j, l}} 
        V_{\sigma'}\bigotimes_{j=1}^3 \left(P_j^{\left(\mathcal{H}_1^j\right)}\otimes\bigotimes_{l=1}^{N_2} {\mathrm{Z}^{w_{j, l}}}^{\left(\mathcal{H}_{2, l}^j\right)}\otimes I^{(\mathcal{H}_3^j)}\right)\in\mathcal{M}_\mathcal{R}. 
	\end{align}
\hfill $\Box$\\

By combining Lemmas~\ref{SMlem:uniform_mixture_property1}, \ref{SMlem:fixed_points_specific} and \ref{SMlem:symmetrization}, we prove Proposition~\ref{SMprop:if_part}. \\

\noindent
\textit{Proof of Proposition~\ref{SMprop:if_part}.}
	Let $\mathcal{X}$ be $\mathcal{C}_{N, \mathcal{Q}}$ or $\mathcal{U}_{N, \mathcal{Q}}$. 
	By Lemma~\ref{SMlem:symmetrization}, we can take $\mathcal{D}\in\mathfrak{C}_{3, \mathcal{Q}}$ such that $\mathcal{D}(L)\in\mathcal{M}_\mathcal{Q}$ for all $L\in\mathcal{L}(\mathcal{H}^{\otimes 3})$. 
	For any $L\in\mathcal{L}(\mathcal{H}^{\otimes 3})$, By Lemma~\ref{SMlem:uniform_mixture_property1}, we have 
	\begin{align}
		\Phi_{3, \mathcal{X}}(L)=\Phi_{3, \mathcal{X}}(\mathcal{D}(L)). \label{SMeq:SMprop:if_part1}
	\end{align}
	Since $\mathcal{D}(L)\in\mathcal{M}_\mathcal{Q}$, by Lemma~\ref{SMlem:fixed_points_specific}, we get 
	\begin{align}
		\Phi_{3, \mathcal{X}}(\mathcal{D}(L))=\mathcal{D}(L). \label{SMeq:SMprop:if_part2}
	\end{align}
	Eqs.~\eqref{SMeq:SMprop:if_part1} and \eqref{SMeq:SMprop:if_part2} imply that $\Phi_{3, \mathcal{X}}=\mathcal{D}$. 
	Since this holds for $\mathcal{X}=\mathcal{C}_{N, \mathcal{Q}}$ and $\mathcal{U}_{N, \mathcal{Q}}$, we get $\Phi_{3, \mathcal{C}_{N, \mathcal{Q}}}=\Phi_{3, \mathcal{U}_{N, \mathcal{Q}}}$, i.e., $\mathcal{C}_{N, \mathcal{Q}}$ is a $\mathcal{Q}$-symmetric unitary $3$-design. 
\hfill $\Box$\\

\subsection{Proof of the ``only if'' part of Theorem~\ref{thm:main} (Theorem~\ref{SMthm:main})} \label{SMsubsec:3design_only_if}

In the ``only if'' part of the proof of Theorem~\ref{SMthm:main}, we take arbitrary unitary subgroup $\mathcal{G}$ such that the $\mathcal{G}$-symmetric Clifford group is a $\mathcal{G}$-symmetric unitary $3$-design, and construct a Pauli subgroup $\mathcal{Q}$ such that the constraints by $\mathcal{G}$ and $\mathcal{Q}$ are the same. 
We rigorously present this statement with the concrete construction of $\mathcal{Q}$ in the following proposition.

\begin{proposition} \label{SMprop:only_if_part}
	Let $N\in\mathbb{N}$, $\mathcal{G}$ be a subgroup of $\mathcal{U}_N$, $\mathcal{C}_{N, \mathcal{G}}$ be a $\mathcal{G}$-symmetric unitary $3$-design, and $\mathcal{Q}$ be defined by 
	\begin{align}
		\mathcal{Q}:=\braket{\{Q\in\mathcal{P}_N^+\ |\ \exists G\in\mathcal{G}\ \mathrm{s.t.}\ \mathrm{tr}(GQ)\neq0\}}. 
	\end{align}
	Then, $\mathcal{U}_{N, \mathcal{G}}=\mathcal{U}_{N, \mathcal{Q}}$. 
\end{proposition}

Since we always have $\mathcal{U}_{N, \mathcal{G}}\supset\mathcal{U}_{N, \mathcal{Q}}$ by the construction of $\mathcal{Q}$, we focus on the proof of $\mathcal{U}_{N, \mathcal{G}}\subset\mathcal{U}_{N, \mathcal{Q}}$. 
In the proof of this, it is central to prove that for any $G\in\mathcal{G}$, each term of $G$ in the Pauli basis with a nonzero coefficient must be invariant under the conjugation action of $U\in\mathcal{U}_{N, \mathcal{G}}$. 
We prove it in three steps, correspondingly in Lemmas~\ref{SMlem:Clifford_invariant_construction}, \ref{SMlem:unitary_design_commutativity} and \ref{SMlem:3_design_discreteness}. \\

First, we show that we can construct an operator $L\in\mathcal{L}(\mathcal{H}^{\otimes 3})$ from arbitrary taken $G\in\mathcal{G}$ such that $L$ is invariant under the conjugation action of $U^{\otimes 3}$ for all $U\in\mathcal{C}_{N, \mathcal{G}}$. 
Here we use the definition of the Clifford operators, which transform a Pauli operator into some Pauli operator by its conjugation action.

\begin{lemma} \label{SMlem:Clifford_invariant_construction}
	Let $N\in\mathbb{N}$, $t, t'\in\mathbb{N}$ satisfy $t\equiv t'$ $(\mathrm{mod}\ 2)$, $U\in\mathcal{C}_N$, $G\in\mathcal{L}(\mathcal{H})$ satisfy $UGU^\dag=G$, and $L\in\mathcal{L}(\mathcal{H}^{\otimes t})$ be defined by 
	\begin{align}
		L:=\sum_{P\in\mathcal{P}_N^+} \gamma_P(G)^{t'} P^{\otimes t}, 
	\end{align}
	where $\gamma_P(G): \mathcal{L}(\mathcal{H})\to\mathbb{C}$ gives the expansion coefficient of $P$ in the Pauli basis, i.e.,   
	\begin{align}
		\gamma_P(G):=\frac{1}{2^N}\mathrm{tr}(GP). \label{SMeq:gamma_def}
	\end{align} 
	Then, $U^{\otimes t}LU^{\dag\otimes t}=L$. 
\end{lemma}

We note that here we present this lemma in a general form because we also use it in the proof of Theorem~\ref{SMthm:2design}.

\begin{proof}
	By Lemma~\ref{SMlem:Clifford_action_on_Pauli} in Appendix~\ref{SMsec:technical}, there exist some function $s_U: \mathcal{P}_N^+\to\{\pm 1\}$ and some bijection $h_U$ on $\mathcal{P}_N^+$ such that for any $P\in\mathcal{P}_N^+$, 
	\begin{align}
		UPU^\dag=s_U(P)h_U(P). 
	\end{align}
	We therefore get 
	\begin{align}
		\sum_{P\in\mathcal{P}_N^+} \gamma_P(G) P
		=G
		=UGU^\dag
		=\sum_{P\in\mathcal{P}_N^+} \gamma_P(G) s_U(P) h_U(P)
		=\sum_{P\in\mathcal{P}_N^+} \gamma_{h_U^{-1}(P)}(G) s_U(h_U^{-1}(P)) P. 
	\end{align}
	By comparing the both sides, we get 
	\begin{align}
		\gamma_P(G)=\gamma_{h_U^{-1}(P)}(G) s_U(h_U^{-1}(P)) 
	\end{align}
	for all $P\in\mathcal{P}_N^+$. 
	By using this relation, we get 
	\begin{align}
		U^{\otimes t}LU^{\dag\otimes t}
		=&\sum_{P\in\mathcal{P}_N^+} \gamma_P(G)^{t'} (UPU^\dag)^{\otimes t} \nonumber\\
		=&\sum_{P\in\mathcal{P}_N^+} \gamma_P(G)^{t'} (s_U(P)h_U(P))^{\otimes t} \nonumber\\
		=&\sum_{P\in\mathcal{P}_N^+} \gamma_P(G)^{t'} s_U(P)^t h_U(P)^{\otimes t} \nonumber\\
		=&\sum_{P\in\mathcal{P}_N^+} (\gamma_P(G) s_U(P))^{t'} h_U(P)^{\otimes t} \nonumber\\
		=&\sum_{P\in\mathcal{P}_N^+} (\gamma_{h_U^{-1}(P)}(G) s_U(h_U^{-1}(P)))^{t'} P^{\otimes t} \nonumber\\
		=&\sum_{P\in\mathcal{P}_N^+} \gamma_P(G)^{t'} P^{\otimes t} \nonumber\\
		=&L. 
	\end{align}
\end{proof}

Second, we suppose that $\mathcal{C}_{N, \mathcal{G}}$ is a $\mathcal{G}$-symmetric unitary $3$-design and that $L\in\mathcal{L}(\mathcal{H}^{\otimes 3})$ is invariant under the conjugation action of $U^{\otimes 3}$ for all $U\in\mathcal{C}_{N, \mathcal{G}}$, and we show that such $L$ is invariant under the conjugation action of $U^{\otimes 3}$ even for all $U\in\mathcal{U}_{N, \mathcal{G}}$. 
This directly follows from the definitions of symmetric unitary designs and the Haar measure.

\begin{lemma} \label{SMlem:unitary_design_commutativity}
	Let $N\in\mathbb{N}$, $t\in\mathbb{N}$, $\mathcal{G}$ be a subgroup of $\mathcal{U}_N$, $\mathcal{C}_{N, \mathcal{G}}$ be a $\mathcal{G}$-symmetric unitary $t$-design, and $L\in\mathcal{L}(\mathcal{H}^{\otimes t})$ satisfy $U^{\otimes t}LU^{\dag\otimes t}=L$ for all $U\in\mathcal{C}_{N, \mathcal{G}}$. 
	Then, $U^{\otimes t}LU^{\dag\otimes t}=L$ for all $U\in\mathcal{U}_{N, \mathcal{G}}$. 
\end{lemma}

We present this lemma in a general form because we also use it in the proofs of Theorem~\ref{SMthm:2design} and Theorem~\ref{thm:4design} in the main text.

\begin{proof}
	Since $L$ satisfies $U'^{\otimes t}LU'^{\dag\otimes t}=L$ for all $U'\in\mathcal{C}_{N, \mathcal{G}}$, we have 
	\begin{align}
		\int_{U'\in\mathcal{C}_{N, \mathcal{G}}} U'^{\otimes t}LU'^{\dag\otimes t} d\mu_{\mathcal{C}_{N, \mathcal{G}}}(U')  
		=\int_{U'\in\mathcal{C}_{N, \mathcal{G}}} L d\mu_{\mathcal{C}_{N, \mathcal{G}}}(U') 
		=L. \label{SMeq:SMlem:unitary_design_commutativity1}
	\end{align}
	Since $\mathcal{C}_{N, \mathcal{G}}$ is a $\mathcal{G}$-symmetric unitary $t$-design, we have 
	\begin{align}
		\int_{U'\in\mathcal{C}_{N, \mathcal{G}}} U'^{\otimes t}LU'^{\dag\otimes t} d\mu_{\mathcal{C}_{N, \mathcal{G}}}(U') 
		=\int_{U'\in\mathcal{U}_{N, \mathcal{G}}} U'^{\otimes t}LU'^{\dag\otimes t} d\mu_{\mathcal{U}_{N, \mathcal{G}}}(U'). \label{SMeq:SMlem:unitary_design_commutativity2}
	\end{align}
	By Eqs.~\eqref{SMeq:SMlem:unitary_design_commutativity1} and \eqref{SMeq:SMlem:unitary_design_commutativity2}, we get 
	\begin{align}
		L=\int_{U'\in\mathcal{U}_{N, \mathcal{G}}} U'^{\otimes t}LU'^{\dag\otimes t} d\mu_{\mathcal{U}_{N, \mathcal{G}}}(U'). 
	\end{align}
	We therefore get for any $U\in\mathcal{U}_{N, \mathcal{G}}$, 
	\begin{align}
		U^{\otimes t}LU^{\dag\otimes t} 
		=&\int_{U'\in\mathcal{U}_{N, \mathcal{G}}} U^{\otimes t}U'^{\otimes t}LU'^{\dag\otimes t}U^{\dag\otimes t} d\mu_{\mathcal{U}_{N, \mathcal{G}}}(U') \nonumber\\
		=&\int_{U'\in\mathcal{U}_{N, \mathcal{G}}} (UU')^{\otimes t}L(UU')^{\dag\otimes t} d\mu_{\mathcal{U}_{N, \mathcal{G}}}(U') \nonumber\\
		=&\int_{U'\in\mathcal{U}_{N, \mathcal{G}}} U'^{\otimes t}LU'^{\dag\otimes t} d\mu_{\mathcal{U}_{N, \mathcal{G}}}(U') \nonumber\\
		=&L, 
	\end{align}
	where we used the left invariance of $\mu_{\mathcal{U}_{N, \mathcal{G}}}$. 
\end{proof}

Finally, under the assumption that $L$ is in the form given in Lemma~\ref{SMlem:Clifford_invariant_construction} and it is invariant under the conjugation action of $U^{\otimes t}$ for all $U\in\mathcal{U}_{N, \mathcal{G}}$, we prove that every Pauli basis composing $L$ with a nonzero coefficient is invariant under the conjugation action of $U$ for all $U\in\mathcal{U}_{N, \mathcal{G}}$. 
In the proof of this, we fix a Pauli operator $P$ and consider a continuous map $U\mapsto UPU^\dag$ from $\mathcal{U}_{N, \mathcal{G}}$ to $\mathcal{U}_N$.  
By the assumption above and the unitary invariance of the Hilbert-Schmidt norm, we know that the value of this map only takes discrete points. 
The combination of this and the connectedness of $\mathcal{U}_{N, \mathcal{G}}$ implies that this map always takes the constant value $P$, where we use the fact that only a singleton is a discrete and connected nonempty set. 
In the concrete proof process, we consider the linear expansion of $UPU^\dag$ in the Pauli basis.

\begin{lemma} \label{SMlem:3_design_discreteness}
	 Let $N\in\mathbb{N}$, $t\geq 3$, $\mathcal{G}$ be a subgroup of $\mathcal{U}_N$, and $L\in\mathcal{L}(\mathcal{H}^{\otimes t})$ be defined by 
	\begin{align}
		L:=\sum_{P\in\mathcal{P}_N^+} \gamma'_P P^{\otimes t}
	\end{align}
	with some $\gamma'_P\in\mathbb{C}$ for $P\in\mathcal{P}_N^+$ and satisfy 
	\begin{align}
		U^{\otimes t}LU^{\dag\otimes t}=L \label{SMeq:SMlem:3_design_discreteness1}
	\end{align}
	for all $U\in\mathcal{U}_{N, \mathcal{G}}$. 
	Then, $UPU^\dag=P$ for all $U\in\mathcal{U}_{N, \mathcal{G}}$ and $P\in\mathcal{P}_N^+$ satisfying $\gamma'_P\neq0$. 
\end{lemma}

This lemma is also used in the proof of Theorem~\ref{thm:4design} in the main text.

\begin{proof}
	We define $\alpha_{P_1, P_2}(U)$: $\mathcal{U}_{N, \mathcal{G}}\to\mathbb{R}$ by 
	\begin{align}
		\alpha_{P_1, P_2}(U):=\frac{1}{2^N}\mathrm{tr}(UP_1U^\dag P_2) 
	\end{align}	
	for $P_1, P_2\in\mathcal{P}_N^+$, and prove that $\{\alpha_{P, P'}(U)\}_{U\in\mathcal{U}_{N, \mathcal{G}}}$ is discrete for all $P, P'\in\mathcal{P}_N^+$ satisfying $\gamma'_P\neq0$. 
	We can confirm that $\alpha_{P_1, P_2}(U)\in\mathbb{R}$ by noting that $UP_1U^\dag$ and $P_2$ are hermitian. 
	Take arbitrary $U\in\mathcal{U}_{N, \mathcal{G}}$. 
	Since $L$ satisfies Eq.~\eqref{SMeq:SMlem:3_design_discreteness1}, for any $P'\in\mathcal{P}_N^+$, we get 
	\begin{align}
		\sum_{P\in\mathcal{P}_N^+} \gamma'_P \alpha_{P, P'}(U)^t
		=&\sum_{P\in\mathcal{P}_N^+} \gamma'_P \cdot\frac{1}{2^{tN}}\mathrm{tr}((P'UPU^\dag)^{\otimes t}) \nonumber\\
		=&\frac{1}{2^{tN}}\mathrm{tr}({P'}^{\otimes t}U^{\otimes t}LU^{\dag\otimes t}) \nonumber\\
		=&\frac{1}{2^{tN}}\mathrm{tr}({P'}^{\otimes t}L) \nonumber\\
		=&\sum_{P\in\mathcal{P}_N^+} \gamma'_P \cdot\frac{1}{2^{tN}}\mathrm{tr}((P'P)^{\otimes t}) \nonumber\\
		=&\sum_{P\in\mathcal{P}_N^+} \gamma'_P \delta_{P, P'} \nonumber\\
		=&\gamma'_{P'}.
	\end{align}
	By the triangle inequality, we have 
	\begin{align}
		\sum_{P\in\mathcal{P}_N^+} |\gamma'_P| |\alpha_{P, P'}(U)|^t
		\geq|\gamma'_{P'}|.
	\end{align}
	By taking the sum over $P'\in\mathcal{P}_N^+$, we get 
	\begin{align}
		\sum_{P, P'\in\mathcal{P}_N^+} |\gamma'_P| |\alpha_{P, P'}(U)|^t
		\geq\sum_{P'\in\mathcal{P}_N^+} |\gamma'_{P'}|
		=\sum_{P\in\mathcal{P}_N^+} |\gamma'_P|. \label{SMeq:SMlem:3_design_discreteness2}
	\end{align}
	For any $P\in\mathcal{P}_N^+$, $UPU^\dag$ can be expanded in the Pauli basis as 
	\begin{align}
		UPU^\dag
		=\sum_{P'\in\mathcal{P}_N^+} \frac{1}{2^N}\mathrm{tr}(UPU^\dag P')P'
		=\sum_{P'\in\mathcal{P}_N^+} \alpha_{P, P'}(U)P'. \label{SMeq:SMlem:3_design_discreteness3}
	\end{align}
	By considering the Hilbert-Schmidt norm of the both sides, we get 
	\begin{align}
		1
		=&\frac{1}{2^N}\mathrm{tr}((UPU^\dag)^\dag(UPU^\dag)) \nonumber\\
		=&\frac{1}{2^N}\mathrm{tr}\left(\sum_{P', P''\in\mathcal{P}_N^+} \alpha_{P, P'}(U)^* \alpha_{P, P''}(U)P'P''\right) \nonumber\\
		=&\sum_{P', P''\in\mathcal{P}_N^+} \alpha_{P, P'}(U)^* \alpha_{P, P''}(U)\frac{1}{2^N}\mathrm{tr}(P'P'') \nonumber\\
		=&\sum_{P', P''\in\mathcal{P}_N^+} \alpha_{P, P'}(U)^* \alpha_{P, P''}(U)\delta_{P', P''} \nonumber\\
		=&\sum_{P'\in\mathcal{P}_N^+} |\alpha_{P, P'}(U)|^2. \label{SMeq:SMlem:3_design_discreteness4}
	\end{align}
	By Eqs.~\eqref{SMeq:SMlem:3_design_discreteness2} and \eqref{SMeq:SMlem:3_design_discreteness4}, we get 
	\begin{align}
		\sum_{P, P'\in\mathcal{P}_N^+} |\gamma'_P||\alpha_{P, P'}(U)|^t
		\geq\sum_{P\in\mathcal{P}_N^+} |\gamma'_P|\cdot 1
		=\sum_{P\in\mathcal{P}_N^+} |\gamma'_P|\left(\sum_{P'\in\mathcal{P}_N^+} |\alpha_{P, P'}(U)|^2\right)
		=\sum_{P, P'\in\mathcal{P}_N^+} |\gamma'_P| |\alpha_{P, P'}(U)|^2. 
	\end{align}
	This is equivalently expressed as  
	\begin{align}
		\sum_{P, P'\in\mathcal{P}_N^+} |\gamma'_P|(|\alpha_{P, P'}(U)|^t-|\alpha_{P, P'}(U)|^2)\geq0. \label{SMeq:SMlem:3_design_discreteness5}
	\end{align}
	Since $|\alpha_{P, P'}(U)|\leq1$ by Eq.~\eqref{SMeq:SMlem:3_design_discreteness4} and $t\geq 3$, we have $|\alpha_{P, P'}(U)|^t-|\alpha_{P, P'}(U)|^2\leq0$. 
	This implies that equality holds in Eq.~\eqref{SMeq:SMlem:3_design_discreteness5}. 
	We therefore get for any $P, P'\in\mathcal{P}_N^+$ satisfying $\gamma'_P\neq0$, 
	\begin{align}
		|\alpha_{P, P'}(U)|^t-|\alpha_{P, P'}(U)|^2=0. 
	\end{align}
	Since $\alpha_{P, P'}(U)\in\mathbb{R}$, we get 
	\begin{align}
		\alpha_{P, P'}(U)=0, \pm1. 
	\end{align} 
	This implies that $\{\alpha_{P, P'}(U)\}_{U\in\mathcal{U}_{N, \mathcal{G}}}$ is discrete. 
	On the other hand, since $\mathcal{U}_{N, \mathcal{G}}$ is connected by Lemma~\ref{SMlem:symmetric_unitary_connectedness} in Appendix~\ref{SMsec:technical} and $\alpha_{P, P'}(U)$ is continuous, $\{\alpha_{P, P'}(U)\}_{U\in\mathcal{U}_{N, \mathcal{G}}}$ is connected. 
	We therefore know that $\{\alpha_{P, P'}(U)\}_{U\in\mathcal{U}_{N, \mathcal{G}}}$ is a singleton, which implies that for any $U\in\mathcal{U}_{N, \mathcal{G}}$, 
	\begin{align}
		\alpha_{P, P'}(U)=\alpha_{P, P'}(I)=\delta_{P, P'}. 
	\end{align}
	By plugging this into Eq.~\eqref{SMeq:SMlem:3_design_discreteness3}, we get 
	\begin{align}
		UPU^\dag
		=\sum_{P'\in\mathcal{P}_N^+} \delta_{P, P'}P'
		=P. 
	\end{align}
\end{proof}

By combining the three lemmas above, we prove Proposition~\ref{SMprop:only_if_part}. \\

\noindent
\textit{Proof of Proposition~\ref{SMprop:only_if_part}.}
	First, we prove that $\mathcal{U}_{N, \mathcal{G}}\supset\mathcal{U}_{N, \mathcal{Q}}$. 
	We define $\mathcal{Q}':=\{Q\in\mathcal{P}_N^+\ |\ \exists G\in\mathcal{G}\ \mathrm{tr}(GQ)\neq0\}$. 
	Then, $\mathcal{Q}$ is the group $\braket{\mathcal{Q}'}$ generated by $\mathcal{Q}'$. 
	For any $G\in\mathcal{G}$, we know that 
	\begin{align}
		G=\sum_{Q\in\mathcal{P}_N^+} \gamma_Q(G)Q
		=\sum_{Q\in\mathcal{Q}'} \gamma_Q(G)Q+\sum_{Q\in\mathcal{P}_N^+\backslash\mathcal{Q}'} \gamma_Q(G)Q
		=\sum_{Q\in\mathcal{Q}'} \gamma_Q(G)Q
		\in\mathrm{span}(\mathcal{Q}')
		\subset\mathrm{span}(\mathcal{Q}), 
	\end{align}
	where $\gamma_Q(G)$ is defined by Eq.~\eqref{SMeq:gamma_def}, and we used $\gamma_Q(G)=0$ for all $Q\in\mathcal{P}_N^+\backslash\mathcal{Q}'$ by the definition of $\mathcal{Q}'$. 
	For any $U\in\mathcal{U}_{N, \mathcal{Q}}$, we therefore get $[U, G]=0$. 
	Since this holds for all $G\in\mathcal{G}$, we get $\mathcal{U}_{N, \mathcal{Q}}\subset\mathcal{U}_{N, \mathcal{G}}$.

	Next, we prove that $\mathcal{U}_{N, \mathcal{G}}\subset\mathcal{U}_{N, \mathcal{Q}}$. 
	Take arbitrary $Q\in\mathcal{Q}'$. 
	By the definition of $\mathcal{Q}'$, we can take $G\in\mathcal{G}$ such that $\gamma_Q(G)\neq0$. 
	By the definition of $\gamma_P(G)$, $G$ can be written as 
	\begin{align}
		G=\sum_{P\in\mathcal{P}_N^+} \gamma_P(G)P. 
	\end{align}
	We define $L\in\mathcal{L}(\mathcal{H}^{\otimes 3})$ by 
	\begin{align}
		L:=\sum_{P\in\mathcal{P}_N^+} \gamma_P(G)P^{\otimes 3}. 
	\end{align}
	By Lemma~\ref{SMlem:Clifford_invariant_construction}, we get $U^{\otimes 3}LU^{\dag\otimes 3}=L$ for all $U\in\mathcal{C}_{N, \mathcal{G}}$. 
	Since $\mathcal{C}_{N, \mathcal{G}}$ is a $\mathcal{G}$-symmetric unitary $3$-design, by Lemma~\ref{SMlem:unitary_design_commutativity}, we get $U^{\otimes 3}LU^{\dag\otimes 3}=L$ for all $U\in\mathcal{U}_{N, \mathcal{G}}$. 
	By Lemma~\ref{SMlem:3_design_discreteness}, we know that $UPU^\dag=P$ for all $U\in\mathcal{U}_{N, \mathcal{G}}$ and $P\in\mathcal{P}_N^+$ satisfying $\gamma_P(G)\neq0$. 
	Since $Q$ satisfies $\gamma_Q(G)\neq0$, we have for any $U\in\mathcal{U}_{N, \mathcal{G}}$, $UQU^\dag=Q$, or equivalently $[U, Q]=0$. 
	Since this holds for all $Q\in\mathcal{Q}'$ and $\mathcal{Q}=\braket{\mathcal{Q}'}$, we get $[U, Q]=0$ for all $U\in\mathcal{U}_{N, \mathcal{G}}$ and $Q\in\mathcal{Q}$. 
	This implies that $\mathcal{U}_{N, \mathcal{G}}\subset\mathcal{U}_{N, \mathcal{Q}}$. 
\hfill $\Box$\\

By combining Propositions~\ref{SMprop:if_part} and \ref{SMprop:only_if_part}, we prove Theorem~\ref{SMthm:main}. \\

\noindent
\textit{Proof of Theorem~\ref{SMthm:main}.}
	First, we consider the ``if'' part. 
	We suppose that $\mathcal{U}_{N, \mathcal{G}}=\mathcal{U}_{N, \mathcal{Q}}$ with some subgroup $\mathcal{Q}$ of $\mathcal{P}_N$. 
	Then, the conditions for $\mathcal{G}$- and $\mathcal{Q}$-symmetric unitary $3$-designs are equivalent and $\mathcal{C}_{N, \mathcal{G}}=\mathcal{C}_N \cap\mathcal{U}_{N, \mathcal{G}}=\mathcal{C}_N \cap\mathcal{U}_{N, \mathcal{Q}}=\mathcal{C}_{N, \mathcal{Q}}$. 
	Thus it suffices to show that $\mathcal{C}_{N, \mathcal{Q}}$ is a $\mathcal{Q}$-symmetric unitary $3$-design, which we proven in Proposition~\ref{SMprop:if_part}.

	Next, we consider the ``only if'' part. 
	We suppose that $\mathcal{C}_{N, \mathcal{G}}$ is a $\mathcal{G}$-symmetric unitary $3$-design. 
	We define $\mathcal{Q}:=\braket{\{Q\in\mathcal{P}_N^+\ |\ \exists G\in\mathcal{G}\ \mathrm{s.t.}\ \mathrm{tr}(GQ)\neq0\}}$. 
	Then, $\mathcal{Q}$ is a subgroup of $\mathcal{P}_N$ and by Proposition~\ref{SMprop:only_if_part}, we get $\mathcal{U}_{N, \mathcal{G}}=\mathcal{U}_{N, \mathcal{Q}}$.
 
\hfill $\Box$\\

\section{Proof of Theorem~\ref{thm:expression} (Construction of symmetric Clifford groups)} \label{SMsec:construction}

In this appendix, we present complete and unique constructions of symmetric Clifford groups with elementary gate sets under Pauli symmetries and certain non-Pauli symmetries in Theorems~\ref{SMthm:symmetric_Clifford_expression} and \ref{SMthm:nonPauli_Clifford_expression}, respectively (see Figs.~\ref{fig:circuit} and \ref{fig:nonPauli_circuit} in the main text). 
Theorem~\ref{SMthm:symmetric_Clifford_expression} corresponds to Theorem~\ref{thm:expression} in the main text. 
We note that the theorems about unitary designs in this paper are proven without using these theorems.

\subsection{Pauli symmetries}

First, for any Pauli symmetry, we present a construction of the symmetric Clifford group and prove that the construction is complete and unique. 
By using the fact that any Pauli subgroup can be transformed into a simple Pauli subgroup in the form of Eq.~\eqref{SMeq:Pauli_subgroup_standard}, we know that it is sufficient only to consider such form of symmetries. 
We show that every symmetric Clifford operator can be written with the four types of operators that we consider in the proofs of Lemmas~\ref{SMlem:S_gates}, \ref{SMlem:CZ^BB_mixture}, \ref{SMlem:Clifford_3_mixture} and \ref{SMlem:CPauli^BC_mixture}, and we also show that it is expressed in a unique way.

\begin{theorem} \label{SMthm:symmetric_Clifford_expression}
    (Restatement of Theorem~\ref{thm:expression}.) 
	Let $N\in\mathbb{N}$ and $\mathcal{Q}$ be a subgroup of $\mathcal{P}_N$. 
	Then, there exist $W\in\mathcal{C}_N$ and $N_1, N_2, N_3\in\mathbb{N}$ such that  
	\begin{align}
		\mathcal{P}_0 W\mathcal{Q}W^\dag=\mathcal{P}_0\{\mathrm{I}, \mathrm{X}, \mathrm{Y}, \mathrm{Z}\}^{\otimes N_1}\otimes\{\mathrm{I}, \mathrm{Z}\}^{\otimes N_2}\otimes\{\mathrm{I}\}^{\otimes N_3}, 
	\end{align}
	which gives a decomposition of $N$ qubits into three subsystems $\mathrm{A}_1$, $\mathrm{A}_2$ and $\mathrm{A}_3$, each consisting of $N_1$, $N_2$ and $N_3$ qubits. 
	Moreover, for any $U\in\mathcal{C}_{N, \mathcal{Q}}$, there uniquely exist $\{\mu_j\}_{j=1}^{N_2} \in\{0, 1, 2, 3\}^{N_2}$, $\{\nu_{j, k}\}_{1\leq j<k\leq N_2}\in\{0, 1\}^{N_2(N_2-1)/2}$, $V\in\mathcal{C}_{N_3}$ and $\{P_j\}_{j=1}^{N_2}\in(\mathcal{P}_{N_3}^+)^{N_2}$ such that 
	\begin{align}
		U=W^\dag\left(\mathrm{T}\prod_{j=1}^{N_2} \mathrm{C}_j(P_j)\right)
		V
		\left(\prod_{j, k: 1\leq j<k\leq N_2} \mathrm{CZ}_{(2, j), (2, k)}^{\nu_{j, k}}\right)
		\left(\prod_{j=1}^{N_2} \mathrm{S}_{(2, j)}^{\mu_j}\right)
		W, \label{SMeq:Pauli_Clifford_general}
	\end{align}
	where $\mathrm{S}_{(2, j)}$ acts on the $j$th qubit in the system $\mathrm{A}_2$, $\mathrm{CZ}_{(2, j), (2, k)}$ acts on the $j$th and $k$th qubits in the system $\mathrm{A}_2$, $V$ acts on the system $\mathrm{A}_3$, $\mathrm{C}_j(P_j)$ is the controlled-$P_j$ gate with the $j$th qubit in the system $\mathrm{A}_2$ as the control qubit and the qubits in the system $\mathrm{A}_3$ as the target qubits, and $\mathrm{T}\prod$ represents the ordered product, i.e., $\mathrm{T}\prod_{j=1}^{n} O_j:=O_n\cdots O_2 O_1$. 
\end{theorem}

We note that here we use different notations from those in Theorem~\ref{thm:expression} in the main text for convenience of explanation. 
The subscript $(j, k)$ in $O_{(j, k)}$ means that $O$ acts on the $k$th qubit in the susbystem $\mathrm{A}_j$.

\begin{proof}
	By Lemma~\ref{SMlem:Pauli_subgroup_equivalence} in Appendix~\ref{SMsec:technical}, we can take $W\in\mathcal{C}_N$ and $N_1, N_2, N_3\geq 0$ such that $\mathcal{P}_0 W\mathcal{Q}W^\dag=\mathcal{P}_0\mathcal{R}$ and $\mathcal{R}=\mathcal{P}_0\{\mathrm{I}, \mathrm{X}, \mathrm{Y}, \mathrm{Z}\}^{\otimes N_1}\otimes\{\mathrm{I}, \mathrm{Z}\}^{\otimes N_2}\otimes\{I\}^{\otimes N_3}$. 
	Since $U\in\mathcal{C}_{N, \mathcal{Q}}$ is equivalent to $W^\dag UW\in\mathcal{C}_{N, \mathcal{R}}$, it is sufficient to prove that for any $U\in\mathcal{C}_{N, \mathcal{R}}$, there uniquely exist $\{\mu_j\}_{j=1}^{N_2} \in\{0, 1, 2, 3\}^{N_2}$, $\{\nu_{j, k}\}_{1\leq j<k\leq N_2}\in\{0, 1\}^{N_2(N_2-1)/2}$, $V\in\mathcal{C}_{N_3}$ and $\{P_j\}_{j=1}^{N_2}\in(\mathcal{P}_{N_3}^+)^{N_2}$ such that 
	\begin{align}
		U=\left(\mathrm{T}\prod_{j=1}^{N_2} \mathrm{C}_j(P_j)\right)
		V
		\left(\prod_{j, k: 1\leq j<k\leq N_2} \mathrm{CZ}_{(2, j), (2, k)}^{\nu_{j, k}}\right)
		\left(\prod_{j=1}^{N_2} \mathrm{S}_{(2, j)}^{\mu_j}\right). \label{SMeq:C_R_expression}
	\end{align}

	First, we take arbitrary $U\in\mathcal{C}_{N, \mathcal{R}}$ and prove that $U$ can be expressed in the form of Eq.~\eqref{SMeq:C_R_expression}. 
	We are going to prove by mathematical induction that for any $l\in\{0, 1, ..., N_2\}$, there exist $\{\mu_j\}_{j=1}^l \in\{0, 1, 2, 3\}^l$, $\{\nu_{j, k}\}_{1\leq j\leq l, j+1\leq k\leq N_2}\in\{0, 1\}^{l(2N_2-l-1)/2}$, $T\in\mathcal{C}_{N_3}$ and $\{Q_j\}_{j=1}^l\in(\pm\mathcal{P}_{N_3}^+)^l$ such that 
	\begin{align}
		&[U', \mathrm{Z}_{(1, m)}]=[U', \mathrm{X}_{(1, m)}]=0\ \mathrm{if}\ 1\leq m\leq N_1, \label{SMeq:SMthm:symmetric_Clifford_expression1}\\
		&[U', \mathrm{Z}_{(2, m)}]=0\ \mathrm{if}\ 1\leq m\leq N_2, \label{SMeq:SMthm:symmetric_Clifford_expression2}\\
		&[U', \mathrm{X}_{(2, m)}]=0\ \mathrm{if}\ 1\leq m\leq l, \label{SMeq:SMthm:symmetric_Clifford_expression3}
	\end{align}
	where 
	\begin{align}
		U':=\left(\mathrm{T}\prod_{j=1}^l \mathrm{C}_j(Q_j)\right)
		T
		\left(\prod_{j=1}^l \prod_{k=j+1}^{N_2} \mathrm{CZ}_{(2, j), (2, k)}^{\nu_{j, k}}\right)
		\left(\prod_{j=1}^l \mathrm{S}_{(2, j)}^{\mu_j}\right)
		U^\dag. \label{SMeq:SMthm:symmetric_Clifford_expression4}
	\end{align}
	We can easily verify that this holds for $l=0$, because $U'=U$ in this case and $U\in\mathcal{C}_{N, \mathcal{R}}$ satisfies $U'\mathrm{Z}_{(1, m)}U'^\dag=\mathrm{Z}_{(1, m)}$ and $U'\mathrm{X}_{(1, m)}U'^\dag=\mathrm{X}_{(1, m)}$ for all $m\in\{1, 2, ..., N_1\}$, and $U'\mathrm{Z}_{(2, m)}U'^\dag=\mathrm{Z}_{(2, m)}$ for all $m\in\{1, 2, ..., N_2\}$. 
	We take arbitrary $l\in\{0, 1, ..., N_2-1\}$ and suppose that we can take $\{\mu_j\}_{j=1}^l \in\{0, 1, 2, 3\}^{N_2}$, $\{\nu_{j, k}\}_{1\leq j\leq l, j+1\leq k\leq N_2}\in\{0, 1\}^{l(2N_2-l-1)/2}$, $T\in\mathcal{C}_{N_3}$ and $\{Q_j\}_{j=1}^l \in(\pm\mathcal{P}_{N_3}^+)^l$ satisfying Eqs.~\eqref{SMeq:SMthm:symmetric_Clifford_expression1}, \eqref{SMeq:SMthm:symmetric_Clifford_expression2} and \eqref{SMeq:SMthm:symmetric_Clifford_expression3}. 
	By Eq.~\eqref{SMeq:SMthm:symmetric_Clifford_expression1}, we get 
	\begin{align}
		[U'\mathrm{X}_{(2, l+1)}U'^\dag, \mathrm{Z}_{(1, m)}]
		=[U'\mathrm{X}_{(2, l+1)}U'^\dag, U'\mathrm{Z}_{(1, m)}U'^\dag]
		=U'[\mathrm{X}_{(2, l+1)}, \mathrm{Z}_{(1, m)}]U'^\dag
		=0\ \mathrm{if}\ 1\leq m\leq N_1. \label{SMeq:SMthm:symmetric_Clifford_expression5}
	\end{align} 
	In the same way, by Eqs.~\eqref{SMeq:SMthm:symmetric_Clifford_expression1}, \eqref{SMeq:SMthm:symmetric_Clifford_expression2} and \eqref{SMeq:SMthm:symmetric_Clifford_expression3}, we get 
	\begin{align}
		&[U'\mathrm{X}_{(2, l+1)}U'^\dag, \mathrm{X}_{(1, m)}]
		=0\ \mathrm{if}\ 1\leq m\leq N_1, \label{SMeq:SMthm:symmetric_Clifford_expression6}\\
		&[U'\mathrm{X}_{(2, l+1)}U'^\dag, \mathrm{Z}_{(2, m)}]
		=U'[\mathrm{X}_{(2, l+1)}, \mathrm{Z}_{(2, m)}]U'^\dag
		=0\ \mathrm{if}\ 1\leq m\leq N_2\ \mathrm{and}\ m\neq l+1, \label{SMeq:SMthm:symmetric_Clifford_expression7}\\
		&[U'\mathrm{X}_{(2, l+1)}U'^\dag, \mathrm{Z}_{(2, l+1)}]
		=U'[\mathrm{X}_{(2, l+1)}, \mathrm{Z}_{(2, l+1)}]U'^\dag
		\neq0, \label{SMeq:SMthm:symmetric_Clifford_expression8}\\
		&[U'\mathrm{X}_{(2, l+1)}U'^\dag, \mathrm{X}_{(2, m)}]
		=U'[\mathrm{X}_{(2, l+1)}, \mathrm{X}_{(2, m)}]U'^\dag
		=0\ \mathrm{if}\ 1\leq m\leq l. \label{SMeq:SMthm:symmetric_Clifford_expression9}
	\end{align}
	Since $U'\in\mathcal{C}_N$ and $\mathrm{X}_{(2, l+1)}\in\mathcal{P}_N^+$, we get 
	\begin{align}
		U'\mathrm{X}_{(2, l+1)}U'^\dag\in\pm\mathcal{P}_N^+. \label{SMeq:SMthm:symmetric_Clifford_expression10}
	\end{align} 
	By Eqs.~\eqref{SMeq:SMthm:symmetric_Clifford_expression5}, \eqref{SMeq:SMthm:symmetric_Clifford_expression6}, \eqref{SMeq:SMthm:symmetric_Clifford_expression7}, \eqref{SMeq:SMthm:symmetric_Clifford_expression8}, \eqref{SMeq:SMthm:symmetric_Clifford_expression9} and \eqref{SMeq:SMthm:symmetric_Clifford_expression10},
we get $U'\mathrm{X}_{(2, l+1)}U'^\dag\in\pm\{\mathrm{X}_{(2, l+1)}, \mathrm{Y}_{(2, l+1)}\}\otimes\{I_{(2, l+2)}, \mathrm{Z}_{(2, l+2)}\}\otimes\cdots\otimes\{I_{(2, N_2)}, \mathrm{Z}_{(2, N_2)}\}\otimes\{I_{(3, 1)}, \mathrm{X}_{(3, 1)}, \mathrm{Y}_{(3, 1)}, \mathrm{Z}_{(3, 1)}\}\otimes\cdots\otimes\{I_{(3, N_3)}, \mathrm{X}_{(3, N_3)}, \mathrm{Y}_{(3, N_3)}, \mathrm{Z}_{(3, N_3)}\}$. 
	By noting that $\mathrm{S}\mathrm{X}\mathrm{S}^\dag=\mathrm{Y}$, $\mathrm{S}\mathrm{Y}\mathrm{S}^\dag=-\mathrm{X}$ and $\mathrm{H}\mathrm{X}\mathrm{H}^\dag=\mathrm{Z}$, we can take $\mu_{l+1}\in\{0, 1, 2, 3\}$ and $T'\in\mathcal{C}_{N_3}$ 
	such that 
	$(\mathrm{S}_{(2, l+1)}^{\mu_{l+1}}\otimes T') U'\mathrm{X}_{(2, l+1)}U'^\dag(\mathrm{S}_{(2, l+1)}^{\mu_{l+1}}\otimes T')^\dag\in\{\mathrm{X}_{(2, l+1)}\}\otimes\{\mathrm{I}_{(2, l+2)}, \mathrm{Z}_{(2, l+2)}\}\otimes\cdots\otimes\{\mathrm{I}_{(2, N_2)}, \mathrm{Z}_{(2, N_2)}\}\otimes\{I_{(3, 1)}, \mathrm{Z}_{(3, 1)}\}\otimes\cdots\otimes\{I_{(3, N_3)}, \mathrm{Z}_{(3, N_3)}\}$. 
	Equivalently, there exist $\{\nu_{l+1, k}\}_{k=l+2}^{N_2}\in\{0, 1\}^{N_2-l-1}$ and $\{\xi_k\}_{k=1}^{N_3}\in\{0, 1\}^{N_3}$ satisfying 
	\begin{align}
		(\mathrm{S}_{(2, l+1)}^{\mu_{l+1}}\otimes T') U'\mathrm{X}_{(2, l+1)}U'^\dag(\mathrm{S}_{(2, l+1)}^{\mu_{l+1}}\otimes T')^\dag=\mathrm{X}_{(2, l+1)}
		\left(\prod_{k=l+2}^{N_2} \mathrm{Z}_{(2, k)}^{\nu_{l+1, k}}\right) 
		\left(\prod_{k=1}^{N_3} \mathrm{Z}_{(3, k)}^{\xi_k}\right). 
	\end{align}
	By noting that $\mathrm{CZ}(\mathrm{X}\otimes \mathrm{I})\mathrm{CZ}^\dag=\mathrm{X}\otimes\mathrm{Z}$, we know that $U''\mathrm{X}_{(2, l+1)}U''^\dag=\mathrm{X}_{(2, l+1)}$, i.e., $[U'', \mathrm{X}_{(2, l+1)}]=0$ with $U''$ defined by 
	\begin{align}
		U'':=\left(\prod_{k=1}^{N_3} \mathrm{CZ}_{(2, l+1), (3, k)}^{\xi_k}\right) 
		\left(\prod_{k=l+2}^{N_2} \mathrm{CZ}_{(2, l+1), (2, k)}^{\nu_{l+1, k}}\right) 
		(\mathrm{S}_{(2, l+1)}^{\mu_{l+1}}\otimes T')U'. \label{SMeq:SMthm:symmetric_Clifford_expression11}
	\end{align}
	By Eqs.~\eqref{SMeq:SMthm:symmetric_Clifford_expression1}, \eqref{SMeq:SMthm:symmetric_Clifford_expression2}, \eqref{SMeq:SMthm:symmetric_Clifford_expression3} and \eqref{SMeq:SMthm:symmetric_Clifford_expression11}, we get $[U'', \mathrm{Z}_{(1, m)}]=[U'', \mathrm{X}_{(1, m)}]=0$ if $1\leq m\leq N_1$, 
	$[U'', \mathrm{Z}_{(2, m)}]=0$ if $1\leq m\leq N_2$, and 
	$[U'', \mathrm{X}_{(2, m)}]=0$ if $1\leq m\leq l$. 
	In order to complete the process of mathematical induction, we show that $U''$ can be written in the form of Eq.~\eqref{SMeq:SMthm:symmetric_Clifford_expression4}. 
	By plugging Eq.~\eqref{SMeq:SMthm:symmetric_Clifford_expression4} into \eqref{SMeq:SMthm:symmetric_Clifford_expression11}, we get 
	\begin{align}
		U''=&\mathrm{C}_{l+1}\left(\prod_{k=1}^{N_3} \mathrm{Z}_{(3, k)}^{\xi_k}\right)
		T'
		\left(\prod_{k=l+2}^{N_2} \mathrm{CZ}_{(2, l+1), (2, k)}^{\nu_{l+1, k}}\right) \mathrm{S}_{(2, l+1)}^{\mu_{l+1}} \nonumber\\
		&\ \ \ \ \ \ \ \ \ \ \times\left(\mathrm{T}\prod_{j=1}^l \mathrm{C}_j(Q_j)\right)
		T
		\left(\prod_{j=1}^l \prod_{k=j+1}^{N_2} \mathrm{CZ}_{(2, j), (2, k)}^{\nu_{j, k}}\right) \left(\prod_{j=1}^l \mathrm{S}_{(2, j)}^{\mu_j}\right)U^\dag. \label{SMeq:SMthm:symmetric_Clifford_expression12}
	\end{align}
	We note that 
	\begin{align}
		T'\left(\prod_{k=l+2}^{N_2} \mathrm{CZ}_{(2, l+1), (2, k)}^{\nu_{l+1, k}}\right) \mathrm{S}_{(2, l+1)}^{\mu_{l+1}}
		\left(\mathrm{T}\prod_{j=1}^l \mathrm{C}_j(Q_j)\right)
		=\left(\mathrm{T}\prod_{j=1}^l \mathrm{C}_j(T'Q_j T'^\dag)\right)
		T'
		\left(\prod_{k=l+2}^{N_2} \mathrm{CZ}_{(2, l+1), (2, k)}^{\nu_{l+1, k}}\right) \mathrm{S}_{(2, l+1)}^{\mu_{l+1}}. \label{SMeq:SMthm:symmetric_Clifford_expression13}
	\end{align}
	By plugging Eq.~\eqref{SMeq:SMthm:symmetric_Clifford_expression13} into Eq.~\eqref{SMeq:SMthm:symmetric_Clifford_expression12}, we get 
	\begin{align}
		U''=\left(\mathrm{T}\prod_{j=1}^{l+1} \mathrm{C}_j(Q'_j)\right)
		T''
		\left(\prod_{j=1}^{l+1} \prod_{k=j+1}^{N_2} \mathrm{CZ}_{(2, j), (2, k)}^{\nu_{j, k}}\right)\left(\prod_{j=1}^{l+1} \mathrm{S}_{(2, j)}^{\mu_j}\right)U^\dag, 
	\end{align}
	where $T'':=T'T$, $Q'_j:=T'Q_j T'^\dag$ for $j\in\{1, 2, ..., l\}$, $Q'_{l+1}:=\prod_{k=1}^{N_3} \mathrm{Z}_{(3, k)}^{\xi_k}$. 
	We have thus completed the process of mathematical induction, which implies that we can take $\{\mu_j\}_{j=1}^{N_2} \in\{0, 1, 2, 3\}^{N_2}$, $\{\nu_{j, k}\}_{1\leq j<k\leq N_2}\in\{0, 1\}^{N_2(N_2-1)/2}$, $T\in\mathcal{C}_{N_3}$ and $\{Q_j\}_{j=1}^{N_2}\in(\pm\mathcal{P}_{N_3}^+)^{N_2}$ satisfying 
	\begin{align}
		\forall m\in\{1, 2, ..., N_1\}\ [U', \mathrm{Z}_{(1, m)}]=[U', \mathrm{X}_{(1, m)}]=0,\ 
		\forall m\in\{1, 2, ..., N_2\}\ [U', \mathrm{Z}_{(2, m)}]=[U', \mathrm{X}_{(2, m)}]=0 \label{SMeq:SMthm:symmetric_Clifford_expression14}
	\end{align}
	with 
	\begin{align}
		U':=\left(\mathrm{T}\prod_{j=1}^{N_2} \mathrm{C}_j(Q_j)\right)
		T\left(\prod_{j, k: 1\leq j<k\leq N_2} \mathrm{CZ}_{(2, j), (2, k)}^{\nu_{j, k}}\right) \left(\prod_{j=1}^{N_2} \mathrm{S}_{(2, j)}^{\mu_j}\right)U^\dag. \label{SMeq:SMthm:symmetric_Clifford_expression15}
	\end{align}
	This implies that 
	\begin{align}
		U=&U'^\dag\left(\mathrm{T}\prod_{j=1}^{N_2} \mathrm{C}_j(Q_j)\right)
		T\left(\prod_{j, k: 1\leq j<k\leq N_2} \mathrm{CZ}_{(2, j), (2, k)}^{\nu_{j, k}}\right) \left(\prod_{j=1}^l \mathrm{S}_{(2, j)}^{\mu_j}\right) \nonumber\\
		=&\left(\mathrm{T}\prod_{j=1}^{N_2} \mathrm{C}_j(U'^\dag Q_j U')\right)
		U'^\dag T\left(\prod_{j, k: 1\leq j<k\leq N_2} \mathrm{CZ}_{(2, j), (2, k)}^{\nu_{j, k}}\right) \left(\prod_{j=1}^l \mathrm{S}_{(2, j)}^{\mu_j}\right). 
	\end{align}
	Eq.~\eqref{SMeq:SMthm:symmetric_Clifford_expression14} implies that $U'$ acts only on the system $\mathrm{A}_3$, and Eq.~\eqref{SMeq:SMthm:symmetric_Clifford_expression15} implies that $U'$ is a Clifford operator. 
	For any $j\in\{1, 2, ..., N_2\}$, the Pauli operator $Q_j$ on the subsystem $\mathrm{A}_3$ satisfies $U'^\dag Q_j U'\in\pm\mathcal{P}_{N_3}^+$, i.e., $U'^\dag Q_j U'=(-1)^{\kappa_j}P_j$ with some $\kappa_j\in\{0, 1\}$ and $P_j\in\mathcal{P}_{N_3}^+$. 
	We note that 
	\begin{align}
		\mathrm{C}_j(U'^\dag Q_j U')
		=\left(\mathrm{C}_j(-I)\right)^{\kappa_j}\mathrm{C}_j(P_j)
		=\mathrm{S}_{(2, j)}^{2\kappa_j}\mathrm{C}_j(P_j). 
	\end{align}
	We therefore get 
	\begin{align}
		U=&\left(\mathrm{T}\prod_{j=1}^{N_2} \mathrm{C}_j(P_j)\right)
		U'^\dag T\left(\prod_{j, k: 1\leq j<k\leq N_2} \mathrm{CZ}_{(2, j), (2, k)}^{\nu_{j, k}}\right) \left(\prod_{j=1}^l \mathrm{S}_{(2, j)}^{\mu_j+2\kappa_j}\right) \nonumber\\
		=&\left(\mathrm{T}\prod_{j=1}^{N_2} \mathrm{C}_j(P_j)\right)
		V\left(\prod_{j, k: 1\leq j<k\leq N_2} \mathrm{CZ}_{(2, j), (2, k)}^{\nu_{j, k}}\right) \left(\prod_{j=1}^l \mathrm{S}_{(2, j)}^{\mu'_j}\right), 
	\end{align}
	where $V:=U'^\dag T\in\mathcal{C}_{N_3}$ and $\mu'_j\in\{0, 1, 2, 3\}$ is defined by $\mu'_j\equiv\mu_j+2\kappa_j\ (\mathrm{mod}\ 4)$.

	Next, we prove that the expression of $U\in\mathcal{C}_{N, \mathcal{R}}$ by Eq.~\eqref{SMeq:C_R_expression} is unique. 
	We suppose that $\{\mu_j\}_{j=1}^{N_2}, \{\mu'_j\}_{j=1}^{N_2} \in\{0, 1, 2, 3\}^{N_2}$, $\{\nu_{j, k}\}_{1\leq j<k\leq N_2}, \{\nu'_{j, k}\}_{1\leq j<k\leq N_2}\in\{0, 1\}^{N_2(N_2-1)/2}$, $V, V'\in\mathcal{C}_{N_3}$ and $\{P_j\}_{j=1}^{N_2}, \{P'_j\}_{j=1}^{N_2}\in(\mathcal{P}_{N_3}^+)^{N_2}$ satisfy 
	\begin{align}
		&\left(\mathrm{T}\prod_{j=1}^{N_2} \mathrm{C}_j(P_j)\right)
		V\left(\prod_{j, k: 1\leq j<k\leq N_2} \mathrm{CZ}_{(2, j), (2, k)}^{\nu_{j, k}}\right)\left(\prod_{j=1}^{N_2} \mathrm{S}_{(2, j)}^{\mu_j}\right) \nonumber\\
		&\ \ \ \ \ \ \ \ \ \ 
		\ \ \ \ \ \ \ \ \ \ =\left(\mathrm{T}\prod_{j=1}^{N_2} \mathrm{C}_j(P'_j)\right)
		V'\left(\prod_{j, k: 1\leq j<k\leq N_2} \mathrm{CZ}_{(2, j), (2, k)}^{\nu'_{j, k}}\right)\left(\prod_{j=1}^{N_2} \mathrm{S}_{(2, j)}^{\mu'_j}\right). 
	\end{align}
	We consider applying these two operators to vectors given in the form of $\ket{\Psi}\otimes\ket{x_1 x_2 \cdots x_{N_2}}\otimes\ket{\Xi}$, where $\ket{\Psi}$, $\ket{x_1 x_2 \cdots x_{N_2}}:=\ket{x_1}\otimes\ket{x_2}\otimes\cdots\otimes\ket{x_{N_2}}$ and $\ket{\Xi}$ are vectors of the system $\mathrm{A}_1$, $\mathrm{A}_2$ and $\mathrm{A}_3$, respectively. 
	We first consider the case when $x_l=0$ for all $l\in\{1, 2, ..., N_2\}$. 
	Then, we get 
	\begin{align}
		\ket{\Psi}\otimes\ket{x_1 x_2 \cdots x_{N_2}}\otimes V\ket{\Xi} 
		=\ket{\Psi}\otimes\ket{x_1 x_2 \cdots x_{N_2}}\otimes V'\ket{\Xi}. 
	\end{align} 
	Since this holds for all $\ket{\Xi}\in\mathcal{H}_3$, we get $V=V'$. 
	We next take arbitrary $j\in\{1, 2, ..., N_2\}$ and consider the case when $x_l=\delta_{l, j}$ for all $l\in\{1, 2, ..., N_2\}$. 
	Then, we get 
	\begin{align}
		\ket{\Psi}\otimes\ket{x_1 x_2 \cdots x_{N_2}}\otimes (i^{\mu_j}P_j) V\ket{\Xi} 
		=\ket{\Psi}\otimes\ket{x_1 x_2 \cdots x_{N_2}}\otimes (i^{\mu'_j}P'_j) V\ket{\Xi}. 
	\end{align} 
	Since this holds for all $\ket{\Xi}\in\mathcal{H}_3$, we get $i^{\mu_j}P_j=i^{\mu'_j}P'_j$. 
	Since $\mu_j, \mu'_j\in\{0, 1, 2, 3\}$ and $P_j, P'_j\in\mathcal{P}_{N_3}^+$, this implies that $\mu_j=\mu'_j$ and $P_j=P'_j$. 
	We finally take arbitrary $j, k\in\{1, 2, ..., N_2\}$ satisfying $j<k$ and consider the case when $x_l=\delta_{l, j}+\delta_{l, k}$ for all $l\in\{1, 2, ..., N_2\}$. 
	Then, we get 
	\begin{align}
		\ket{\Psi}\otimes(-1)^{\nu_{j, k}}\ket{x_1 x_2 \cdots x_{N_2}}\otimes (i^{\mu_j}P_j) V\ket{\Xi} 
		=\ket{\Psi}\otimes(-1)^{\nu'_{j, k}}\ket{x_1 x_2 \cdots x_{N_2}}\otimes (i^{\mu_j}P_j) V\ket{\Xi}. 
	\end{align} 
	This implies that $(-1)^{\nu_{j, k}}=(-1)^{\nu'_{j, k}}$. 
	Since $\nu_{j, k}, \nu'_{j, k}\in\{0, 1\}$, we get $\nu_{j, k}=\nu'_{j, k}$. 
	We therefore get $\mu_j=\mu'_j$, $\nu_{j, k}=\nu'_{j, k}$, $V=V'$ and $P_j=P'_j$ for all $j, k\in\{1, 2, ..., N_2\}$. 
	This means that the expression of $U\in\mathcal{C}_{N, \mathcal{R}}$ in the form of Eq.~\eqref{SMeq:C_R_expression} is unique. 
\end{proof}

\subsection{U(1) and SU(2) symmetries}

Next, we take U(1) and SU(2) symmetries given by Eqs.~(7) and (8) in the main text as examples of non-Pauli symmetries, and present complete and unique constructions for the symmetric Clifford groups. 
In both cases, every symmetric Clifford operator can be written as the product of a permutation operator and a Pauli-symmetric Clifford operator as shown in Fig.~\ref{fig:nonPauli_circuit} in the main text.

\begin{theorem} \label{SMthm:nonPauli_Clifford_expression}
    (Construction of the U(1) and SU(2)-symmetric Clifford groups.) 
	Let $N\in\mathbb{N}$, and $\mathcal{G}_1$ and $\mathcal{G}_2$ be given by 
	\begin{align}
		&\mathcal{G}_1=\left\{\left(e^{i\theta \mathrm{Z}}\right)^{\otimes N}\ |\ \theta\in\mathbb{R}\right\}, \\ 
		&\mathcal{G}_2=\left\{\left(e^{i(\theta_\mathrm{X} \mathrm{X}+\theta_\mathrm{Y} \mathrm{Y}+\theta_\mathrm{Z} \mathrm{Z})}\right)^{\otimes N}\ |\ \theta_\mathrm{X}, \theta_\mathrm{Y}, \theta_\mathrm{Z}\in\mathbb{R}\right\}. 
	\end{align}
	Then, for any $U\in\mathcal{C}_{N, \mathcal{G}_1}$, there uniquely exist $\{\mu_j\}_{j=1}^N\in\{0, 1, 2, 3\}^N$, $\{\nu_{j, k}\}_{1\leq j<k\leq N}\in\{0, 1\}^{N(N-1)/2}$, $\sigma\in\mathfrak{S}_N$ and $c\in\mathcal{U}_0$ such that 
	\begin{align}
		U=c\left(\prod_{j, k: 1\leq j<k\leq N} \mathrm{CZ}_{j, k}^{\nu_{j, k}}\right)\left(\prod_{j=1}^{N} \mathrm{S}_j^{\mu_j}\right)K_\sigma, \label{SMeq:U1_Clifford_expression}
	\end{align}
	and for any $U\in\mathcal{C}_{N, \mathcal{G}_2}$, there uniquely exist $\sigma\in\mathfrak{S}_N$ and $c\in\mathcal{U}_0$ such that 
	\begin{align}
		U=cK_\sigma, \label{SMeq:SU2_Clifford_expression}
	\end{align}
	where $K_\sigma$ is defined as the permutation operator on $\mathcal{H}$ defined by 
	\begin{align}
		K_\sigma\left(\bigotimes_{j=1}^N \ket{\psi_j}\right):=\bigotimes_{j=1}^N \ket{\psi_{\sigma^{-1}(j)}}. \label{SMeq:permutation_single_copy}
	\end{align}
\end{theorem}

\begin{proof}
	First, we consider the expression for $\mathcal{C}_{N, \mathcal{G}_1}$. 
	For the completeness of the expression, we take arbitrary $U\in\mathcal{C}_{N, \mathcal{G}_1}$ and show that $U$ can be written in the form of Eq.~\eqref{SMeq:U1_Clifford_expression}. 
	Since $U$ is $\mathcal{G}_1$-symmetric, we have 
	\begin{align}
		U \left(e^{i\theta\mathrm{Z}}\right)^{\otimes N}U^\dag=\left(e^{i\theta\mathrm{Z}}\right)^{\otimes N} 
	\end{align}
	for all $\theta\in\mathbb{R}$. 
	By taking the derivative at $\theta=0$, we get 
	\begin{align}
		\sum_{j=1}^N U\mathrm{Z}_j U^\dag
		=U\left(\sum_{j=1}^N \mathrm{Z}_j\right)U^\dag
		=\sum_{j=1}^N \mathrm{Z}_j. 
	\end{align}
	By noting that the both sides are the sum of $N$ different Pauli operators with equal coefficients, we know that 
	\begin{align}
		U\mathrm{Z}_j U^\dag=\mathrm{Z}_{\sigma(j)}\ \forall j\in\{1, 2, ..., N\} \label{SMeq:SMthm:nonPauli_Clifford_expression1}
	\end{align}
	with some $\sigma\in\mathfrak{S}_N$. 
	We define $U':=U K_\sigma^\dag$. 
	Then, we get 
	\begin{align}
		U'\mathrm{Z}_j U'^\dag
		=UK_\sigma^\dag\mathrm{Z}_j K_\sigma U^\dag  
		=U\mathrm{Z}_{\sigma^{-1}(j)} U^\dag 
		=\mathrm{Z}_j\ \forall j\in\{1, 2, ..., N\}. 
	\end{align}
	By Theorem~\ref{SMthm:symmetric_Clifford_expression}, $U'$ can be written as 
	\begin{align}
		U'=c\left(\prod_{j, k: 1\leq j<k\leq N} \mathrm{CZ}_{j, k}^{\nu_{j, k}}\right)\left(\prod_{j=1}^{N} \mathrm{S}_j^{\mu_j}\right) 
	\end{align}
	with some $\mu_j\in\{0, 1, 2, 3\}$, $\nu_{j, k}\in\{0, 1\}$ and $c\in\mathcal{U}_0$. 
	Note that the term $c$ corresponds to the term $V$ in Eq.~\eqref{SMeq:Pauli_Clifford_general} when $N_3=0$. 
	We therefore get 
	\begin{align}
		U=c\left(\prod_{j, k: 1\leq j<k\leq N} \mathrm{CZ}_{j, k}^{\nu_{j, k}}\right)\left(\prod_{j=1}^{N} \mathrm{S}_j^{\mu_j}\right)K_\sigma. 
	\end{align}
	For the uniqueness of this expression, for any $U\in\mathcal{C}_{N, \mathcal{G}_1}$, we take arbitrary two representations 
	\begin{align}
		U
		=c\left(\prod_{j, k: 1\leq j<k\leq N} \mathrm{CZ}_{j, k}^{\nu_{j, k}}\right)\left(\prod_{j=1}^N \mathrm{S}_j^{\mu_j}\right)K_\sigma 
		=c'\left(\prod_{j, k: 1\leq j<k\leq N} \mathrm{CZ}_{j, k}^{\nu'_{j, k}}\right)\left(\prod_{j=1}^N \mathrm{S}_j^{\mu'_j}\right)K_{\sigma'} \label{SMeq:SMthm:nonPauli_Clifford_expression2}
	\end{align}
	with $\sigma, \sigma'\in\mathfrak{S}_N$, $\{\mu_j\}_{j=1}^N, \{\mu'_j\}_{j=1}^N \in\{0, 1, 2, 3\}^N$, $\{\nu_{j, k}\}_{1\leq j<k\leq N}, \{\nu'_{j, k}\}_{1\leq j<k\leq N}\in\{0, 1\}^{N(N-1)/2}$ and $c, c'\in\mathcal{U}_0$, and show that all these parameters for the two representations are the same. 
	We suppose that $\sigma\neq\sigma'$. 
	Then, we can take $j\in\{1, 2, ..., N\}$ such that $\sigma(j)\neq\sigma'(j)$. 
	By using the two representations for $U$, we get 
	\begin{align}
		\mathrm{Z}_{\sigma(j)}=U\mathrm{Z}_j U^\dag=\mathrm{Z}_{\sigma'(j)}, 
	\end{align}
	but this is a contradiction. 
	We therefore get $\sigma=\sigma'$. 
	By Plugging this into Eq.~\eqref{SMeq:SMthm:nonPauli_Clifford_expression2}, we get 
	\begin{align}
		c\left(\prod_{j, k: 1\leq j<k\leq N} \mathrm{CZ}_{j, k}^{\nu_{j, k}}\right)\left(\prod_{j=1}^N \mathrm{S}_j^{\mu_j}\right) 
		=c'\left(\prod_{j, k: 1\leq j<k\leq N} \mathrm{CZ}_{j, k}^{\nu'_{j, k}}\right)\left(\prod_{j=1}^N \mathrm{S}_j^{\mu'_j}\right). 
	\end{align}
	By Theorem~\ref{SMthm:symmetric_Clifford_expression}, we get $\mu_j=\mu'_j$ for all $j\in\{1, 2, ..., N\}$, $\nu_{j, k}=\nu'_{j, k}$ for all $j, k\in\{1, 2, ..., N\}$ satisfying $j<k$, and $c=c'$.

	Next, we consider the expression for $\mathcal{C}_{N, \mathcal{G}_2}$. 
	For the completeness of the expression, we take arbitrary $U\in\mathcal{C}_{N, \mathcal{G}_2}$ and show that $U$ can be written in the form of Eq.~\eqref{SMeq:SU2_Clifford_expression}. 
	In the same way as the case above, we have Eq.~\eqref{SMeq:SMthm:nonPauli_Clifford_expression1} and 
	\begin{align}
		U\mathrm{X}_j U^\dag=\mathrm{X}_{\sigma'(j)} 
	\end{align}
	with some $\sigma'\in\mathfrak{S}_N$. 
	We suppose that $\sigma\neq\sigma'$. 
	Then, we can take $j\in\{1, 2, ..., N\}$ such that $\sigma(j)\neq\sigma'(j)$. 
	This implies that 
	\begin{align}
		[\mathrm{Z}_j, \mathrm{X}_j] 
		=U^\dag[U\mathrm{Z}_j U^\dag, U\mathrm{X}_j U^\dag]U
		=U^\dag[\mathrm{Z}_{\sigma(j)}, \mathrm{X}_{\sigma'(j)}]U
		=0
	\end{align}
	but this contradicts with $[\mathrm{Z}_j, \mathrm{X}_j]\neq 0$. 
	We thus get $\sigma=\sigma'$. 
	We define $U':=UK_\sigma^\dag$. 
	Then, we get $U'\mathrm{Z}_j U'^\dag=\mathrm{Z}_j$ and $U'\mathrm{X}_j U'^\dag=\mathrm{X}_j$ for all $j\in\{1, 2, ..., N\}$. 
	Since such $U'$ is restricted to $U'=cI$ with some $c\in\mathcal{U}_0$, we get
	\begin{align}
		U=cK_\sigma. 
	\end{align} 
	The uniqueness of this expression is trivial.  
\end{proof}

\section{Proof of Unitary $1$-Designs in Theorem~\ref{thm:1design}} \label{SMsec:1design}

In this appendix, we take U(1) and SU(2) symmetries given by Eqs.~(7) and (8) in the main text as examples of non-Pauli symmetries, and show that the symmetric Clifford groups are symmetric unitary $1$-designs for those symmetries. 
This corresponds to the former half of the statement of Theorem~\ref{thm:1design} in the main text. 
The proof method is similar to the one in the ``if'' part of Theorem~\ref{SMthm:main}.

\begin{theorem} \label{SMthm:1design}
    (1-design part of Theorem~\ref{thm:1design}.) 
	Let $N\in\mathbb{N}$, $\mathcal{G}_1$ and $\mathcal{G}_2$ be defined by 
	\begin{align}
		&\mathcal{G}_1=\left\{\left(e^{i\theta \mathrm{Z}}\right)^{\otimes N}\ |\ \theta\in\mathbb{R}\right\}, \\
		&\mathcal{G}_2=\left\{\left(e^{i(\theta_\mathrm{X} \mathrm{X}+\theta_\mathrm{Y} \mathrm{Y}+\theta_\mathrm{Z} \mathrm{Z})}\right)^{\otimes N}\ |\ \theta_\mathrm{X}, \theta_\mathrm{Y}, \theta_\mathrm{Z}\in\mathbb{R}\right\}. 
	\end{align}
	Then, $\mathcal{C}_{N, \mathcal{G}_j}$ is a $\mathcal{G}_j$-symmetric unitary $1$-design for $j=1, 2$. 
\end{theorem}

\begin{proof}
	First, we consider the symmetry given by $\mathcal{G}_2$. 
	We define $\mathcal{D}$ by 
	\begin{align}
		\mathcal{D}(L):=\frac{1}{N!}\sum_{\sigma\in\mathfrak{S}_N} K_\sigma LK_\sigma^\dag\ \forall L\in\mathcal{L}(\mathcal{H}), \label{SMeq:permutation_mixture}
	\end{align}
	where $K_\sigma$ is the permutation operator defined by Eq.~\eqref{SMeq:permutation_single_copy}. 
	Then, for any $L\in\mathcal{L}(\mathcal{H})$ and $\sigma'\in\mathfrak{S}_N$, we get 
	\begin{align}
		K_{\sigma'}\mathcal{D}(L)K_{\sigma'}^\dag
		=\frac{1}{N!}\sum_{\sigma\in\mathfrak{S}_N} K_{\sigma'}K_\sigma LK_\sigma^\dag K_{\sigma'}^\dag
		=\frac{1}{N!}\sum_{\sigma\in\mathfrak{S}_N} K_{\sigma'\sigma} LK_{\sigma'\sigma}^\dag 
		=\frac{1}{N!}\sum_{\sigma\in\mathfrak{S}_N} K_\sigma LK_\sigma^\dag 
		=\mathcal{D}(L). 
	\end{align}
	This implies that 
	\begin{align}
		\mathcal{D}(L)\in\mathrm{span}\{V^{\otimes N}\ |\ V\in\mathcal{U}_1\} \label{SMeq:SMthm:1design1}
	\end{align}
	by Theorem~7.11 of Ref.~\cite{watrous2018theory}. 
	It is therefore sufficient to show that $UV^{\otimes N}U^\dag=V^{\otimes N}$ for all $U\in\mathcal{U}_{N, \mathcal{G}_2}$ and $V\in\mathcal{U}_1$, in order to show that $\mathcal{D}(L)$ satisfies $U\mathcal{D}(L)U^\dag=\mathcal{D}(L)$ for all $U\in\mathcal{U}_{N, \mathcal{G}_2}$. 
	Take arbitrary $U\in\mathcal{U}_{N, \mathcal{G}_2}$ and $V\in\mathcal{U}_1$. 
	Since $U$ satisfies $[U, (e^{i\theta \mathrm{X}})^{\otimes N}]=[U, (e^{i\theta \mathrm{Y}})^{\otimes N}]=[U, (e^{i\theta \mathrm{Z}})^{\otimes N}]=0$ for all $\theta\in\mathbb{R}$, by taking the derivative at $\theta=0$, we get 
	\begin{align}
		[U, \mathrm{X}_\mathrm{tot}]=[U, \mathrm{Y}_\mathrm{tot}]=[U, \mathrm{Z}_\mathrm{tot}]=0, \label{SMeq:SMthm:1design2}
	\end{align}
	where $\mathrm{X}_\mathrm{tot}:=\sum_{k=1}^N \mathrm{X}_k$, $\mathrm{Y}_\mathrm{tot}:=\sum_{k=1}^N \mathrm{Y}_k$ and $\mathrm{Z}_\mathrm{tot}:=\sum_{k=1}^N \mathrm{Z}_k$. 
	Since $V$ is a unitary operator on a single qubit, $V$ can be written as $V=e^{i(\phi_\mathrm{I} \mathrm{I}+\phi_\mathrm{X} \mathrm{X}+\phi_\mathrm{Y} \mathrm{Y}+\phi_\mathrm{Z} \mathrm{Z})}$ with some $\phi_\mathrm{I}, \phi_\mathrm{X}, \phi_\mathrm{Y}, \phi_\mathrm{Z}\in\mathbb{R}$. 
	This implies that 
	\begin{align}
		V^{\otimes N}=e^{i(N\phi_\mathrm{I} I+\phi_\mathrm{X} \mathrm{X}_\mathrm{tot}+\phi_\mathrm{Y} \mathrm{Y}_\mathrm{tot}+\phi_\mathrm{Z} \mathrm{Z}_\mathrm{tot})}. \label{SMeq:SMthm:1design3}
	\end{align}
	By Eqs.~\eqref{SMeq:SMthm:1design2} and \eqref{SMeq:SMthm:1design3}, we get $[V^{\otimes N}, U]=0$ for all $U\in\mathcal{U}_{N, \mathcal{G}_2}$ and $V\in\mathcal{U}_1$. 
	By Eq.~\eqref{SMeq:SMthm:1design1}, this implies that $[\mathcal{D}(L), U]=0$ for all $U\in\mathcal{U}_{N, \mathcal{G}_2}$. 
	By Lemma~\ref{SMlem:uniform_mixture_property1}, for $\mathcal{X}=\mathcal{C}_{N, \mathcal{G}_2}$ and $\mathcal{U}_{N, \mathcal{G}_2}$, we get 
	\begin{align}
		\Phi_{1, \mathcal{X}}(L)
		=\Phi_{1, \mathcal{X}}(\mathcal{D}(L))
		=\int_{U\in\mathcal{X}} U\mathcal{D}(L)U^\dag d\mu_\mathcal{X}(U) 
		=\int_{U\in\mathcal{X}} \mathcal{D}(L) d\mu_\mathcal{X}(U)
		=\mathcal{D}(L)
	\end{align}
	for all $L\in\mathcal{L}(\mathcal{H})$. 
	Since this holds for $\mathcal{X}=\mathcal{C}_{N, \mathcal{G}_2}$ and $\mathcal{U}_{N, \mathcal{G}_2}$, this implies that $\Phi_{1, \mathcal{C}_{N, \mathcal{G}_2}}=\Phi_{1, \mathcal{U}_{N, \mathcal{G}_2}}$, or equivalently, $\mathcal{C}_{N, \mathcal{G}_2}$ is a $\mathcal{G}_2$-symmetric unitary $1$-design.

	Next, we consider the symmetry given by $\mathcal{G}_1$. 
	We define $\mathcal{D}$ by Eq.~\eqref{SMeq:permutation_mixture}. 
	Then, $\mathcal{D}$ satisfies $\mathcal{D}\in\mathfrak{C}_{1, \mathcal{G}}$. 
	By the same argument as the case of $\mathcal{G}_2$, we get $\mathcal{D}(L)\in\mathrm{span}\{V^{\otimes N}\ |\ V\in\mathcal{U}_1\}$. 
	We define $\mathcal{D}'\in\mathfrak{C}_{1, \mathcal{G}_1}$ by 
	\begin{align}
		\mathcal{D}'(L):=\frac{1}{2^N}\sum_{(\mu_1, \mu_2, ..., \mu_N)\in\{0, 1\}^N} \left(\prod_{k=1}^N \mathrm{Z}_k^{\mu_k}\right)L\left(\prod_{k=1}^N \mathrm{Z}_k^{\dag\mu_k}\right)\ \forall L\in\mathcal{L}(\mathcal{H}). 
	\end{align}
	and $\mathcal{D}''$ by $\mathcal{D}'':=\mathcal{D}'\circ\mathcal{D}$. 
	Then, we get $\mathcal{D}''\in\mathfrak{C}_{1, \mathcal{G}_1}$ by Lemma~\ref{SMeq:symmetric_Clifford_mixture_composition} in Appendix~\ref{SMsec:technical}, and 
	\begin{align}
		\mathcal{D}''(L)\in\mathrm{span}\{(V+\mathrm{Z}V\mathrm{Z})^{\otimes N}\ |\ V\in\mathcal{U}_1\}. \label{SMeq:SMthm:1design4}
	\end{align}
	It is therefore sufficient to prove that $[(V+\mathrm{Z}V\mathrm{Z})^{\otimes N}, U]=0$ for all $U\in\mathcal{U}_{N, \mathcal{G}_1}$ and $V\in\mathcal{U}_1$, in order to show that $U\mathcal{D}''(L)U^\dag=\mathcal{D}''(L)$ for all $U\in\mathcal{U}_{N, \mathcal{G}_1}$. 
	Take arbitrary $U\in\mathcal{U}_{N, \mathcal{G}}$ and $V\in\mathcal{U}_1$. 
	In the same way as Eq.~\eqref{SMeq:SMthm:1design2}, we can prove that $U$ satisfies 
	\begin{align}
		[U, \mathrm{Z}_\mathrm{tot}]=0. \label{SMeq:SMthm:1design5}
	\end{align} 
	Since $V$ is a unitary operator on a single qubit, we note that $V=e^{i(\phi_\mathrm{I} \mathrm{I}+\phi_\mathrm{X} \mathrm{X}+\phi_\mathrm{Y} \mathrm{Y}+\phi_\mathrm{Z} \mathrm{Z})}$, i.e., $V=e^{i\phi_\mathrm{I}}\cos(\phi)\mathrm{I}+i\sin(\phi)(\phi_\mathrm{X} \mathrm{X}+\phi_\mathrm{Y} \mathrm{Y}+\phi_\mathrm{Z} \mathrm{Z})/\phi$ with some $\phi_\mathrm{I}, \phi_\mathrm{X}, \phi_\mathrm{Y}, \phi_\mathrm{Z}\in\mathbb{R}$, where $\phi:=\sqrt{\phi_\mathrm{X}^2+\phi_\mathrm{Y}^2+\phi_\mathrm{Z}^2}$. 
	We thus get $V+\mathrm{Z}V\mathrm{Z}=2e^{i\phi_\mathrm{I}}[\cos(\phi)\mathrm{I}+i\sin(\phi)\cdot\phi_\mathrm{Z} \mathrm{Z}/\phi]$. 
	We take $r\geq0$ and $\psi\in\mathbb{R}$ such that $r\cos(\psi)=\cos(\phi)$ and $r\sin(\psi)=\sin(\phi)\cdot\phi_\mathrm{Z}/\phi$. 
	Then we get $V+\mathrm{Z}V\mathrm{Z}=2r e^{i\phi_\mathrm{I}}e^{i\psi\mathrm{Z}}$. 
	This implies that 
	\begin{align}
		(V+\mathrm{Z}V\mathrm{Z})^{\otimes N}=(2re^{i\phi_\mathrm{I}})^N e^{i\psi\mathrm{Z}_\mathrm{tot}}. \label{SMeq:SMthm:1design6}
	\end{align}
	By Eqs.~\eqref{SMeq:SMthm:1design5} and \eqref{SMeq:SMthm:1design6}, we get $[(V+\mathrm{Z}V\mathrm{Z})^{\otimes N}, U]=0$ for all $U\in\mathcal{U}_{N, \mathcal{G}_1}$ and $V\in\mathcal{U}_1$. 
	By Eq.~\eqref{SMeq:SMthm:1design4}, this implies that $[\mathcal{D}''(L), U]=0$ for all $U\in\mathcal{U}_{N, \mathcal{G}_1}$. 
	By the same argument as the case of $\mathcal{G}_2$, we know that $\mathcal{C}_{N, \mathcal{G}_1}$ is a $\mathcal{G}_1$-symmetric unitary $1$-design. 
\end{proof}

We note that there exists a group $\mathcal{G}$ such that the $\mathcal{G}$-symmetric Clifford group $\mathcal{C}_{N, \mathcal{G}}$ is not even a $\mathcal{G}$-symmetric unitary $1$-design. 
As a simple example, we can take $N=1$ and $\mathcal{G}=\{e^{i\theta(\alpha\mathrm{Z}+\beta\mathrm{X})}\ |\ \theta\in\mathbb{R}\}$ with $\alpha, \beta\in\mathbb{R}$ satisfying $\alpha>\beta>0$. 
In this case, we can show that $\mathcal{C}_{1, \mathcal{G}}$ is not a $\mathcal{G}$-symmetric unitary $1$-design as follows: 
We take arbitrary $U\in\mathcal{C}_{1, \mathcal{G}}$. 
Then, $U$ satisfies 
\begin{align}
	Ue^{i\theta(\alpha\mathrm{Z}+\beta\mathrm{X})}U^\dag=e^{i\theta(\alpha\mathrm{Z}+\beta\mathrm{X})} 
\end{align}
for all $\theta\in\mathbb{R}$. 
By taking the derivative at $\theta=0$, we get 
\begin{align}
	\alpha U\mathrm{Z}U^\dag+\beta U\mathrm{X}U^\dag=U(\alpha\mathrm{Z}+\beta\mathrm{X})U^\dag=\alpha\mathrm{Z}+\beta\mathrm{X}. 
\end{align}
By noting that both $U\mathrm{Z}U^\dag$ and $U\mathrm{X}U^\dag$ are Pauli operators and $\alpha>\beta>0$, we get 
\begin{align}
	U\mathrm{Z}U^\dag=\mathrm{Z},\ U\mathrm{X}U^\dag=\mathrm{X}\ \forall U\in\mathcal{C}_{1, \mathcal{G}}. 
\end{align}
By Lemma~\ref{SMlem:unitary_design_commutativity}, we get 
\begin{align}
	U\mathrm{Z}U^\dag=\mathrm{Z},\ U\mathrm{X}U^\dag=\mathrm{X}\ \forall U\in\mathcal{U}_{1, \mathcal{G}}. 
\end{align}
This implies that $U=e^{i\theta}I$ with some $\theta\in\mathbb{R}$ for all $U\in\mathcal{U}_{1, \mathcal{G}}$, i.e., $\mathcal{U}_{1, \mathcal{G}}\subset\mathcal{U}_0 I$, but this contradicts with $(\alpha\mathrm{Z}+\beta\mathrm{X})/\sqrt{\alpha^2+\beta^2}\in\mathcal{U}_{1, \mathcal{G}}$. 
We therefore know that $\mathcal{C}_{1, \mathcal{G}}$ is not a $\mathcal{G}$-symmetric unitary $1$-design.

\section{Disproof of Unitary $2$-Designs in Theorem~\ref{thm:1design}} \label{SMsec:2design}

In this appendix, we show that for a certain class of non-Pauli symmetries, the symmetric Clifford group is not a symmetric unitary $2$-design. 
This is a generalized statement of the latter half of Theorem~\ref{thm:1design} in the main text. 
Concretely, we consider the setup where a system consists of $M\geq 2$ copies of $n$ qubits and a symmetry group $\mathcal{G}$ is given by 
\begin{align}
	\mathcal{G}=\left\{F^{\otimes M}\ |\ F\in\mathcal{F}\right\} \label{SMeq:tensor_product_group}
\end{align}
with a connected Lie subgroup $\mathcal{F}$ of $\mathcal{U}_n$ on $n$ qubits. 
In a physical perspective, this symmetry represents the conservation of the total $M$ quantities each of which is defined on $n$ qubits. 
In a mathematical perspective, $\mathcal{G}$ is isomorphic to $\mathcal{F}$ and the conserved quantities on $n$ qubits are elements of the Lie algebra $\mathfrak{f}$ of $\mathcal{F}$. 
This form of symmetries includes the U(1) and SU(2) symmetries given by Eqs.~(7) and (8) in the main text as special cases. 
In fact, those two symmetry groups are represented with $M=N$, $n=1$ and $\mathcal{F}$ given by 
\begin{align}
	\mathcal{F}=\left\{e^{i\theta\mathrm{Z}}\ |\ \theta\in\mathbb{R}\right\},\ 
	\mathcal{F}=\left\{e^{i(\theta_\mathrm{X} \mathrm{X}+\theta_\mathrm{Y} \mathrm{Y}+\theta_\mathrm{Z} \mathrm{Z})}\ |\ \theta_\mathrm{X}, \theta_\mathrm{Y}, \theta_\mathrm{Z}\in\mathbb{R}\right\}. 
\end{align}

\begin{theorem} \label{SMthm:2design}
    (Generalized version of the 2-design part of Theorem~\ref{thm:1design}.)
	Let $N\in\mathbb{N}$ and $\mathcal{G}$ be a subgroup of $\mathcal{U}_N$ given in the form Eq.~\eqref{SMeq:tensor_product_group} with $M\geq 2$, $n\in\mathbb{N}$ and a connected Lie subgroup $\mathcal{F}$ of $\mathcal{U}_n$. 
	Then, $\mathcal{C}_{N, \mathcal{G}}$ is a $\mathcal{G}$-symmetric unitary $2$-design if and only if $\mathcal{U}_{N, \mathcal{G}}=\mathcal{U}_N$. 
\end{theorem}

Since we are going to deal with many Hilbert spaces in the proof, we define the notations for Hilbert spaces as follows: 
In the context of unitary $2$-designs, we consider two copies of the Hilbert space $\mathcal{H}$ associated with $N$ qubits, which we denote by $\mathcal{H}^1$ and $\mathcal{H}^2$. 
The symmetry $\mathcal{G}$ naturally induces the decomposition of the Hilbert space $\mathcal{H}^j$ into $M$ parts, which we denote the $k$th part by $\mathcal{H}_k^j$. 
We denote $\mathcal{H}_k^j$ simply by $\mathcal{H}_k$ when we need not specify $j$.

\begin{proof}
	Since $\mathcal{C}_N$ is a unitary $3$-design~\cite{webb2016clifford, zhu2017multiqubit}, the ``if'' part is trivial. 
	In the following, we consider the ``only if'' part. 
	Suppose that $\mathcal{C}_{N, \mathcal{G}}$ is a $\mathcal{G}$-symmetric unitary $2$-design. 
	We define $\mathfrak{f}$ as the Lie algebra of $\mathcal{F}$, and we are going to prove that $\mathfrak{f}\subset\{aI\ |\ a\in\mathbb{R}\}$. 
	We take arbitrary $A\in\mathfrak{f}$. 
	Since $e^{i\theta A}\in\mathcal{F}$ for all $\theta\in\mathbb{R}$, we have 
	\begin{align}
		\bigotimes_{k=1}^M e^{i\theta {A^{(\mathcal{H}_k)}}}\in\mathcal{G}. 
	\end{align}
	We take arbitrary $U\in\mathcal{C}_{N, \mathcal{G}}$. 
	Then, $U$ satisfies 
	\begin{align}
		\left[U, \bigotimes_{k=1}^M e^{i\theta {A^{(\mathcal{H}_k)}}}\right]=0. 
	\end{align}
	By taking the derivative at $\theta=0$, we get 
	\begin{align}
		\left[U, \sum_{k=1}^M A^{(\mathcal{H}_k)}\right]=0. \label{SMeq:SMthm:2design1}
	\end{align}
	We define $\beta_P$ as the expansion coefficient of $P$ in $A$ in the Pauli basis, i.e., $\beta_P:=\mathrm{tr}(AP)/2^n$. 
	Then, $A$ can be written as 
	\begin{align}
		A=\sum_{P\in\mathcal{P}_n^+} \beta_P P. \label{SMeq:SMthm:2design2}
	\end{align}
	By plugging Eq.~\eqref{SMeq:SMthm:2design2} into Eq.~\eqref{SMeq:SMthm:2design1}, we get 
	\begin{align}
		\left[U, \sum_{k=1}^M \sum_{P\in\mathcal{P}_n^+} \beta_P P^{(\mathcal{H}_k)}\right]=0, 
	\end{align}
	or equivalently, 
	\begin{align}
		U\left(\sum_{k=1}^M \sum_{P\in\mathcal{P}_n^+} \beta_P P^{(\mathcal{H}_k)}\right)U^\dag
		=\sum_{k=1}^M \sum_{P\in\mathcal{P}_n^+} \beta_P P^{(\mathcal{H}_k)}. 
	\end{align}
	We define 
	\begin{align}
		B:=\sum_{k=1}^M \sum_{P\in\mathcal{P}_n^+} \beta_P^2 P^{(\mathcal{H}_k^1)}\otimes P^{(\mathcal{H}_k^2)}. \label{SMeq:SMthm:2design3}
	\end{align}
	Then, we get $U^{\otimes 2}BU^{\dag\otimes 2}=B$ by Lemma~\ref{SMlem:Clifford_invariant_construction}. 
	Since this holds for all $U\in\mathcal{C}_{N, \mathcal{G}}$ and $\mathcal{C}_{N, \mathcal{G}}$ is a $\mathcal{G}$-symmetric unitary $2$-design, we get $U^{\otimes 2}BU^{\dag\otimes 2}=B$ for all $U\in\mathcal{U}_{N, \mathcal{G}}$ by Lemma~\ref{SMlem:unitary_design_commutativity}. 
	By noting that the SWAP operator $\mathrm{SWAP}^{(\mathcal{H}_1, \mathcal{H}_2)}$ between the Hilbert space $\mathcal{H}_1$ and $\mathcal{H}_2$ satisfies $[\mathrm{SWAP}^{(\mathcal{H}_1, \mathcal{H}_2)}, G]=0$ for all $G\in\mathcal{G}$, we know that $e^{i\theta\cdot\mathrm{SWAP}^{(\mathcal{H}_1, \mathcal{H}_2)}}\in\mathcal{U}_{N, \mathcal{G}}$ for all $\theta\in\mathbb{R}$. 
	We thus get 
	\begin{align}
		\left(\bigotimes_{j=1}^2 e^{i\theta\cdot\mathrm{SWAP}^{(\mathcal{H}_1^j, \mathcal{H}_2^j)}}\right)B\left(\bigotimes_{j=1}^2 e^{-i\theta\cdot\mathrm{SWAP}^{(\mathcal{H}_1^j, \mathcal{H}_2^j)}}\right)=B 
	\end{align}
    for all $\theta\in\mathbb{R}$. 
	By taking the derivative at $\theta=0$, we get 
	\begin{align}
		\left[\sum_{j=1}^2 \mathrm{SWAP}^{(\mathcal{H}_1^j, \mathcal{H}_2^j)}, B\right]=0. 
	\end{align}
	This implies that for any $Q, R, S\in\mathcal{P}_n^+$, 
	\begin{align}
		\mathrm{tr}\left(\left[\sum_{j=1}^2 \mathrm{SWAP}^{(\mathcal{H}_1^j, \mathcal{H}_2^j)}, B\right]\left(Q^{(\mathcal{H}_1^1)}\otimes R^{(\mathcal{H}_2^1)}\otimes S^{(\mathcal{H}_1^2)}\right)\right)
		=0. \label{SMeq:SMthm:2design4}
	\end{align}
	By using Eq.~\eqref{SMeq:SMthm:2design3}, we expand the left-hand side of this as follows: 
	\begin{align}
		&\mathrm{tr}\left(\left[\sum_{j=1}^2 \mathrm{SWAP}^{(\mathcal{H}_1^j, \mathcal{H}_2^j)}, B\right]\left(Q^{(\mathcal{H}_1^1)}\otimes R^{(\mathcal{H}_2^1)}\otimes S^{(\mathcal{H}_1^2)}\right)\right) \nonumber\\
		=&\mathrm{tr}\left(\left[\sum_{j=1}^2 \mathrm{SWAP}^{(\mathcal{H}_1^j, \mathcal{H}_2^j)}, \sum_{k=1}^M \sum_{P\in\mathcal{P}_n^+} \beta_P^2 P^{(\mathcal{H}_k^1)}\otimes P^{(\mathcal{H}_k^2)}\right]\left(Q^{(\mathcal{H}_1^1)}\otimes R^{(\mathcal{H}_2^1)}\otimes S^{(\mathcal{H}_1^2)}\right)\right) \nonumber\\
		=&\mathrm{tr}\left(\left[\sum_{j=1}^2 \mathrm{SWAP}^{(\mathcal{H}_1^j, \mathcal{H}_2^j)}, \sum_{k=1}^2 \sum_{P\in\mathcal{P}_n^+} \beta_P^2 P^{(\mathcal{H}_k^1)}\otimes P^{(\mathcal{H}_k^2)}\right]\left(Q^{(\mathcal{H}_1^1)}\otimes R^{(\mathcal{H}_2^1)}\otimes S^{(\mathcal{H}_1^2)}\right)\right) \nonumber\\
		=&\sum_{P\in\mathcal{P}_n^+} \beta_P^2 \mathrm{tr}\left(\left[\left(\left[\mathrm{SWAP}^{(\mathcal{H}_1^1, \mathcal{H}_2^1)}\otimes I^{(\mathcal{H}_1^2)}\otimes I^{(\mathcal{H}_2^2)}, P^{(\mathcal{H}_1^1)}\otimes I^{(\mathcal{H}_2^1)}\otimes P^{(\mathcal{H}_1^2)}\otimes I^{(\mathcal{H}_2^2)}\right]\right.\right.\right. \nonumber\\
		&\left.\left.\left.\ \ \ \ \ \ \ \ \ \ \ \ \ \ \ \ \ \ \ +\left[\mathrm{SWAP}^{(\mathcal{H}_1^1, \mathcal{H}_2^1)}\otimes I^{(\mathcal{H}_1^2)}\otimes I^{(\mathcal{H}_2^2)}, I^{(\mathcal{H}_1^1)}\otimes P^{(\mathcal{H}_2^1)}\otimes I^{(\mathcal{H}_1^2)}\otimes P^{(\mathcal{H}_2^2)}\right]\right.\right.\right. \nonumber\\
		&\left.\left.\left. \ \ \ \ \ \ \ \ \ \ \ \ \ \ \ \ \ \ \ +\left[I^{(\mathcal{H}_1^1)}\otimes I^{(\mathcal{H}_2^1)}\otimes\mathrm{SWAP}^{(\mathcal{H}_1^2, \mathcal{H}_2^2)}, P^{(\mathcal{H}_1^1)}\otimes I^{(\mathcal{H}_2^1)}\otimes P^{(\mathcal{H}_1^2)}\otimes I^{(\mathcal{H}_2^2)}\right]\right.\right.\right. \nonumber\\
		&\left.\left.\left.\ \ \ \ \ \ \ \ \ \ \ \ \ \ \ \ \ \ \ +\left[I^{(\mathcal{H}_1^1)}\otimes I^{(\mathcal{H}_2^1)}\otimes\mathrm{SWAP}^{(\mathcal{H}_1^2, \mathcal{H}_2^2)}, I^{(\mathcal{H}_1^1)}\otimes P^{(\mathcal{H}_2^1)}\otimes I^{(\mathcal{H}_1^2)}\otimes P^{(\mathcal{H}_2^2)}\right]\right)\right.\right. \nonumber\\
		&\left.\left.\ \ \ \ \ \ \ \ \ \ \ \ \ \ \ \ \ \ \ \ \ \ \times\left(Q^{(\mathcal{H}_1^1)}\otimes R^{(\mathcal{H}_2^1)}\otimes S^{(\mathcal{H}_1^2)}\otimes I^{(\mathcal{H}_2^2)}\right)\right]\otimes\bigotimes_{j=1}^2 \bigotimes_{k=3}^M I^{(\mathcal{H}_k^j)}\right) \nonumber\\
		=&2^{2(M-2)n}\sum_{P\in\mathcal{P}_n^+} \beta_P^2 \left(\mathrm{tr}\left(\left[\mathrm{SWAP}^{(\mathcal{H}_1^1, \mathcal{H}_2^1)}, P^{(\mathcal{H}_1^1)}\otimes I^{(\mathcal{H}_2^1)}\right]\left(Q^{(\mathcal{H}_1^1)}\otimes R^{(\mathcal{H}_2^1)}\right)\otimes PS^{(\mathcal{H}_1^2)}\otimes I^{(\mathcal{H}_2^2)}\right)\right. \nonumber\\
		&\left.\ \ \ \ \ \ \ \ \ \ \ \ \ \ \ \ \ \ \ \ \ \ \ \ \ \ \ \ +\mathrm{tr}\left(\left[\mathrm{SWAP}^{(\mathcal{H}_1^1, \mathcal{H}_2^1)}, I^{(\mathcal{H}_1^1)}\otimes P^{(\mathcal{H}_1^2)}\right]\left(Q^{(\mathcal{H}_1^1)}\otimes R^{(\mathcal{H}_2^1)}\right)\otimes S^{(\mathcal{H}_1^2)}\otimes P^{(\mathcal{H}_2^2)}\right)\right. \nonumber\\
		&\left.\ \ \ \ \ \ \ \ \ \ \ \ \ \ \ \ \ \ \ \ \ \ \ \ \ \ \ \ +\mathrm{tr}\left(PQ^{(\mathcal{H}_1^1)}\otimes R^{(\mathcal{H}_1^2)}\otimes\left[\mathrm{SWAP}^{(\mathcal{H}_1^2, \mathcal{H}_2^2)}, P^{(\mathcal{H}_1^2)}\otimes I^{(\mathcal{H}_2^2)}\right] \left(S^{(\mathcal{H}_1^2)}\otimes I^{(\mathcal{H}_2^2)}\right)\right)\right. \nonumber\\
		&\left.\ \ \ \ \ \ \ \ \ \ \ \ \ \ \ \ \ \ \ \ \ \ \ \ \ \ \ \ +\mathrm{tr}\left(Q^{(\mathcal{H}_1^1)}\otimes PR^{(\mathcal{H}_1^2)}\otimes\left[\mathrm{SWAP}^{(\mathcal{H}_1^2, \mathcal{H}_2^2)}, I^{(\mathcal{H}_1^2)}\otimes P^{(\mathcal{H}_2^2)}\right] \left(S^{(\mathcal{H}_1^2)}\otimes I^{(\mathcal{H}_2^2)}\right)\right)\right). \label{SMeq:SMthm:2design5}
	\end{align}
	The four terms in the sum can be calculated as follows: 
	\begin{align}
		\mathrm{tr}([\mathrm{SWAP}, P\otimes I](Q\otimes R)\otimes PS \otimes I)
		=&\mathrm{tr}([\mathrm{SWAP}, P\otimes I](Q\otimes R))\mathrm{tr}(PS)\mathrm{tr}(I) \nonumber\\
		=&\mathrm{tr}(\mathrm{SWAP}([P\otimes I, Q\otimes R]))\cdot 2^n\delta_{P, S}\cdot 2^n \nonumber\\
		=&2^{2n}\delta_{P, S}\mathrm{tr}(\mathrm{SWAP}([P, Q]\otimes R)) \nonumber\\
		=&2^{2n}\delta_{P, S}\mathrm{tr}([P, Q]R), \label{SMeq:SMthm:2design6}\\
		\mathrm{tr}([\mathrm{SWAP}, I\otimes P](Q\otimes R)\otimes S\otimes R)
		=&\mathrm{tr}([\mathrm{SWAP}, I\otimes P](Q\otimes R))\mathrm{tr}(S)\mathrm{tr}(R) \nonumber\\
		=&\mathrm{tr}(\mathrm{SWAP}[I\otimes P, Q\otimes R])\cdot 2^n\delta_{S, I}\cdot 2^n\delta_{R, I} \nonumber\\
		=&2^{2n}\delta_{R, I}\delta_{S, I}\mathrm{tr}(\mathrm{SWAP}[I\otimes P, Q\otimes I]) \nonumber\\
		=&0, \label{SMeq:SMthm:2design7}\\
		\mathrm{tr}(PQ\otimes R\otimes [\mathrm{SWAP}, P\otimes I](S\otimes I))
		=&\mathrm{tr}(PQ)\mathrm{tr}(R)\mathrm{tr}([\mathrm{SWAP}, P\otimes I](S\otimes I))  \nonumber\\
		=&2^n\delta_{P, Q}\cdot 2^n\delta_{R, I}\cdot\mathrm{tr}(\mathrm{SWAP}[P\otimes I, S\otimes I]) \nonumber\\
		=&2^{2n}\delta_{P, Q}\delta_{R, I}\mathrm{tr}(\mathrm{SWAP}([P, S]\otimes I)) \nonumber\\
		=&2^{2n}\delta_{P, Q}\delta_{R, I}\mathrm{tr}([P, S]) \nonumber\\
		=&0, \label{SMeq:SMthm:2design8}\\
		\mathrm{tr}(Q\otimes PR\otimes [\mathrm{SWAP}, I\otimes P](S\otimes I))
		=&\mathrm{tr}(Q)\mathrm{tr}(PR)\mathrm{tr}([\mathrm{SWAP}, I\otimes P](S\otimes I))  \nonumber\\
		=&2^n\delta_{Q, I}\cdot 2^n\delta_{P, R}\cdot\mathrm{tr}(\mathrm{SWAP}[I\otimes P, S\otimes I]) \nonumber\\
		=&0, \label{SMeq:SMthm:2design9}
	\end{align}
	where we used the cyclicity of the trace and the swap trick $\mathrm{tr}(\mathrm{SWAP}(L\otimes M))=LM$. 
	By plugging Eqs.~\eqref{SMeq:SMthm:2design5}, \eqref{SMeq:SMthm:2design6}, \eqref{SMeq:SMthm:2design7}, \eqref{SMeq:SMthm:2design8} and \eqref{SMeq:SMthm:2design9} into Eq.~\eqref{SMeq:SMthm:2design4}, we get 
	\begin{align}
		2^{2(M-1)n}\beta_S^2 \mathrm{tr}([S, Q]R)=0. \label{SMeq:SMthm:2design10}
	\end{align}
	Suppose that $\beta_S\neq0$. 
	Then, Eq.~\eqref{SMeq:SMthm:2design10} means that $\mathrm{tr}([S, Q]R)=0$. 
	Since this holds for all $Q, R\in\mathcal{P}_n^+$, we have $[S, Q]=0$ for all $Q\in\mathcal{P}_n^+$. 
	This implies that $S=I$. 
	We therefore know that $\beta_S\neq 0$ for all $S\in\mathcal{P}_n^+\backslash\{I\}$, or equivalently, $A=\beta_I I$. 
	Since this holds for all $A\in\mathfrak{f}$, we get $\mathfrak{f}\subset\{aI\ |\ a\in\mathbb{R}\})$. 
	Since the connected Lie group $\mathcal{F}$ can be generated by $e^{i\mathfrak{f}}$ by Corollary~2.31 of Ref.~\cite{hall2003lie}, we know that $\mathcal{F}\subset\{e^{ia}I\ |\ a\in\mathbb{R}\}$, which implies that $\mathcal{G}\subset\{e^{ia}I\ |\ a\in\mathbb{R}\}$. 
	We thus get $\mathcal{U}_{N, \mathcal{G}}=\mathcal{U}_N$. 
\end{proof}

\section{Weighted Unitary Designs} \label{SMsec:weighted_unitary_design}

In this appendix, we introduce weighted symmetric unitary designs and show that the condition for a symmetric Clifford group to be a weighted symmetric unitary design is equivalent to the condition for it to be an unweighted symmetric unitary design, which we defined in Definition~\ref{SMdef:symmetric_unweighted_unitary_design}. 
The definition of weighted symmetric unitary designs is as follows.

\begin{definition} \label{SMdef:symmetric_weighted_unitary_design}
    (Weighted symmetric unitary designs.) 
	Let $N, n, t\in\mathbb{N}$, $\mathcal{G}$ be a subgroup of $\mathcal{U}_N$, $\lambda_1,\lambda_2, ..., \lambda_n\in\mathbb{R}$, and $U_1, U_2, ..., U_n\in\mathcal{U}_N$. 
	A finite set of pairs $\{(\lambda_j, U_j)\}_{j=1}^n$ is a weighted $\mathcal{G}$-symmetric unitary $t$-design if $\sum_{j=1}^n \lambda_j=1$ and 
	\begin{align}
		\sum_{j=1}^n \lambda_j \mathcal{E}_{t, U_j}
		=\Phi_{t, \mathcal{U}_{N, \mathcal{G}}}. \label{SMeq:symmetric_weighted_unitary_design_def}
	\end{align}
\end{definition}

We present the equivalence between the conditions for a symmetric Clifford group to be weighted and unweighted symmetric unitary designs in the following theorem.

\begin{theorem} \label{SMthm:weighted_unitary_design}
    (Equivalence between unweighted and weighted symmetric unitary designs for symmetric Clifford groups.) 
	Let $N, t\in\mathbb{N}$, $\mathcal{G}$ be a subgroup of $\mathcal{U}_{N, \mathcal{G}}$. 
	Then, there exists $n\in\mathbb{N}$, $\lambda_1, \lambda_2, ..., \lambda_n\in\mathbb{R}$ and $U_1, U_2, ..., U_n\in\mathcal{C}_{N, \mathcal{G}}$ such that $\{(\lambda_j, U_j)\}_{j=1}^n$ is a weighted $\mathcal{G}$-symmetric unitary $t$-design if and only if $\mathcal{C}_{N, \mathcal{G}}$ is an unweighted $\mathcal{G}$-symmetric unitary $t$-design. 
\end{theorem}

By combining this theorem with Theorems~\ref{thm:main}, \ref{thm:1design}, and \ref{thm:4design} in the main text, we know that the condition for an ensemble of $\mathcal{C}_{N, \mathcal{G}}$ is an unweighted $\mathcal{G}$-symmetric unitary $t$-design.

\begin{proof}
	First, we prove the ``if'' part. 
	We suppose that $\mathcal{C}_{N, \mathcal{G}}$ is an unweighted $\mathcal{G}$-symmetric unitary $t$-design. 
	We define $n:=|\mathcal{C}_{N, \mathcal{G}}/(\mathcal{U}_0 I)|$, $\lambda_j:=1/|\mathcal{C}_{N, \mathcal{G}}/(\mathcal{U}_0 I)|$ for all $j\in\{1, 2, ..., n\}$, and take $U_1$, $U_2$, ..., $U_n$ from all the equivalence classes of $\mathcal{C}_{N, \mathcal{G}}/(\mathcal{U}_0 I)$. 
	Then, we get 
	\begin{align}
		\sum_{j=1}^n \lambda_j\mathcal{E}_{t, U_j}=\Phi_{t, \mathcal{C}_{N, \mathcal{G}}}=\Phi_{t, \mathcal{U}_{N, \mathcal{G}}}.  
	\end{align}
	This means that $\{(\lambda_j, U_j)\}_{j=1}^n$ is a weighted $\mathcal{G}$-symmetric unitary $t$-design.

	Next, we prove the ``only if'' part. 
	We suppose that there exists $n\in\mathbb{N}$, $\lambda_1, \lambda_2, ..., \lambda_n\in\mathbb{R}$ and $U_1, U_2, ..., U_n\in\mathcal{C}_{N, \mathcal{G}}$ such that $\{(\lambda_j, U_j)\}_{j=1}^n$ is a weighted unitary $t$-design. 
	We define a map $\mathcal{D}$ on $\mathcal{L}(\mathcal{H}^{\otimes t})$ by 
	\begin{align}
		\mathcal{D}:=\sum_{j=1}^n \lambda_j\mathcal{E}_{t, U_j}. 
	\end{align}
	Then, we have 
	\begin{align}
		\mathcal{D}=\Phi_{t, \mathcal{U}_{N, \mathcal{G}}}. \label{SMeq:SMthm:weighted_unitary_design1}
	\end{align}
	Since $\mathcal{D}\in\mathfrak{C}_{t, \mathcal{G}}$, by Lemma~\ref{SMlem:uniform_mixture_property1}, we get 
	\begin{align}
		\Phi_{t, \mathcal{C}_{N, \mathcal{G}}}\circ\mathcal{D}=\Phi_{t, \mathcal{C}_{N, \mathcal{G}}}. \label{SMeq:SMthm:weighted_unitary_design2}
	\end{align}
	Since $\mu_{\mathcal{U}_{N, \mathcal{G}}}$ is left-invariant and $\mu_{\mathcal{C}_{N, \mathcal{G}}}$ is normalized, we get 
	\begin{align}
		\Phi_{t, \mathcal{C}_{N, \mathcal{G}}}\circ\Phi_{t, \mathcal{U}_{N, \mathcal{G}}} 
		=&\int_{U'\in\mathcal{C}_{N, \mathcal{G}}} \int_{U\in\mathcal{U}_{N, \mathcal{G}}} \mathcal{E}_{t, U'}\circ\mathcal{E}_{t, U} d\mu_{\mathcal{C}_{N, \mathcal{G}}}(U') d\mu_{\mathcal{U}_{N, \mathcal{G}}}(U) \nonumber\\
		=&\int_{U'\in\mathcal{C}_{N, \mathcal{G}}} \int_{U\in\mathcal{U}_{N, \mathcal{G}}} \mathcal{E}_{t, U'U} d\mu_{\mathcal{C}_{N, \mathcal{G}}}(U') d\mu_{\mathcal{U}_{N, \mathcal{G}}}(U) \nonumber\\
		=&\int_{U'\in\mathcal{C}_{N, \mathcal{G}}} \int_{U\in\mathcal{U}_{N, \mathcal{G}}} \mathcal{E}_{t, U} d\mu_{\mathcal{C}_{N, \mathcal{G}}}(U') d\mu_{\mathcal{U}_{N, \mathcal{G}}}(U) \nonumber\\
		=&\int_{U\in\mathcal{U}_{N, \mathcal{G}}} \mathcal{E}_{t, U} d\mu_{\mathcal{U}_{N, \mathcal{G}}}(U) \nonumber\\
		=&\Phi_{t, \mathcal{U}_{N, \mathcal{G}}}. \label{SMeq:SMthm:weighted_unitary_design3}
	\end{align}
	By Eqs.~\eqref{SMeq:SMthm:weighted_unitary_design1}, \eqref{SMeq:SMthm:weighted_unitary_design2} and \eqref{SMeq:SMthm:weighted_unitary_design3}, we get 
	\begin{align}
		\Phi_{t, \mathcal{C}_{N, \mathcal{G}}}
		=\Phi_{t, \mathcal{C}_{N, \mathcal{G}}}\circ\mathcal{D}
		=\Phi_{t, \mathcal{C}_{N, \mathcal{G}}}\circ\Phi_{t, \mathcal{U}_{N, \mathcal{G}}}
		=\Phi_{t, \mathcal{U}_{N, \mathcal{G}}}. 
	\end{align}
	This means that $\mathcal{C}_{N, \mathcal{G}}$ is an unweighted $\mathcal{G}$-symmetric unitary $t$-design. 
\end{proof}

\section{Technical Lemmas} \label{SMsec:technical}

In this appendix, we present the technical lemmas that we use in the proofs of the theorems in this work.

\subsection{Transformation of Pauli subgroups into the standard form}

For the proofs of Lemma~\ref{SMlem:symmetrization} and Theorem~\ref{SMthm:symmetric_Clifford_expression}, we show that any Pauli subgroup can be transformed into the form of Eq.~\eqref{SMeq:Pauli_subgroup_standard} up to phase via some Clifford conjugation action. 
For that purpose, we prepare two simple properties about Pauli operators. 
First, we prove that any Pauli operator can be transformed into the Z operator on the first qubit up to multiplicity of a constant via some Clifford conjugation action.

\begin{lemma} \label{SMlem:Pauli_to_Z}
	Let $N\in\mathbb{N}$, $P\in\mathcal{P}_N$ and $P\not\in\mathcal{P}_0 I$. 
	Then, there exists some $W\in\mathcal{C}_N$ such that $WPW^\dag=\chi\mathrm{Z}_1$ with $\chi\in\mathcal{P}_0$. 
\end{lemma}

\begin{proof}
	By noting that $\mathrm{S}_j^\dag\mathrm{Y}_j{(\mathrm{S}_j^\dag)}^\dag=\mathrm{X}_j$ and $\mathrm{H}_j\mathrm{X}_j\mathrm{H}_j^\dag=\mathrm{Z}_j$, we can construct $W_1\in\mathcal{C}_N$ with $\mathrm{S}_j^\dag$ and $\mathrm{H}_j$ such that $W_1 PW_1^\dag=\chi\prod_{j=1}^N \mathrm{Z}_j^{\mu_j}$ with $\chi\in\mathcal{P}_0$ and $\mu_j\in\{0, 1\}$. 
	Since $P\not\in\mathcal{P}_0 I$, we can take some $a\in\{1, 2, ..., N\}$ such that $\mu_a=1$. 
	We next define $W_2:=\prod_{j\in\{1, 2, ..., N\}\backslash\{a\}} (\mathrm{CNOT}_{j, a})^{\mu_j}$. 
	Then, $W_2(\prod_{j=1}^N \mathrm{Z}_j^{\mu_j})W_2^\dag=\mathrm{Z}_j$. 
	We finally define $W_3:=\mathrm{SWAP}_{1, j}$ if $j\neq1$ and $W_3:=I$ if $j=1$. 
	We define $W:=W_3 W_2 W_1$. 
	Then, $WPW^\dag=\chi\mathrm{Z}_1$. 
\end{proof}

Next, we prove that any pair of two non-commutative Pauli operators can be simultaneously transformed into the Z and X operators on the first qubit up to multiplicity of a constant via some Clifford conjugation action.

\begin{lemma} \label{SMlem:noncommutative_Pauli_to_XZ}
	Let $N\in\mathbb{N}$, $P, P'\in\mathcal{P}_N$, and $P$ and $P'$ be non-commutative with each other. 
	Then, there exists some $W\in\mathcal{C}_N$ such that $WPW^\dag=\chi \mathrm{Z}_1$ and $WP'W^\dag=\chi' \mathrm{X}_1$ with $\chi, \chi'\in\mathcal{P}_0$. 
\end{lemma}

\begin{proof}
	Since $P$ and $P'$ are non-commutative with each other, $P\not\in\mathcal{P}_0 I$. 
	By Lemma~\ref{SMlem:Pauli_to_Z}, we can take $W_1\in\mathcal{C}_N$ such that $W_1PW_1^\dag=\chi \mathrm{Z}_1$ with $\chi\in\mathcal{P}_0$. 
	Since $W_1 PW_1^\dag$ and $W_1P'W_1^\dag$ are non-commutative, $W_1 P'W_1^\dag$ can be written as $W_1 P'W_1^\dag=\eta \mathrm{X}_1\otimes P''$ or $W_1^\dag P'W_1=\eta \mathrm{Y}_1\otimes P''$ with $\eta\in\mathcal{P}_0$ and $P''\in\mathcal{P}_{N-1}$. 
	We define $W_2:=I$ in the former case and $W_2:=\mathrm{S}_1^\dag$ in the latter case. 
	Then, $(W_2 W_1)P'(W_2 W_1)^\dag=\eta \mathrm{X}_1\otimes P''$. 
	If $P''\not\in\mathcal{P}_0 I$, we can take $W'\in\mathcal{C}_{N-1}$ by Lemma~\ref{SMlem:Pauli_to_Z} such that $W' P'' W'^\dag=\eta' \mathrm{Z}_2$ with $\eta'\in\mathcal{P}_0$. 
	We define $W_3:=W'\cdot\mathrm{CZ}_{1, 2}$ and $\chi':=\eta\eta'$. 
	If $P''\in\mathcal{P}_0 I$, we define $W_3:=I$ and $\chi':=\eta$. 
	We define $W:=W_3 W_2 W_1$. 
	Then, $WPW^\dag=\chi \mathrm{Z}_1$ and $WP'W^\dag=\chi' \mathrm{X}_1$. 
\end{proof}

By using the two lemmas above, we prove that we can transform any Pauli subgroup into the form of Eq.~\eqref{SMeq:Pauli_subgroup_standard} up to multiplicity of a constant via some Clifford conjugation action.

\begin{lemma} \label{SMlem:Pauli_subgroup_equivalence}
	Let $n\in\mathbb{N}$ and $\mathcal{Q}$ be a subgroup of $\mathcal{P}_n$. 
	Then, there exists $W\in\mathcal{C}_n$ and $n_1, n_2, n_3\geq0$ such that $\mathcal{P}_0 W\mathcal{Q}W^\dag=\mathcal{P}_0\{\mathrm{I}, \mathrm{X}, \mathrm{Y}, \mathrm{Z}\}^{\otimes n_1}\otimes\{\mathrm{I}, \mathrm{Z}\}^{\otimes n_2}\otimes\{\mathrm{I}\}^{\otimes n_3}$. 
\end{lemma}

\begin{proof}
	We prove this lemma by mathematical induction about $n$. 
	Since the statement of Lemma~\ref{SMlem:Pauli_subgroup_equivalence} trivially holds for $n=1$, it is sufficient to show that any subgroup $\mathcal{Q}$ of $\mathcal{P}_{k+1}$ satisfies the following two properties for all $k\in\mathbb{N}$:

	\noindent
	i) If $\mathcal{Q}$ has a non-commutative pair of elements, then $\mathcal{P}_0 W\mathcal{Q}W^\dag=\mathcal{P}_0\{\mathrm{I}, \mathrm{X}, \mathrm{Y}, \mathrm{Z}\}\otimes\mathcal{Q}'$ with some $W\in\mathcal{C}_{k+1}$ and subgroup $\mathcal{Q}'$ of $\mathcal{P}_k$.

	\noindent
	ii) If every pair of the elements of $\mathcal{Q}$ is commutative and $\mathcal{Q}\not\subset\mathcal{P}_0 I$, then $\mathcal{P}_0 W\mathcal{Q}W^\dag=\mathcal{P}_0\{\mathrm{I}, \mathrm{Z}\}\otimes\mathcal{Q}'$ with some $W\in\mathcal{C}_{k+1}$ and commutative subgroup $\mathcal{Q}'$ of $\mathcal{P}_k$.

	\noindent
	Under the assumption of these two properties, we can construct the mathematical induction as follows. 
	Suppose that the statement of Lemma~\ref{SMlem:Pauli_subgroup_equivalence} holds for $n=k\in\mathbb{N}$ and take an arbitrary subgroup $\mathcal{Q}$ of $\mathcal{P}_{k+1}$. 
	If $\mathcal{Q}$ has a non-commutative pair of elements, by the property i), $\mathcal{P}_0 W\mathcal{Q}W^\dag=\mathcal{P}_0\{\mathrm{I}, \mathrm{X}, \mathrm{Y}, \mathrm{Z}\}\otimes\mathcal{Q}'$ with some $W\in\mathcal{C}_{k+1}$ and subgroup $\mathcal{Q}'$ of $\mathcal{P}_k$. 
	Since we suppose that the statement of Lemma~\ref{SMlem:Pauli_subgroup_equivalence} holds for $n=k$, $\mathcal{Q}'$ satisfies $\mathcal{P}_0 W'\mathcal{Q}' W'^\dag=\mathcal{P}_0\{\mathrm{I}, \mathrm{X}, \mathrm{Y}, \mathrm{Z}\}^{\otimes k_1}\otimes\{\mathrm{I}, \mathrm{Z}\}^{\otimes k_2}\otimes\{I\}^{\otimes k_3}$ with some $W'\in\mathcal{C}_k$ and $k_1, k_2, k_3\geq0$. 
	We therefore get $\mathcal{P}_0[(I\otimes W')W]\mathcal{Q}[(I\otimes W')W]^\dag=\mathcal{P}_0\{\mathrm{I}, \mathrm{X}, \mathrm{Y}, \mathrm{Z}\}^{\otimes k_1+1}\otimes\{\mathrm{I}, \mathrm{Z}\}^{\otimes k_2}\otimes\{\mathrm{I}\}^{\otimes k_3}$. 
	If every pair of the elements of $\mathcal{Q}$ is commutative and $\mathcal{Q}\not\subset\mathcal{P}_0 I$, by the property  ii), $\mathcal{P}_0 W\mathcal{Q}W^\dag=\mathcal{P}_0\{\mathrm{I}, \mathrm{Z}\}\otimes\mathcal{Q}'$ with some $W\in\mathcal{C}_{k+1}$ and commutative subgroup $\mathcal{Q}'$ of $\mathcal{P}_k$. 
	Since we suppose that the statement of Lemma~\ref{SMlem:Pauli_subgroup_equivalence} holds for $n=k$, $\mathcal{Q}'$ satisfies $\mathcal{P}_0 W'\mathcal{Q}' W'^\dag=\mathcal{P}_0\{\mathrm{I}, \mathrm{Z}\}^{\otimes k_2}\otimes\{\mathrm{I}\}^{\otimes k_3}$ with some $W'\in\mathcal{C}_k$ and $k_2, k_3\geq0$. 
	We therefore get $\mathcal{P}_0[(I\otimes W')W]\mathcal{Q}[(I\otimes W')W]^\dag=\mathcal{P}_0\{\mathrm{I}, \mathrm{Z}\}^{\otimes k_2+1}\otimes\{\mathrm{I}\}^{\otimes k_3}$. 
	If $\mathcal{Q}\subset\mathcal{P}_0 I$, we trivially get $\mathcal{P}_0\mathcal{Q}=\mathcal{P}_0$. 
	In all cases, the statement of Lemma~\ref{SMlem:Pauli_subgroup_equivalence} holds for $n=k+1$, and we can complete the proof by mathematical induction.

	In the following, we prove the properties i) and ii). 
	First, we prove the property i). 
	Since $\mathcal{Q}$ is a finite group, $\mathcal{Q}$ can be expressed as the group $\braket{\{Q_j\}_{j=1}^M}$ generated by some $Q_1, Q_2, ..., Q_M\in\mathcal{Q}$. 
	We take $a, b\in\{1, 2, ..., M\}$ such that $Q_a$ and $Q_b$ are non-commutative with each other. 
	By Lemma~\ref{SMlem:noncommutative_Pauli_to_XZ}, we can take $W\in\mathcal{C}_{k+1}$ such that $W Q_a W^\dag=\chi \mathrm{Z}_1$ and $W Q_b W^\dag=\chi' \mathrm{X}_1$ with some $\chi, \chi'\in\mathcal{P}_0$. 
	For any $j\in\{1, 2, ..., M\}$, $W Q_j W^\dag$ can be written as 
	\begin{align}
		W Q_j W^\dag=\mathrm{Z}_1^{\mu_j}\mathrm{X}_1^{\nu_j}\otimes Q'_j
	\end{align}
	with some $\mu_j, \nu_j\in\{0, 1\}$ and $Q'_j\in\mathcal{P}_k$. 
	We therefore get 
	\begin{align}
		\mathcal{P}_0 W\mathcal{Q}W^\dag
		=&\mathcal{P}_0\braket{\{W Q_j W^\dag\}_{j=1}^M} \nonumber\\
		=&\mathcal{P}_0\braket{\chi \mathrm{Z}_1, \chi' \mathrm{X}_1, \{W Q_j W^\dag\}_{j\in\{1, 2, ..., M\}\backslash\{a, b\}}} \nonumber\\
		=&\mathcal{P}_0\braket{\mathrm{Z}_1, \mathrm{X}_1, \{Q'_j\}_{j\in\{1, 2, ..., M\}\backslash\{a, b\}}} \nonumber\\
		=&\mathcal{P}_0\braket{\mathrm{Z}, \mathrm{X}}\otimes\mathcal{Q}' \nonumber\\
		=&\mathcal{P}_0\{\mathrm{I}, \mathrm{X}, \mathrm{Y}, \mathrm{Z}\}\otimes\mathcal{Q}', 
	\end{align}
	where $\mathcal{Q}':=\braket{\{Q'_j\}_{j\in\{1, 2, ..., M\}\backslash\{a, b\}}}$.

	Next, we prove the property ii). 
	As in the proof of the property i), $\mathcal{Q}$ can be expressed as $\mathcal{Q}=\braket{\{Q_j\}_{j=1}^M}$ with some $Q_1, Q_2, ..., Q_M\in\mathcal{Q}$. 
	Since $\mathcal{Q}\not\subset\mathcal{P}_0 I$, we can take $a\in\{1, 2, ..., M\}$ such that $Q_a\not\in\mathcal{P}_0 I$. 
	By Lemma~\ref{SMlem:Pauli_to_Z}, we can take $W\in\mathcal{C}_{k+1}$ such that $W Q_a W^\dag=\chi \mathrm{Z}_1$ with some $\chi\in\mathcal{P}_0$. 
	Since every pair of elements of $\mathcal{Q}$ is commutative, $[W Q_j W^\dag, \mathrm{Z}_1]=\chi^{-1}[W Q_j W^\dag, W Q_a W^\dag]=\chi^{-1} W[Q_j, Q_a]W^\dag=0$ for all $j\in\{1, 2, ..., M\}$. 
	This implies that for any $j\in\{1, 2, ..., M\}$, $W Q_j W^\dag$ can be written as 
	\begin{align}
		W Q_j W^\dag=\mathrm{Z}_1^{\mu_j}\otimes Q'_j
	\end{align}
	with some $\mu_j\in\{0, 1\}$ and $Q'_j\in\mathcal{P}_k$. 
	We therefore get 
	\begin{align}
		\mathcal{P}_0 W\mathcal{Q}W^\dag
		=&\mathcal{P}_0\braket{\{W Q_j W^\dag\}_{j=1}^M} \nonumber\\
		=&\mathcal{P}_0\braket{\chi \mathrm{Z}_1, \{W Q_j W^\dag\}_{j\in\{1, 2, ..., M\}\backslash\{a\}}} \nonumber\\
		=&\mathcal{P}_0\braket{\mathrm{Z}_1, \{Q'_j\}_{j\in\{1, 2, ..., M\}\backslash\{a\}}} \nonumber\\
		=&\mathcal{P}_0\braket{\mathrm{Z}}\otimes\mathcal{Q}' \nonumber\\
		=&\mathcal{P}_0\{\mathrm{I}, \mathrm{Z}\}\otimes\mathcal{Q}', 
	\end{align}
	where $\mathcal{Q}':=\braket{\{Q'_j\}_{j\in\{1, 2, ..., M\}\backslash\{a\}}}$. 
	Since every pair of elements of $\mathcal{Q}$ is commutative, we have for any $j, j'\in\{1, 2, ..., M\}\backslash\{a\}$, 
	\begin{align}
		[Q'_j, Q'_{j'}]
		=&[\mathrm{Z}_1^{\mu_j}W Q_j W^\dag, \mathrm{Z}_1^{\mu_{j'}}W Q_{j'} W^\dag] \nonumber\\
		=&[(\chi^{-1} W Q_a W^\dag)^{\mu_j}W Q_j W^\dag, (\chi^{-1} W Q_a W^\dag)^{\mu_{j'}}W Q_{j'} W^\dag] \nonumber\\
		=&\chi^{-\mu_j-\mu_{j'}}W[Q_a^{\mu_j}Q_j, Q_a^{\mu_{j'}}Q_{j'}]W^\dag \nonumber\\
		=&0. 
	\end{align}
	This means that every pair of elements of $\mathcal{Q}'$ is commutative. 
\end{proof}

\subsection{Property of $3$-bit sequences}

For the proof of Lemma~\ref{SMlem:CPauli^BC_mixture}, we prove the following property of 3-bit sequences.

\begin{lemma} \label{SMlem:3_bit_property}
	Let $\{x_j\}, \{x'_j\}, \{y_j\}, \{y'_j\}\in\{0, 1\}^3$ satisfy 
	\begin{align}
		\sum_{j=1}^3 x_j=\sum_{j=1}^3 y_j,\ 
		\sum_{j=1}^3 x'_j=\sum_{j=1}^3 y'_j,\ 
		\sum_{j=1}^3 x_j x'_j\equiv\sum_{j=1}^3 y_j y'_j\ (\mathrm{mod}\ 2), 
	\end{align}
	and $p, q\in\{1, 2, 3\}$ and $\sigma\in\mathfrak{S}_3$ satisfy $x_{\sigma(j)}\neq x_{\sigma(p)}$ for all $j\in\{1, 2, 3\}\backslash\{p\}$, $y_j\neq y_q$ for $j\in\{1, 2, 3\}\backslash\{q\}$, and $x'_{\sigma(j)}=y'_j$ for all $j\in\{1, 2, 3\}$. 
	Then, $x'_{\sigma(p)}=x'_{\sigma(q)}$. 
\end{lemma}

\begin{proof}
	We define $z_j:=x_{\sigma(j)}$ and $z'_j:=x'_{\sigma(j)}$. 
	Then, 
	\begin{align}
		&\sum_{j=1}^3 z_j
		=\sum_{j=1}^3 x_{\sigma(j)}
		=\sum_{j=1}^3 x_j
		=\sum_{j=1}^3 y_j, \label{SMeq:SMlem:3_bit_property1}\\ 
		&\sum_{j=1}^3 z'_j
		=\sum_{j=1}^3 x'_{\sigma(j)}
		=\sum_{j=1}^3 x'_j
		=\sum_{j=1}^3 y'_j, \label{SMeq:SMlem:3_bit_property2}\\ 
		&\sum_{j=1}^3 z_j z'_j
		=\sum_{j=1}^3 x_{\sigma(j)} x'_{\sigma(j)}
		=\sum_{j=1}^3 x_j x'_j
		\equiv\sum_{j=1}^3 y_j y'_j\ (\mathrm{mod}\ 2). \label{SMeq:SMlem:3_bit_property3}
	\end{align}
	Since $\{z_j\}$ and $\{y_j\}$ satisfy Eq.~\eqref{SMeq:SMlem:3_bit_property1}, $z_j\neq z_p$ for all $j\in\{1, 2, 3\}\backslash\{p\}$, and $y_j\neq y_q$ for all $j\in\{1, 2, 3\}\backslash\{q\}$, we can take $w\in\{0, 1\}$ such that for any $j\in\{1, 2, 3\}$, 
	\begin{align}
		z_j\equiv w+\delta_{j, p},\ y_j\equiv w+\delta_{j, q}\ (\mathrm{mod}\ 2). \label{SMeq:SMlem:3_bit_property4}
	\end{align} 
	By Eqs.~\eqref{SMeq:SMlem:3_bit_property2}, \eqref{SMeq:SMlem:3_bit_property3} and \eqref{SMeq:SMlem:3_bit_property4}, we get 
	\begin{align}
		z'_p-y'_q
		=&\sum_{j=1}^3 \delta_{j, p}z'_j-\sum_{j=1}^3 \delta_{j, q}y_j \nonumber\\
		\equiv&\sum_{j=1}^3 (z_j-w)z'_j-\sum_{j=1}^3 (y_j-w)y'_j\ (\mathrm{mod}\ 2) \nonumber\\
		=&\left(\sum_{j=1}^3 z_j z'_j-\sum_{j=1}^3 y_j y'_j\right)-w\left(\sum_{j=1}^3 z'_j-\sum_{j=1}^3 y_j\right) \nonumber\\
		=&0. 
	\end{align}
	Since $z'_p, y'_q\in\{0, 1\}$, this implies that $z'_p=y'_q$. 
	By combining this, the definition of $\{z'_j\}$ and the assumption that $x'_{\sigma(j)}=y'_j$ for all $j\in\{1, 2, 3\}$, we get 
	\begin{align}
		x'_{\sigma(p)}
		=z'_p
		=y'_q
		=x'_{\sigma(q)}. 
	\end{align}
\end{proof}

\subsection{Property of $t$-fold mixture maps}

In order to construct $t$-fold $\mathcal{G}$-symmetric Clifford conjugation mixture maps $\mathcal{D}$ in Lemma~\ref{SMlem:symmetrization} and $\mathcal{D}''$ in Theorem~\ref{SMthm:1design}, we prove that the set $\mathfrak{C}_{t, \mathcal{G}}$ of all $t$-fold $\mathcal{G}$-symmetric Clifford mixture maps is closed under composition.

\begin{lemma} \label{SMeq:symmetric_Clifford_mixture_composition}
	Let $N, t\in\mathbb{N}$, $\mathcal{G}$ be a subgroup of $\mathcal{U}_N$, $\mathfrak{C}_{t, \mathcal{G}}$ be the set of all $t$-fold $\mathcal{G}$-symmetric Clifford conjugation mixture maps defined by Eq.~\eqref{SMeq:symmetric_Clifford_mixture_def} and $\mathcal{D}, \mathcal{D}'\in\mathfrak{C}_{t, \mathcal{G}}$. 
	Then, $\mathcal{D}\circ\mathcal{D}'\in\mathfrak{C}_{t, \mathcal{G}}$. 
\end{lemma}

\begin{proof}
	Since $\mathcal{D}, \mathcal{D}'\in\mathfrak{C}_{t, \mathcal{G}}$, $\mathcal{D}$ and $\mathcal{D}'$ can be written as 
	\begin{align}
		\mathcal{D}=\sum_{j=1}^n \lambda_j \mathcal{E}_{t, U_j},\ 
		\mathcal{D}'=\sum_{j'=1}^{n'} \lambda'_{j'} \mathcal{E}_{t, U'_{j'}} 
	\end{align}
	with some $n, n'\in\mathbb{N}$, $U_1, U_2, ..., U_n, U'_1, U'_2, ..., U'_{n'}\in\mathcal{C}_{N, \mathcal{G}}$ and $\lambda_1, \lambda_2, ..., \lambda_n, \lambda'_1, \lambda'_2, ..., \lambda'_{n'}\in\mathbb{R}$ satisfying $\sum_{j=1}^n \lambda_j=\sum_{j'=1}^{n'} \lambda'_{j'}=1$. 
	Then we get 
	\begin{align}
		\mathcal{D}\circ\mathcal{D}' 
		=\sum_{j=1}^n\sum_{j'=1}^{n'} \lambda_j\lambda'_{j'}\mathcal{E}_{t, U_j}\circ\mathcal{E}_{t, U'_{j'}} 
		=\sum_{j=1}^n \sum_{j'=1}^{n'} \lambda_j\lambda'_{j'}\mathcal{E}_{t, U_jU'_{j'}}, 
	\end{align}
	$U_jU'_{j'}\in\mathcal{C}_{N, \mathcal{G}}$, and the coefficients satisfy 
	\begin{align}
		\sum_{j=1}^n \sum_{j'=1}^{n'} \lambda_j\lambda'_{j'} 
		=\left(\sum_{j=1}^n \lambda_j\right)\left(\sum_{j'=1}^{n'} \lambda'_{j'}\right) 
		=1. 
	\end{align}
	We therefore get $\mathcal{D}\circ\mathcal{D}'\in\mathfrak{C}_{t, \mathcal{G}}$. 
\end{proof}

\subsection{Bijections induced by Clifford operators}

We use the following lemma in the proofs of Lemma~\ref{SMlem:Clifford_invariant_construction} and Theorem~\ref{thm:4design} in the main text.

\begin{lemma} \label{SMlem:Clifford_action_on_Pauli}
	Let $N\in\mathbb{N}$ and $U\in\mathcal{C}_N$. 
	Then, there exist some function $s_U:\mathcal{P}_N^+\to\{\pm 1\}$ and some bijection $h_U$ on $\mathcal{P}_N^+$ such that 
	\begin{align}
		UPU^\dag=s_U(P) h_U(P)
	\end{align}
	for all $P\in\mathcal{P}_N^+$. 
\end{lemma}

\begin{proof}
	$U\in\mathcal{C}_N$ implies that for any $P\in\mathcal{P}_N^+$, $UPU^\dag\in\mathcal{P}_N$, i.e., 
	\begin{align}
		UPU^\dag=s' P' 
	\end{align}
	with some $s'\in\{\pm1, \pm i\}$ and $P'\in\mathcal{P}_N^+$. 
	Since $P$ and $P'$ are hermitian, we have 
	\begin{align}
        s'^* P'
        =(s' P')^\dag
        =(UPU^\dag)^\dag
        =UPU^\dag
        =s' P'. 
	\end{align} 
	Thus we get $s'^*=s'$, and $s'\in\{\pm 1\}$. 
	We define $s_U(P):=s'$ and $h_U(P):=P'$. 
	For any $P_1, P_2\in\mathcal{P}_N^+$ satisfying $P_1\neq P_2$, we get 
	\begin{align}
        \mathrm{tr}(h_U(P_1)h_U(P_2))
        =\mathrm{tr}(s_U(P_1)UP_1 U^\dag\cdot s_U(P_2)UP_2 U^\dag)
        =s_U(P_1)s_U(P_2)\mathrm{tr}(P_1P_2)=0, 
	\end{align}
	and thus we get $h_U(P_1)\neq h_U(P_2)$. 
	This implies that $h_U$ is injective. 
	By noting that $\mathcal{P}_N^+$ is a finite set, we know that $h_U$ is bijective. 
\end{proof}

\subsection{Connectedness of symmetric unitary groups}

For the proof of Lemma~\ref{SMlem:3_design_discreteness}, we prove the connectedness of $\mathcal{U}_{N, \mathcal{G}}$ for a general subgroup $\mathcal{G}$ of $\mathcal{U}(\mathcal{H})$. 
Here we give a detailed explanation that an operator that is commutative with all representations of a group should be in the form of Eq.~(2.26) of Ref.~\cite{bartlett2007reference}, or Eq.~\eqref{SMeq:SMlem:symmetric_unitary_connectedness9}.

\begin{lemma} \label{SMlem:symmetric_unitary_connectedness}
	Let $N\in\mathbb{N}$ and $\mathcal{G}$ be a subgroup of $\mathcal{U}_N$. 
	Then, $\mathcal{U}_{N, \mathcal{G}}$ is connected. 
\end{lemma}

\begin{proof}
	We consider the regular representation $\rho(G)$ of $\mathcal{G}$, i.e., $\rho(G)=G$ for all $G\in\mathcal{G}$. 
	Since $\rho(G)$ is a unitary representation, $\rho(G)$ is completely reducible. 
	Thus there exist Hilbert spaces $\{\mathcal{H}_\lambda\}_\lambda$, $\{\mathcal{I}_\lambda\}_\lambda$ and $\{\mathcal{J}_\lambda\}_\lambda$ satisfying 
	\begin{align}
		\mathcal{H}=\bigoplus_\lambda \mathcal{H}_\lambda,\ 
		\mathcal{H}_\lambda=\mathcal{I}_\lambda\otimes\mathcal{J}_\lambda
	\end{align}
	and $\rho(G)$ is decomposed in the form of 
	\begin{align}
		\rho(G)=\bigoplus_\lambda \rho_\lambda(G)^{(\mathcal{I}_\lambda)}\otimes I^{(\mathcal{J}_\lambda)} \label{SMeq:SMlem:symmetric_unitary_connectedness1}
	\end{align}
	with irreducible representations $\rho_\lambda(G)$ of $\mathcal{G}$ on $\mathcal{I}_\lambda$ such that $\rho_{\lambda_1}(G)$ and $\rho_{\lambda_2}(G)$ are inequivalent if $\lambda_1\neq\lambda_2$. 
	For the proof of this lemma, it is sufficient to prove 
	\begin{align}
		\mathcal{U}_{N, \mathcal{G}}=\left\{\bigoplus_\lambda I^{(\mathcal{I}_\lambda)}\otimes U_\lambda^{(\mathcal{J}_\lambda)}\ |\ U_\lambda\in\mathcal{U}(\mathcal{J}_\lambda)\right\}. \label{SMeq:SMlem:symmetric_unitary_connectedness2}
	\end{align}
	By using this relation, the connectedness of $\mathcal{U}_{N, \mathcal{G}}$ follows from the connectedness of $\mathcal{U}(\mathcal{J}_\lambda)$ for all $\lambda$.

	Since $\mathcal{U}_{N, \mathcal{G}}\supset\left\{\bigoplus_\lambda I^{(\mathcal{I}_\lambda)}\otimes U_\lambda^{(\mathcal{J}_\lambda)}\ |\ U_\lambda\in\mathcal{U}(\mathcal{J}_\lambda)\right\}$ is trivial, we are going to prove the converse inclusion relation. 
	We take arbitrary $U\in\mathcal{U}_{N, \mathcal{G}}$. 
	Eq.~\eqref{SMeq:SMlem:symmetric_unitary_connectedness1} can equivalently be expressed as 
	\begin{align}
		\rho(G)=\sum_\lambda \Gamma_\lambda \left(\rho_\lambda(G)\otimes I\right)\Gamma_\lambda^\dag \label{SMeq:SMlem:symmetric_unitary_connectedness3}
	\end{align}
	with isometries $\Gamma_\lambda$ from $\mathcal{H}_\lambda$ to $\mathcal{H}$. 
	Since $U$ commutes with $\rho(G)$ for all $G\in\mathcal{G}$, for any $\mu$ and $\nu$, we have 
	\begin{align}
		\Gamma_\mu^\dag[U, \rho(G)]\Gamma_\nu=0. \label{SMeq:SMlem:symmetric_unitary_connectedness4}
	\end{align}
	By plugging Eq.~\eqref{SMeq:SMlem:symmetric_unitary_connectedness3} into Eq.~\eqref{SMeq:SMlem:symmetric_unitary_connectedness4}, we get 	
	\begin{align}
		\sum_\lambda \Gamma_\mu^\dag \left[U\Gamma_\lambda(\rho_\lambda(G)\otimes I)\Gamma_\lambda^\dag-\Gamma_\lambda(\rho_\lambda(G)\otimes I)\Gamma_\lambda^\dag U\right]\Gamma_\nu=0.
	\end{align}
	By noting that $\Gamma_{\lambda_1}^\dag\Gamma_{\lambda_2}=I$ if $\lambda_1=\lambda_2$ and $\Gamma_{\lambda_1}^\dag\Gamma_{\lambda_2}=0$ if $\lambda_1\neq\lambda_2$, this implies that 
	\begin{align}
		\Gamma_\mu U\Gamma_\nu(\rho_\nu(G)\otimes I)-(\rho_\mu(G)\otimes I)\Gamma_\mu^\dag U\Gamma_\nu^\dag=0. \label{SMeq:SMlem:symmetric_unitary_connectedness5}
	\end{align}
	For each $\mu$ and $\nu$, we take a basis $\{E_{\mu, \nu, l}\}_l$ of $\mathcal{L}(\mathcal{I}_\nu\to\mathcal{I}_\mu)$. 
	Then, $\Gamma_\mu^\dag U\Gamma_\nu$ can be written as 
	\begin{align}
		\Gamma_\mu^\dag U\Gamma_\nu=\sum_l U_{\mu, \nu, l}\otimes E_{\mu, \nu, l} \label{SMeq:SMlem:symmetric_unitary_connectedness6}
	\end{align}
	with some $U_{\mu, \nu, l}\in\mathcal{L}(\mathcal{\mathcal{J}_\nu}\to\mathcal{J}_\mu)$. 
	By plugging Eq.~\eqref{SMeq:SMlem:symmetric_unitary_connectedness6} into Eq.~\eqref{SMeq:SMlem:symmetric_unitary_connectedness5}, we get 	
	\begin{align}
		\sum_l (U_{\mu, \nu, l}\rho_\nu(G)-\rho_\mu(G)U_{\mu, \nu, l})\otimes E_{\mu, \nu, l}=0. 
	\end{align}
	This implies that $U_{\mu, \nu, l}\rho_\nu(G)-\rho_\mu(G)U_{\mu, \nu, l}=0$. 
	Since this holds for all $G\in\mathcal{G}$, by Schur's lemma (Propositions~5.3.3 and 5.3.4 of Ref.~\cite{raczka1986theory}), we get 
	\begin{align}
		U_{\mu, \nu, l}=
		\left\{
		\begin{array}{ll}
			u_{\mu, l}I\ &(\textrm{if}\ \mu=\nu)\\
			0\ &(\textrm{if}\ \mu\neq\nu). 
		\end{array}
		\right. \label{SMeq:SMlem:symmetric_unitary_connectedness7}
	\end{align}
	Since $\Gamma_\lambda\Gamma_\lambda^\dag$ is the projection onto $\mathcal{H}_\lambda$ and $\mathcal{H}=\bigoplus_\lambda \mathcal{H}_\lambda$, $U$ can be written as 
	\begin{align}
		U=\left(\sum_\mu \Gamma_\mu \Gamma_\mu^\dag\right)U\left(\sum_\nu \Gamma_\nu \Gamma_\nu^\dag\right)
		=\sum_{\mu, \nu} \Gamma_\mu(\Gamma_\mu^\dag U\Gamma_\nu)\Gamma_\nu^\dag. \label{SMeq:SMlem:symmetric_unitary_connectedness8}
	\end{align}
	By plugging Eqs.~\eqref{SMeq:SMlem:symmetric_unitary_connectedness6} and \eqref{SMeq:SMlem:symmetric_unitary_connectedness7} into Eq.~\eqref{SMeq:SMlem:symmetric_unitary_connectedness8}, we get 	
	\begin{align}
		U
		=\sum_{\mu, \nu, l} \Gamma_\mu(U_{\mu, \nu, l}\otimes E_{\mu, \mu, l})\Gamma_\mu^\dag
		=\sum_{\mu, l} \Gamma_\mu(u_{\mu, l}I\otimes E_{\mu, \mu, l})\Gamma_\mu^\dag
		=\sum_\mu \Gamma_\mu(I\otimes U_\mu)\Gamma_\mu^\dag, 
	\end{align}
	where $U_\mu:=\sum_l u_{\mu, l}E_{\mu, \mu, l}$. 
	This can equivalently be expressed as 
	\begin{align}
		U=\bigoplus_\lambda I^{(\mathcal{I}_\lambda)}\otimes U_\lambda^{(\mathcal{J}_\lambda)} \label{SMeq:SMlem:symmetric_unitary_connectedness9}
	\end{align}
	with some $U_\lambda\in\mathcal{L}(\mathcal{J}_\lambda)$. 
	For any $\lambda$, $U_\lambda\in\mathcal{U}(\mathcal{J}_\lambda)$ follows from $U\in\mathcal{U}(\mathcal{H})$. 
	We therefore get Eq.~\eqref{SMeq:SMlem:symmetric_unitary_connectedness2}. 
\end{proof}

\end{document}